\documentclass[11pt]{article}

\usepackage{amsmath}
\usepackage{amsthm}
\usepackage{amssymb}
\usepackage[usenames]{color}
\usepackage{geometry}
\usepackage{enumerate}
\usepackage{multicol}
\usepackage[ruled,vlined,linesnumbered]{algorithm2e}
\DontPrintSemicolon
\SetAlgorithmName{Listing}{Listing}{Listing}
\usepackage{ifthen}
\usepackage{fullpage}
\usepackage{paralist}
\usepackage{mathtools}
\usepackage[colorlinks,linkcolor=blue,filecolor=blue,citecolor=blue,urlcolor=blue,pdfstartview=FitH]{hyperref}
\usepackage{xspace}

\theoremstyle{plain}
\newtheorem{theorem}{Theorem}[section]
\newtheorem{lemma}[theorem]{Lemma}
\newtheorem{condition}[theorem]{Condition}
\newtheorem{corollary}[theorem]{Corollary}
\newtheorem{definition}[theorem]{Definition}

\newcommand{\set}[1]{\left\{#1\right\}}
\newcommand{\st}{\medspace | \medspace}
\newcommand{\posreals}{\mathbb{R}^{+}_0}

\newcommand{\reals}{\mathbb{R}}

\newcommand{\R}{\mathbb{R}}
\newcommand{\nat}{\mathbb{N}}
\newcommand{\Nb}[1]{\N_{#1}}

\newcommand{\coloneq}{\mathrel{\mathop:}=}
\newcommand{\BO}{\mathcal{O}}

\newcommand{\localskew}{\mathcal{S}}
\newcommand{\stableskew}{{\mathcal{S}}^{\infty}}
\newcommand{\stabtime}{\mathcal{T}_S}

\newcommand{\N}[1][]{\ifthenelse{\equal{#1}{}}
{\ensuremath{N}}
{\ensuremath{N_{#1}}}} % neighbor set
\DeclareMathOperator{\dist}{dist}

\newcommand{\err}[1]{\epsilon_{#1}}
\newcommand{\eps}{\varepsilon}
\newcommand{\drift}{\ensuremath{\rho}} % hw clock rate is 1-\drift .. 1+\drift
\newcommand{\delay}{\ensuremath{\mathcal{T}}} % edge propagation delay

\newcommand{\tup}{\ensuremath{\tau}}
\newcommand{\rate}{\ensuremath{\mathit{mult}}}

\newcommand{\lc}[1]{\ensuremath{L_{#1}}} % locigal clock value
\newcommand{\hc}[1]{\ensuremath{H_{#1}}} % hw clock value
\newcommand{\lce}[2]{\ensuremath{\tilde{L}_{#1}^{#2}}} % logical clock estimate

\newcommand{\dH}[1]{{\frac{d}{dt}H_{{#1}}}}
\newcommand{\dL}[1]{{\frac{d}{dt}L_{{#1}}}}

\newcommand{\sr}[2]{\stackrel{\eqref{#1}}{#2}}

\newcommand{\SC}{\ensuremath{\mathbf{SC}}\xspace}
\newcommand{\FC}{\ensuremath{\mathbf{FC}}\xspace}
\newcommand{\MC}{\ensuremath{\mathbf{MC}}\xspace}

\newcommand{\hide}[1]{ }

\newcommand{\Aopt}{\ensuremath{\mathcal{A}^{\mathrm{OPT}}}}

\newcommand{\mc}[1]{M_{#1}}

\newcommand{\G}{{\cal G}}
\newcommand{\Gu}{\hat{\G}} 

\newcommand{\Ge}{\tilde{\G}}

\newcommand{\Tset}{\mathbb{T}}

\newcommand{\true}{\ensuremath{\mathbf{true}}}
\newcommand{\false}{\ensuremath{\mathbf{false}}}

\newcommand{\insconst}{{\cal B}}
\newcommand{\instime}{\mathcal{I}}

\SetKwBlock{When}{when}{end}
\SetKwBlock{Proc}{procedure}{end}

\newcommand{\stabint}[1]{\Theta_{#1}}
\newcommand{\barstabint}[1]{\bar{\Theta}_{#1}}

\renewcommand{\phi}{\varphi}

\begin{document}
\title{Optimal Gradient Clock Synchronization in Dynamic Networks}

\author{Fabian Kuhn$^1$ ~~ Christoph Lenzen$^2$ ~~ Thomas Locher$^3$ ~~ Rotem Oshman$^4$\\
  \small $^1$University of Freiburg, Germany\\
  \small kuhn@cs.uni-freiburg.de\\
  \small $^2$MPI Saarbruecken, Germany\\
  \small clenzen@mpi-inf.mpg.de\\
  \small $^3$ABB Research, Switzerland\\
  \small thomas.locher@ch.abb.com\\
  \small $^4$Tel Aviv University, Israel\\
  \small roshman@tau.ac.il
}
\date{}

\maketitle
\thispagestyle{empty}

\begin{abstract}
  We study the problem of clock synchronization in highly dynamic
  networks, where communication links can appear or disappear at any
  time. The nodes in the network are equipped with hardware clocks,
  but the rate of the hardware clocks can vary arbitrarily within
  specific bounds, and the estimates that nodes can obtain about the
  clock values of other nodes are inherently inaccurate. Our goal in
  this setting is to output a logical clock at each node such that the
  logical clocks of any two nodes are not too far apart, and nodes
  that remain close to each other in the network for a long time are
  better synchronized than distant nodes. This property is called
  \emph{gradient clock synchronization}.

  Gradient clock synchronization has been widely studied in the static
  setting, where the network topology does not change.  We show that
  the asymptotically optimal bounds obtained for the static case also
  apply to our highly dynamic setting: if two nodes remain at distance
  $d$ from each other for sufficiently long, it is possible to upper
  bound the difference between their clock values by
  $\BO(d \log (D / d))$, where $D$ is the diameter of the
  network. This is known to be optimal even for static networks.
  Furthermore, we show that our algorithm has optimal
  \emph{stabilization time}: when a path of length $d$ appears between
  two nodes, the time required until the clock skew between the two
  nodes is reduced to $\BO(d \log (D / d))$ is $\BO(D)$, which we
  prove to be optimal. Finally, the techniques employed for the more
  intricate analysis of the algorithm for dynamic graphs provide
  additional insights that are also of interest for the static
  setting. In particular, we establish self-stabilization of the
  gradient property within $\BO(D)$ time.
\end{abstract}

\section{Introduction}

A core algorithmic problem in distributed computing is to establish coordination
among the participants of a distributed system, which is often achieved through a
common notion of time. Typically, every node in a network has its own local
hardware clock, which can be used for this purpose; however, hardware clocks of
different nodes run at slightly different rates, and the rates can change over
time. This \emph{clock drift} causes clocks to drift out of synch, requiring
periodic communication to restore synchronization. However, communication is
typically subject to delay, and although an upper bound on the delay may be
known, specific message delays are unpredictable. Consequently, estimates for the
current local time at other nodes are inherently inaccurate.

A distributed clock synchronization algorithm computes at each node a
\emph{logical clock}, and the goal is to synchronize these clocks as tightly as
possible. Traditionally, distributed clock synchronization algorithms focus on
minimizing the \emph{clock skew} between the logical clocks of any two nodes in
the network. The clock skew between two clocks is simply the difference between
the two clock values. The maximum clock skew that may occur in the worst case
between any two nodes at any time is called the \emph{global skew} of a clock
synchronization algorithm. A well-known result states that no algorithm can
guarantee a global skew better than $\Omega(D)$, where $D$ denotes the diameter
of the network~\cite{biaz01}. However, in many cases it is more important to
tightly synchronize the logical clocks of nearby nodes in the network than it is
to minimize the global skew. For example, if a time division multiple access
(TDMA) protocol is used to coordinate access to a shared communication medium in
a wireless sensor network, it suffices to synchronize the clocks of nodes that
interfere with each other when transmitting. The problem of providing better
guarantees on the synchronization quality between nodes that are closer is
called \emph{gradient clock synchronization}. The problem was introduced in a
seminal paper by Fan and Lynch~\cite{fan04}, where the authors show that a clock
skew of $\Omega(\log D/\log\log D)$ cannot be prevented between immediate
neighbors in the network. The largest possible clock skew that may occur between
the logical clocks of any two adjacent nodes at any time is called the
\emph{local skew} of a clock synchronization algorithm. For static networks, it
has been proved that the best possible local skew that an algorithm can achieve
is bounded by $\Theta(\log D)$~\cite{lenzen08,lenzen10}.

While tight bounds have been shown for the static model, the dynamic case has not
been as well understood. A dynamic network arises in many natural contexts: for
example, when nodes are mobile, or when communication links are unreliable and
may fail and recover. The dynamic network model we consider in this article is
general: it allows communication links to appear and disappear arbitrarily,
subject only to a global connectivity constraint (which is required to
maintain a bounded global skew). Hence the model is suitable for modeling
various types of dynamic networks which remain connected over time.

In a dynamic network the distances between nodes change over time as
communication links appear and disappear. Consequently, we divide the
synchronization guarantee into two parts: a \emph{global skew guarantee} bounds
the skew between any two nodes in the network at any time, and a \emph{dynamic
gradient skew guarantee} that bounds the skew between two nodes as a function of
the distance between them and how long they remain at that distance.

In~\cite{kuhn09}, three of the authors showed that a clock
synchronization algorithm cannot react immediately to the formation of new links,
and that a certain \emph{stabilization time} is required before the clocks of
newly-adjacent nodes can be brought into synch. The stabilization time is
inversely related to the synchronization guarantee: the tighter the
synchronization required in stable state, the longer the time to reach that
state. Intuitively, this is because when strict synchronization guarantees are
imposed, the algorithm cannot change clock values quickly without violating the
guarantee, and hence it takes longer to react. The algorithm given
in~\cite{kuhn09} achieves the optimal trade-off between skew bound and
stabilization time; however, its local skew bound is $\BO(\sqrt{D})$, which is far from optimal.

In this article, we propose an algorithm, referred to as $\Aopt$, that achieves
the same asymptotically optimal skew bounds as in the static model: if two nodes
remain at distance $d$ for sufficiently long, the skew between them is reduced
to $\BO(d \log (D / d))$, where $D$ is the dynamic diameter of the network
(corresponding roughly to the time it takes for information to propagate from
one end of the network to the other). The stabilization time of the algorithm,
that is, the time to reach this guarantee, is $\BO(D)$.

\section{Related Work}
\label{sec:relwork}

The fundamental problem of synchronizing clocks in distributed systems has been
studied extensively and many results have been published for various models over
the course of the last approximately 30 years (see,
e.g,~\cite{lundelius84,ostrovsky99,shamir94,srikanth87}).
Until recently, the main focus has been on bounding the clock skew that may
occur between any two nodes in the network. Using the well-known shifting
argument~\cite{lundelius84}, which exploits the variable message delays to
construct indistinguishable executions, it has been shown that a clock skew of
$D/2$ cannot be prevented on any graph of diameter $D$~\cite{biaz01}. This lower
bound holds even if clocks do not drift. Indistinguishable executions can also
be constructed by exploiting variable clock rates~\cite{dolev84}, which can be
used together with the shifting argument to prove a stronger lower bound of
roughly $D$ for algorithms that must ensure that all clock values are always
within a linear envelope of real time~\cite{lenzen10}. In light of these
results, the algorithm proposed by Srikanth and Toueg~\cite{srikanth87} is
asymptotically optimal as it guarantees a skew of at most $\BO(D)$ between any
two clocks. The accuracy of their algorithm is also optimal in the sense that
all clock values are within a linear envelope of real time, i.e., a better
accuracy with respect to real time cannot be guaranteed. A crucial shortcoming
of this algorithm is that a clock skew of $\Omega(D)$ may occur between
neighboring nodes.
  
The problem of synchronizing clocks of nodes that are close-by as accurately as
possible has been introduced by Fan and Lynch~\cite{fan04}. In their work, the
authors show that a clock skew of $\Omega(\log D / \log\log D)$ between
neighboring nodes cannot be avoided if the clock values must increase at a
constant minimum progress rate. Subsequently, this result has been improved to
$\Omega(\log D)$~\cite{lenzen10}. If we take the minimum logical clock rate
$\alpha$, the maximum logical clock rate $\beta$, and the maximum clock drift
rate $\rho$ into account, the more general statement of the lower bound is that
a clock skew of $\Omega(\log_b D)$, where $b :=
\min\{1/\rho,(\beta-\alpha)/(\alpha\rho)\}$ cannot be avoided. The first
algorithm guaranteeing a sublinear bound on the worst-case clock skew between
neighbors achieves a bound of $\BO(\sqrt{\rho D})$~\cite{locher09diss,locher06}.
Recently, this result has been improved to $\BO(\log D)$~\cite{lenzen08} (where the
base of the logarithm is a constant) and subsequently to $\BO(\log_b
D)$~\cite{lenzen10}. Thus, tight bounds have been achieved for static networks
in which neither nodes nor edges fail.

The problem of synchronizing clocks in the presence of faults has also received
considerable attention (see, e.g.,
\cite{dolev95,halpern84,lamport85,lundelius84b,marzullo83}). Some of the proposed
algorithms are able to handle not only simple crash failures but also
Byzantine behavior, which is outside the scope of this article. However, while
these algorithms can tolerate a broader range of failures, their network model is
not fully dynamic as their results rely on the assumption that a large part of
the network remains non-faulty and stable at all times. For the fully dynamic
setting, it has been shown that there is an inherent trade-off between the clock
skew $\mathcal{S}$ guaranteed between neighboring nodes that have been connected
for a long time and the time it takes to guarantee a small clock skew over newly
added edges. In particular, the time it takes to reduce the clock skew over new
edges to $\BO(\mathcal{S})$ is $\Omega(D / \mathcal{S})$, where $n$ denotes the
number of nodes in the network~\cite{kuhn09}. In the same work, it is shown that
for $\mathcal{S} \in \Omega(\sqrt{\rho D})$, there is an algorithm that reduces
the clock skew between any two nodes to $\BO(\mathcal{S})$ in
$\Theta(D/\mathcal{S})$ time. In this article, we show that $\mathcal{S}$ can be
reduced to $\BO(\log_b D)$, i.e., the same optimal bound as for static
networks can be achieved, while still establishing this bound within
$\Theta(D/\mathcal{S})$ time on newly formed edges.

Another notion of fault-tolerance is
\emph{self-stabilization}~\cite{dijkstra74}, i.e., the ability to recover
correct operation after a period of arbitrary transient faults. Many clock
synchronization algorithms are self-stabilizing simply because of their
continuous strive for maintaining synchronization. However, a strong gradient
property is a more involved requirement than just minimizing the global skew,
hence self-stabilization is not immediate for our algorithm; the previous works
on the static case do not yield this result. In contrast, in the dynamic
setting, we exploit self-stabilization properties of the algorithm in order to
safely establish the gradient property on recently appeared edges (without
disrupting the guarantees for edges that have been present for a long time).
Consequently, we obtain self-stabilization of the gradient property as a
corollary of our analysis.

\section{Preliminaries}
\label{sec:model}

In this section we introduce the dynamic clock synchronization problem and the
model for dynamic networks that will be used in this paper.
We begin by reviewing classical (static) clock synchronization.

\paragraph{Clock synchronization.} In the clock synchronization problem, each
node $u$ is equipped with a continuous and differentiable
\emph{hardware clock} $\hc{u}:\R^+_0\to \R^+_0$, which is initialized to
$\hc{u}(0) \coloneq 0$. We use $h_u(t)$ to denote the rate $\dH{u}(t)$ at which
node $u$'s hardware clock advances at time $t$.\footnote{Unless otherwise
specified, times are always in $\R^+_0$.} The hardware clocks advance at
roughly the rate of real time, but they suffer from clock drift bounded by
$\drift \in (0,1)$; formally, we assume that at all times $t$ we have $h_u(t)
\in [1 - \drift, 1 + \drift]$ for all nodes $u$.
As a result, for any two times $t_1 \leq t_2$ we have
\begin{equation*}
	(1 - \drift)(t_2 - t_1) \leq \hc{u}(t_2) - \hc{u}(t_1) \leq (1 + \drift)(t_2 - t_1).
\end{equation*}

The objective of a clock synchronization algorithm (CSA) is to output a
left-differentiable\footnote{This requirement can be dropped. It is introduced
to simplify the presentation. The same results can be derived even for
discontinuous (in particular discrete) clocks by approximating the true clocks
by left-differentiable functions and accounting for the difference in the
uncertainty of estimates.} \emph{logical clock} $\lc{u} : \R^+_0 \to \R^+_0$
(also initialized to $\lc{u}(0)\coloneq 0$), such that at all times, the logical
clock values of different nodes are close to each other (we elaborate on this
requirement below). We use $l_u(t)$ to denote the rate $\dL{u}(t)$ of $u$'s
hardware clock at time $t$. The logical clocks are also required to have bounded
drift: there must exist constants $\alpha, \beta > 0$, such that for all $t$ we
have $l_u(t) \in [\alpha, \beta]$.

In the algorithm we present in this paper, nodes always increase their logical
clocks at either the rate of their hardware clock $h_u(t)$, or at a rate of $(1
+ \mu) \cdot h_u(t)$, where $\mu\in \BO(1)$ is a parameter of the algorithm.
Thus, the algorithm bounds the drift of the logical clocks, and we have that $\alpha \coloneq 1 -
\drift$ and $\beta \coloneq (1 + \drift)(1 + \mu)$.

\subsection{The Dynamic Graph Model}
\label{sec:dynamicgraphmodel}

\paragraph{The estimate graph.} In~\cite{kuhn09b} two of the authors introduced an abstraction called \emph{the
  estimate layer}, which simplifies reasoning about
CSAs. Synchronization typically involves periodic exchanges of clock
values between nodes, either through direct communication, or by other
means (e.g., reference broadcast synchronization~\cite{RBS}). The
estimate layer encapsulates all means by which nodes can estimate the
clock values of other nodes,
and eliminates the need to
reason explicitly about delay bounds and other parameters of the system.

The estimate layer provides an \emph{estimate graph}, where each edge $\set{u,v}$
represents the fact that node $u$ has some means of estimating $v$'s current
clock value and vice versa. The edges of the estimate graph are not necessarily
direct communication links between nodes (see~\cite{RBS} for examples). Node $u$
is provided with a \emph{local estimate} $\lce{u}{v}$ of $\lc{v}$, whose
accuracy is guaranteed by the estimate layer:
\begin{equation}
	\forall t \; \forall u \in V, v \in \Nb{u}(t) : |\lc{v}(t)-\lce{u}{v}(t)| \le
\err{\set{u,v}},
\label{eq:lce}
\end{equation}
where $\err{\set{u,v}}\in \R^+$ is called the \emph{uncertainty}, or the
\emph{weight}, of the edge $\set{u,v}$, and $\Nb{u}(t)$ is the set of neighbors
of $u$ at time $t$, which will be formally introduced shortly. The
uncertainty of a
path $p=(u_0,u_1,\dots,u_k)$, $\err{p}$ is defined as
\[
\err{p} := \sum_{i=1}^k \err{\set{u_{i-1},u_i}}.
\]

In the sequel, we refer to estimate edges of the sort described above simply as
\emph{edges}; similarly, when we say ``the graph'' we mean the estimate graph. We
do not reason explicitly about the communication graph, as the salient aspects of
communication are encapsulated by the estimate layer.

\paragraph{Dynamic networks.} We consider dynamic networks over a fixed
set of nodes $V$ of size $n \coloneq |V|$.
Edge insertions and removals are modeled as discrete events controlled by a
worst-case adversary. In keeping with the abstract representation
from~\cite{kuhn09b}, we say that there is an \emph{estimate edge} $\set{u,v}$
between two nodes $u,v\in V$ at time $t\geq 0$ iff $u$ and $v$ have a means of
obtaining clock value estimates about each other at time $t$. As explained above,
this does not necessarily mean that there is a direct communication link between
$u$ and $v$ at time $t$. 

We do not assume that nodes detect the formation or failure of a communication link between them at the same time,
which introduces some asymmetry into the model.
Hence, we model the network as a \emph{directed dynamic graph} $G = (V, E)$, where
$E : \posreals \rightarrow 2^{(V \times V)}$ maps non-negative times $t$ to a set of directed
estimate edges $E(t)$ that exist at time $t$. 
If $(u, v) \in E(t)$, then at time $t$ node $u$ has an estimate for node $v$'s logical clock, but not necessarily
vice-versa.
Formally, the set of node $u$'s neighbors at time $t$ is defined as $\Nb{u}(t) \coloneq \set{ v \st (u,v) \in E(t)}$.
We assume that any asymmetry in the graph corresponds to the delay in nodes finding 
out about link status changes and is only temporary; this is explained below.

In the following, we frequently refer to \emph{undirected} edges $\set{u,v}$; when we write $\set{u,v} \in E(t)$, we mean
that both $(u,v) \in E(t)$ and $(v,u) \in E(t)$.
We say that edge $\set{u,v}$ {exists} throughout a time interval $[t_1,
t_2]$ if for all $t \in [t_1, t_2]$ we have $\set{u,v} \in E(t)$. By extension, a
path $p$ is said to exist throughout $[t_1, t_2]$ if all its edges exist
throughout the interval.

Each undirected estimate edge $\set{u,v}$ is associated with three parameters:
\begin{itemize}
  \item The \emph{estimate uncertainty} $\err{\set{u,v}}$, as explained above.
  \item The \emph{detection delay} $\tup_{\set{u,v}}$. We assume that $u$ and
  $v$ detect if the edge disappears ``at'' the respective other node within
  $\tup_{\set{u,v}}\in \R^+$ time. Formally,
  \begin{enumerate}[(a)]
      \item if $(u,v)\notin E(t)$, then there is some time
      $t' \in [t - \tup_{\{u,v\}}, t + \tup_{\{u,v\}}]$ so that $(v,u)\notin
      E(t')$; and, symmetrically,
      \item if $(v,u)\notin E(t)$, then there is some time
      $t' \in [t - \tup_{\{u,v\}}, t + \tup_{\{u,v\}}]$ so that $(u,v)\notin
      E(t')$.
    \end{enumerate}
    \item We assume that $u$ and $v$ can exchange messages with
    \emph{message delay} $\delay_{\{u,v\}}$. More precisely, nodes that share an
    estimate edge can actively exchange information if required (possibly
    through other nodes, if there is no direct communication link between them),
    and $\delay_{\set{u,v}}$ bounds how long such communication might be
    delayed. Formally, if $u$ sends a message at time $t$ and $u\in \Nb{v}(t')$
    for all $t'\in [t,t+\delay_{\{u,v\}}]$, then $v$ will receive this message
    at some time $t''\in [t,t+\delay_{\{u,v\}}]$.\footnote{Note that the
    neighbor relation seems to be ``reversed'' here in the prerequisite for the
    reception of a message. This definition reflects that the estimate edge must
    exist for the node receiving the message. However, this detail is irrelevant
    to the functionality of the algorithm.} If $u$ is not during this entire
    interval in $\Nb{v}$, the message may or may not be delivered; if it is
    delivered, however, it is guaranteed to arrive within the specified
    interval.
\end{itemize}
We remark that our algorithm will use explicit communication by messages as
above only upon formation of an edge, to perform a simple handshake.

\paragraph{Causality and the dynamic estimate diameter.}
While local message exchange may be infrequent, flooding techniques
may ensure quick dissemination of timing information on a global level
without necessitating a large (amortized) number of messages per time
unit. We therefore characterize how information propagates through the
dynamic graph without imposing a particular communication
structure. In this context, we are interested in the global skew. Our
algorithm will ensure that any node whose logical clock attains the
current maximum clock value will run at the speed of the hardware
clock, i.e., no faster than at rate $1+\drift$. Further, the logical
clock of any node always runs at least at rate $1-\drift$. For
estimating the global skew, the maximum logical clock speed
$(1+\drift)(1+\mu)$ is of no significance.  More generally, this is
true for any algorithm that satisfies an optimal envelope condition,
i.e., that guarantees the best approximation of real-time offered by the
hardware clocks.

For a synchronization message $M$ sent from $u$ at time $t$ that is
received by $v$ at time $t'>t$ let $U(M)$ denote the uncertainty in its delay,
i.e., in particular the receiver $v$ knows that $M$ was in transit for
at least $t'-t-U(M)$ time units (clearly, $U(M)\leq \delay_{\{u,v\}}$
but potentially it is much smaller).

We define the family of relations $\stackrel{\eta}{\leadsto}$,
$\eta\in \R^+_0$, on $V\times \mathbb{R}$ as specified
below. Intuitively, for node $u$ and $v$, times $t$ and $t'\geq t$,
and a value $\eta\geq 0$, $(u,t)\stackrel{\eta}{\leadsto}(v,t')$ can
be interpreted as follows. At time $t'$, node $v$ can lower bound
$u$'s clock value at time $t$ (hardware or logical) with an error of at most
$\eta$. Specifically,

\begin{itemize}
  \item $\forall u\in V,\,\forall t:
  (u,t)\stackrel{0}{\leadsto} (u,t)$ ($u$ knows its own clock perfectly).
\item $\forall u,v\in V,\,\forall t''\geq t'\geq t,\,\forall \eta\in
  \R^+_0: (u,t)\stackrel{\eta}{\leadsto} (v,t')\Rightarrow
  (u,t)\stackrel{\eta'}{\leadsto} (v,t'')$, where
  $\eta':=\eta+\frac{4\drift}{1+\drift}(t''-t')$\\ ($v$ knows
  that $u$'s clock runs at least at $\frac{1-\drift}{1+\drift}$ times
  the rate of $v$'s hardware clock. The maximum error is obtained if
  $u$'s clock runs at rate $1+\drift$ and $v$'s hardware clock runs at
  rate $1-\drift$).
\item If $M$ is a message sent by $v$ at time $t'$ and received by $w$
  at time $t''\geq t'$, then $\forall u\in V,\,\forall t\leq
  t',\,\forall \eta\in\R^+_0:(u,t)\stackrel{\eta}{\leadsto}
  (v,t')\Rightarrow (u,t)\stackrel{\eta'}{\leadsto} (w,t'')$, where
  $\eta':=\eta+(1-\drift)U(M)+2\drift(t''-t')$\\ ($u$'s hardware clock
  progresses by at most $(1+\drift)(t''-t')$ during the transit time,
  but $w$ can safely add $(1-\drift)(t''-t'-U(M))$ to the estimate).
\end{itemize}

A fundamental lower bound~\cite{nodriftbound} shows that the
performance of a CSA in a static network depends on the diameter of
the network. In dynamic networks there is no immediate equivalent to a
diameter. Informally, the diameter corresponds to the time
it takes (at most) for information to spread from one end of the
network to the other. The above relation integrates this information
with the amount of uncertainty that is attached to this communication;
this is crucial in our scenario comprising heterogeneous edges since, for
instance, a communication path that is slower in terms of the time it
takes to traverse it might yield much more accurate estimates of clock
values.

\begin{definition}[Dynamic Estimate Radius and Diameter] Given a
  dynamic graph $G$, we say that node $v\in V$ has a \emph{dynamic
    estimate radius} of $R_v(t)$ at time $t$ if for every $u \in V$,
  there is some $t'\leq t$ so that $(u,t') \stackrel{R_v(t)}{\leadsto}
  (v,t)$, where $R_v(t)$ is minimal with this property. Moreover, $G$
  has a \emph{dynamic estimate diameter} $D(t):=\max_{v\in V}\{R_v(t)\}$
  (or simply ``diameter'' for short).
\end{definition}
Because this definition refers to the actual communication, (some) dynamic
estimate radii might be much smaller than the dynamic diameter at the same
instant of time. Moreover, both values strongly depend on the structure of
message exchange. However, the lower bounds from the static case apply in the
sense that the dynamic estimate diameter is lower bounded in terms of the
maximum over all pairs of nodes $v,w$ of the minimal sum of uncertainties on
\emph{any} possible communication path from $v$ to $w$. Hence, if the
communication layer provides an asymptotically optimal dynamic diameter, a
global skew bound that behaves roughly as $\BO(D(t))$ (neglecting disturbances
due to large fluctuations of $D(t)$) is asymptotically optimal.

As the primary focus of this work is not on the global skew, we refrain from
further discussing these points except for the following remark. We can make
an arbitrary node $u_0$ artificially faster (by multiplying its hardware clock
rate by $(1+\drift)/(1-\drift)$) so that it is always the node with the maximal
hardware clock value in the network. This can be seen as replacing its hardware
clock by one of drift $\tilde{\rho}\leq (1+\drift)^2/(1-\drift)-1\approx 3\rho$.
A CSA can easily guarantee that this node also has the largest logical clock value
in the network at all times.
All our statements then apply if we replace the drift bound $\rho$ by
$\tilde{\rho}$ and $D(t)$ by $R_{u_0}(t)$, which might be beneficial in networks
with a large discrepancy between $D(t)$ and $R_{u_0}(t)$.

\subsection{Dynamic Clock Synchronization}

Throughout the paper, we frequently refer to \emph{the skew on a path $p = (u_0, \ldots, u_k)$} at
time $t$, by which we mean $|\lc{u_0}(t) - \lc{u_k}(t)|$.
The goal of a CSA is to minimize the skew on all paths.

To measure the quality of a CSA we consider two kinds of
requirements: a \emph{global skew constraint} which gives a bound on the
difference between any two logical clock values in the system, and a
\emph{gradient skew constraint}, which becomes stronger the closer two
nodes $u,v$ are in the subgraph induced by the edges that have been present for
a sufficiently long time to stabilize. In particular, for nodes that remain
neighbors for a long time, the gradient skew constraint imposes a much
smaller permissible clock skew than the global skew constraint.

\begin{definition}[Global Skew]
	For any time $t$, a CSA guarantees a \emph{global skew} of ${\cal G}(t)$, if for any two nodes
	$u,v\in V$ it holds that $\lc{u}(t) - \lc{v}(t) \leq {\cal G}(t)$.
\end{definition}

\begin{definition}[Stable Gradient Skew] Given a non-decreasing
	function $\localskew : \posreals \rightarrow \posreals$, we say that a CSA
	$\mathcal{A}$ guarantees a \emph{stable gradient skew} of $\localskew$ with
	\emph{stabilization time} $\stabtime$ if for each time $t$ and each
	path $p = (u_0, \ldots, u_k)$ that exists throughout $[t-\stabtime, t]$, we
	have that
	\begin{equation*}
		\lc{u_0}(t) - \lc{u_k}(t) \leq \localskew(\err{p}).
	\end{equation*}
\end{definition}
More generally, one can express the skew bound as a function of the length of the
time interval during which the path $p$ existed (cf.~\cite{kuhn10}).
The literature on gradient clock
synchronization (e.g.,~\cite{gradientclocksynch,kuhn09,lenzen10,locher06}) is
typically concerned with the \emph{local skew} of a CSA, which bounds the skew
on any single edge. The local skew can be considered equivalent to the
stable gradient skew $\localskew(1)$, provided that all edges are of uniform
weight~$1$.

The stable gradient skew and the stabilization time are functions of $D$, a
bound on the dynamic estimate diameter of the network that held for sufficient
time (and are thus inherently dependent on $t$ as well), and potentially other
parameters such as the bound on the clock drift $\rho$ or the minimum edge
weight. We usually omit these dependencies to simplify the notation.

%%% Local Variables: 
%%% mode: latex
%%% TeX-master: "main"
%%% End: 

\section{An Optimal Dynamic Gradient CSA}
\label{sec:algorithm}

In this section we describe a CSA $\Aopt$ which achieves the optimal stable
gradient skew, and reaches this stable skew in the optimal stabilization time,
in light of the trade-off presented in Section~\ref{sec:lower}. We begin in
Section~\ref{sec:algstrategy} by introducing the overall strategy used to
achieve a stable skew of $\Theta(d \log (D / d))$ in static graphs; this
strategy also underlies the design of the dynamic algorithm. In
Section~\ref{sec:algoverview}, we give an informal overview of the algorithm, and
the technical details follow in Section~\ref{sec:algformal}. We remark that both
the description of the algorithm and in particular its analysis given in
Section~\ref{sec:analysis} is complicated by a number of technical details that
need to be resolved, but may obfuscate the key ideas behind the reasoning. We
refer the reader to \cite{lenzen11} for a simplified presentation in a less involved (but
unrealistic) model focusing on the key aspects of the problem and its analysis.

\subsection{Achieving a Stable Skew of \texorpdfstring{\boldmath{$\Theta(d \log(D /
d))$}}{Theta(d log(D/d))}}
\label{sec:algstrategy}
The optimal static algorithm~\cite{kuhn09b,lenzen10} and the
algorithm we present here share the same high-level structure. Both achieve a
(static or stable) gradient skew of $\Theta(d \log_{\sigma}(D / d))$ on paths of
length (or weight) $d$, where the base $\sigma$ of the logarithm is a function
of the parameter $\mu$ and the drift $\drift$. In this section we introduce
several notions that underlie the design of both algorithms. For simplicity, we
ignore here the dynamic behavior of the graph, and present the \emph{static-graph}
version of the definitions (as used in~\cite{kuhn09b,lenzen10}),
assuming that all edge weights are 1. This version is simpler than the weighted
dynamic-graph version and is helpful in understanding the dynamic algorithm. In
Section~\ref{sec:analysis} we give the full dynamic versions of these notions
and use them to analyze the dynamic skew of the algorithm.

The static algorithm is based on a discretized version of the gradient skew requirement.
Let $C = \set{C_s}_{s \in \nat}$ be the non-increasing sequence defined by $C_s
\coloneq D / \sigma^s$. The algorithm guarantees the following condition (up to
constants we neglect here): for any path $p = (u_0, \ldots, u_k)$ and any
integer $s \in \nat$, if the path $p$ has length $d_p \geq C_s$, then at all
times $t$ we have
\begin{equation*}
	\lc{u_0}(t) - \lc{u_k}(t) \leq s \cdot d_p.
\end{equation*}
This discretized condition is equivalent to the standard $\Theta(d \log_{\sigma}
(D / d))$-gradient skew requirement: if $p$ is a path of length $d_p$, then for
$s = \lceil \log_{\sigma} (D / d_p) \rceil$ we have
\begin{equation*}
	C_s = \frac{D}{\sigma^{\lceil \log_{\sigma} (D / d_p) \rceil}} \leq d_p,
\end{equation*}
and therefore the discretized condition asserts that the skew on $p$ is no
greater than $s \cdot d_p = \lceil \log_{\sigma} (D / d_p) \rceil \cdot d_p \in
\Theta(d_p \log_{\sigma}(D / d_p))$.

From the algorithm's point of view, the discretized condition divides the paths
into \emph{levels}, where paths of level $s$ are of length $d \approx
D / \sigma^s$ and the skew on such paths is upper bounded by $s \cdot
d$. If we evenly distribute the permissible skew over the edges of the path,
we see that each of the $d$ edges should only contribute a skew of roughly $s$
to the total. And indeed, this is exactly what each node executing the algorithm
tries to accomplish: it tries to ensure that for all $s \in \nat$, none of its
edges exhibit a skew of more than $s$. Similarly, in the weighted version of the
static algorithm~\cite{kuhn09b}, each node tries to ensure that no
adjacent edge of weight $w_e$ carries a skew of more than $s \cdot w_e$, so that
when we sum over all the edges of a path $p$ of weight $w_p$ the total skew will
be no more than $s \cdot w_p$. The overall gradient skew is then $\BO(w_p \cdot
\log_{\sigma} (D / w_p))$, a direct generalization of the unweighted case.

The description above is informal but we will see that tests of
the form ``is there some neighbor whose clock is more than $s \cdot w_e$ ahead
or behind?'' make up the basis of the algorithm. Essentially, through such tests
nodes check if their adjacent edges contribute more than their fair share of the
skew on some path.

If a node finds that the skew over some of its edges is too large, it can adjust the
speed of its logical clock to compensate. The algorithm uses only two rates, a
\emph{slow rate} and a \emph{fast rate}. When a node uses the slow rate we
say that it is in \emph{slow mode}, and when it uses the fast rate we say that
it is in \emph{fast mode}. At the heart of the algorithm are the rules for
deciding which mode to use; we proceed to describe these rules, which are based
on the static rules from~\cite{kuhn09b,lenzen10} but also take into
account the dynamic behavior of the graph.

\subsection{Overview of the Algorithm}
\label{sec:algoverview}

When an edge first appears, the algorithm is first concerned with reducing the
skew on \emph{long} paths that contain the edge.
Once this is accomplished, it allows the skew on shorter paths to also be
reduced, and then on even shorter paths, until eventually the skew on individual
edges is reduced to its stable value. In some sense, the algorithm takes the
global skew $\G$, which cannot be avoided, and \emph{redistributes} it
throughout the network until the gradient property is satisfied. Notice that
longer paths have a larger (that is, weaker) gradient skew bound, so they are in
some sense easier to deal with. In particular, for the longest paths in the
network, the gradient skew bound is the same as the global skew bound. Since the
global skew bound holds for any two nodes in the network, it can never be
violated by adding new edges, so these longest paths immediately satisfy their
gradient skew requirement as soon as they appear.

%###
\paragraph{Neighbor sets.}
%###
Throughout the algorithm, each node partitions its neighbors according to the
amount of time it has had an edge to each neighbor. More precisely, each node
$u$ maintains an ordered list $\N_u^0, \N_u^1, \ldots$ of neighbor sets, where
$\N_u^0 \supseteq \N_u^1 \supseteq \ldots$. To simplify the presentation we
initially assume an infinite list of sets; we will later see that nodes only
need to store a finite prefix of the list, but we defer this discussion to a
later point. Moreover, as the neighbor sets change at discrete
times, we need a convention what $\N_u^s(t)$ means if the set is modified at
time $t$. We define that if node $v$ is added to $\N_u^s$ at time $t^-$ and
removed at time $t^+$ (without intermediately leaving the set) then $v\in
N_u^s(t)$ for all $t\in [t^-,t^+]$. We assume that the neighbor sets
change only finitely often in finite time, implying that for all $u$, $v$, and
$s$, the set $\{t\,|\,v\in N_u^s(t)\}$ is closed.\footnote{This convention
simplifies the notation in our proofs. However, since clocks are continuous
functions, this convention does not bear any implication for the behavior of
the algorithm.}

Informally, if $v \in \N_u^s$ at time $t$, then at time $t$ node $u$ is
concerned with maintaining a good skew on paths of level $s$ containing edge
$\set{u,v}$. In contrast, if $v \not \in \N_u^s$, then node $u$ is ``not
worried'' about level $s$ paths containing $\set{u,v}$. Accordingly, when an
edge $\set{u,v}$ is discovered, node $u$ first adds $v$ to $\N_u^0$, then after
some time it adds $v$ to $\N_u^1$, and so on. Recall from
Section~\ref{sec:algstrategy} that the index $s$ of a level decreases as the
length of the path in the level increases, so adding edges in the order $\N_u^0,
\N_u^1, \ldots$ corresponds to dealing first with longer paths and then with
shorter ones.

Specifically, when node $u$ first discovers edge $\set{u,v}$, it
\emph{immediately} adds $v$ to $\N_u^0$; hence $\N_u=\N_u^0$, because
this is the set of all neighbors that node $u$ has discovered. Each of
the remaining sets is updated within time $\Theta(\G/\mu)$. The sets
are updated in a loosely synchronized manner. Both nodes $u$ and $v$
coordinate adding the edge $\set{u,v}$ to their respective
sets. In a time interval during which nodes add edges to their level $s$
neighbor set $N_u^s$, we can only show non-trivial gradient skew
guarantees for levels different from $s$. In order to always have
non-trivial guarantees for the skew on paths of all lengths, we need
to loosely synchronize the insertions of different edges such that
insertions of different edges on different levels are sufficiently
separated from each other. The details appear in Section
\ref{sec:algformal}.

Whenever a node $u$ discovers that one of its edges $\set{u,v}$ has
disappeared, it immediately removes node $v$ from all neighbor sets
$\N_u^0, \N_u^1, \ldots$ Finally, for simplicity, we assume that at
time $0$, $N_u(0)$ contains all edges that are present at time $0$ and
all neighbor sets $N_u^s$ are initialized to $N_u(0)$, i.e., for all
$s\geq 0$, $N_u^s(0)=N_u(0)$ as there is no violation at time $0$.

%###
\paragraph{The fast and slow conditions.}
%###
Each edge $e$ is associated with a weight $\kappa_e$, which roughly corresponds
to the uncertainty $\epsilon_e$ of the edge. The algorithm is designed to 
guarantee the following conditions governing when a node is in fast or in slow
mode. These conditions are not the actual rules used by nodes to determine when
to enter fast or slow mode, but we will see in Section~\ref{sec:algformal} that
the rules are quite similar; the conditions we give here refer to the clock
values of neighbors, which a node cannot estimate exactly, and the actual
triggers for entering fast or slow mode have to compensate for this inaccuracy.
We will see in Section~\ref{sec:analysis} that the fast and slow mode triggers
(given in Section~\ref{sec:algformal}) implement the fast and slow mode
conditions (given below).

The first condition, \FC, specifies when a node $u$ must be in fast mode. It
states that some neighbor in $\N_u^s$ is ``too far ahead of $u$'', and no other
neighbor in $\N_u^s$ is ``too far behind'', where ``too far'' here roughly
corresponds to $s$ times the weight of the edge (as outlined in
Section~\ref{sec:algstrategy}).
\begin{definition}[\FC: The Fast Mode Condition] For all $s \in \nat, u \in V$,
and times $t$, if
	\begin{itemize}
		\item For some $w \in \N_u^s(t)$ we have $L_w(t) - L_u(t) \geq s \cdot \kappa_{\{u,w\}}$, and
		\item For all $v \in \N_u^s(t)$ we have $L_u(t) - L_v(t) \leq s \cdot \kappa_{\{u,v\}}+2\mu\tup_{\{u,v\}}$,
	\end{itemize}
then node $u$ is in fast mode at time $t$.
\label{def:fc}
\end{definition}
The term $2 \mu \tup_{\set{u,v}}$ in the second requirement compensates for the
drift that can accumulate on an edge while only one of its endpoints is aware of
its existence (recall that the length of this period is bounded by
$\tau_{\set{u,v}}$).

The condition \SC\ for being in slow mode is roughly symmetric to the fast mode
condition: it states that some node in $\N_u^s$ is ``too far behind $u$'', and
no other node in $\N_u^s$ is ``too far ahead''. The condition uses a value
$\delta>0$ that corresponds to the smallest uncertainty in the network. An
exact value will be defined in Lemma~\ref{lemma:algo_correct}; the algorithm is
oblivious of $\delta$ and is correct if there exists \emph{any} such $\delta>0$
under which the slow mode condition is satisfied, and
Lemma~\ref{lemma:algo_correct} shows that such a value exists.
\begin{definition}[\SC: The Slow Mode Condition] For all $s \in \nat, u \in V$,
and times $t$, if
	\begin{itemize}
		\item For some $w \in \N_u^s(t)$ we have $L_u(t)-L_w(t)\geq \left(s+\frac{1}{2}\right) \cdot \kappa_{\{u,w\}}-\delta$, and
		\item For all $v \in \N_u^s(t)$ we have $L_v(t)-L_u(t)\leq \left(s+\frac{1}{2}\right) \cdot \kappa_{\{u,v\}}+\delta + \mu(1 + \drift)\tau_{\set{u,v}}$,
	\end{itemize}
	then node $u$ is in slow mode at time $t$.
	\label{def:sc}
\end{definition}
The slow mode condition uses a slightly different value for ``too far'' from the
fast mode condition. There are two immediate reasons for this: first, the
conditions for being in fast mode and in slow mode must be mutually exclusive,
otherwise a node might be required to be in both modes at the same time; hence
the term $(s + 1/2)$ instead of $s$. And second, the slack $\delta$ is necessary
to smooth out the discontinuities that occur when a neighbor is removed from
$\N_u^s$, by providing a small region around $(s + 1/2)\cdot \kappa_e$ in which
a neighbor of $u$ that is behind $u$ can still keep node $u$ in slow mode.
Lemma~\ref{lemma:neighbors_clocks} below captures this intuition formally and
shows how the slack $\delta$ is used. Moreover, for technical reasons in \SC
a smaller term of $(1+\drift)\mu \tau_{\set{u,v}}$ is sufficient to address the
issue that skew may accumulate while only one endpoint of an edge is aware of
the edge.

Note that the fast and slow mode conditions are disjoint, and their union does
not cover the entire state space.
In such cases nodes choose their mode according to their estimate
of the maximum logical clock in the network, as described below.

%###
\paragraph{Max estimates.}
%###
As in~\cite{kuhn09,lenzen10,locher09diss,locher06}, each node maintains a
local estimate $\mc{u}$ of the maximum logical clock value in the
network. Max estimates are computed by flooding: each node always adds
its current estimate $\mc{u}$ to each message it sends, and updates
the estimate conservatively so that it cannot exceed the actual
maximum logical clock value in the system. When a node receives a
larger max estimate from some neighbor, it updates its own max
estimate to match. The max estimates are computed such that the
following constraints can be guaranteed.

\begin{condition}\label{cond:max_estimate} If the dynamic graph has a dynamic
  estimate diameter of $D(t)$, then for all $t \geq 0$ and for all
  nodes $u$ we have
\begin{align}
	&\mc{u}(t) \leq \max_{v \in V} \set{\lc{v}(t)},
	\label{eq:mc_upper}\\
	&\mc{u}(t) \geq \max_{v \in V} \set{\lc{v}(t)} - D(t),
	\label{eq:mc_lower}\\
	&\mc{u}(t) \geq \lc{u}(t),
	\label{eq:mc_self}
      \end{align}
\end{condition}
Specifically, node $u$ updates its max estimate
$\mc{u}$ as follows. Whenever $\mc{u}=\lc{u}$, node $u$ increases
$\mc{u}$ at the rate of its logical clock.  If $\mc{u}>\lc{u}$, node
$u$ has to make sure that it only increases $\mc{u}$ at a rate such
that $\mc{u}$ remains upper bounded by the largest logical clock
$\lc{v}$ in the network. As the largest logical clock progresses at
rate at least $1-\drift$ and node $u$'s hardware clock progresses at
rate at most $1+\drift$, this can be achieved if $u$ increases
$\mc{u}$ at rate $\frac{1-\drift}{1+\drift}$ times the rate of its
hardware clock. These rules suffice to guarantee \eqref{eq:mc_upper}
and \eqref{eq:mc_self}. In order to also guarantee
\eqref{eq:mc_lower}, nodes piggy-back their current max estimate to
each message sent. Whenever a node $u$ receives a message from a node
$v$, $u$ increases its max estimate to the largest possible value such
that $\mc{u}$ is guaranteed to remain upper bounded by $\mc{v}$ (or by
the existing max estimate $\mc{u}$ if that is larger). Condition
\eqref{eq:mc_lower} now follows directly from the  definition of
$D(t)$.

Note that as a result of Condition~\ref{cond:max_estimate}, the max
estimate of any node is always accurate up to the diameter $D(t)$.  In
addition,~\eqref{eq:mc_self} asserts that nodes cannot set their
logical clock ahead of their max estimate. The max estimate
$\mc{u}(t)$ is used to determine the mode of node $u$ when neither
$\FC$ nor $\SC$ are satisfied.
\begin{definition}[\MC: The Max Estimate Condition]
  \label{def:MC}
  For all $u\in V$ and times $t$:
  \begin{itemize}
  \item If $\lc{u}(t) = \mc{u}(t)$ and for all $v\in \Nb{u}(t)$, we
    have $\lc{u}(t)\geq\lc{v}(t)$, then node $u$ is in slow mode at time $t$.
  \item If $\lc{u}(t) \leq \mc{u}(t) -\iota$ and for all
    $v\in\Nb{u}(t)$, we have $\lc{u}(t)\leq\lc{v}(t)$, then node $u$
    is in fast mode at time $t$,
  \end{itemize}
  where $\iota >0$ is some small constant used to separate the two conditions.%
  \footnote{The analysis of the algorithm goes through even if we
    change the condition for entering fast mode to $\lc{u}(t) <
    \mc{u}(t)$. However, such a requirement cannot be realized,
    because there is no ``first point in time'' when $\lc{u}(t) <
    \mc{u}(t)$. To ensure that the algorithm is realizable we make
    sure that when we require a node to be in a certain mode, the
    conditions of the requirement form a closed region, and are strictly separated
    from any other requirement.}%
\end{definition}

%###
\paragraph{Global skew estimates.}
%###

At all times $t\geq 0$, the algorithm requires each node $u$ to have
an estimate $\Ge_u(t)$ of the global skew $\G(t)$. We require that
\begin{equation}\label{eq:global_est_dynamic}
  \text{For all nodes $u\in V$ and for all times $t\geq 0$, }
  \Ge_u(t) \geq \G(t).
\end{equation}
It turns out that a lack of guarantees on the accuracy of these estimates and/or
their speed of change over time and across the network significantly complicates
edge insertion. For the sake of a more accessible presentation, we thus assume a
static (i.e., neither time- nor node-dependent) global skew estimate
\begin{equation}\label{eq:global_est_static}
  \text{For all times $t\geq 0$, }
  \Ge \geq \G(t).
\end{equation}
for now. Note that $\Ge$ must be chosen conservatively, as it must bound the
global skew for all times, and thus relying on it may result in unnecessarily
slow edge insertions. We will discuss how to adapt edge insertion to the much
weaker condition \eqref{eq:global_est_dynamic} in Section~\ref{sec:dynamic_est},
alongside a proof of the resulting (time-dependent) gradient property.

\subsection{Detailed Description of the Algorithm}
\label{sec:algformal}

We describe the parameters and constants used to define the algorithm, the local
variables maintained at each node, and finally the continuous and discrete
transitions that modify these variables.

\subsubsection{Parameters and Constants}
\label{sec:parameters}
\paragraph{\underline{\boldmath{$\drift$}:}} As specified in Section
\ref{sec:model}, the constant $\drift\in(0,1)$ specifies an upper
bound on the drift of the hardware clocks.
\paragraph{\underline{\boldmath{$\mu$}:}} This parameter governs the fastest
possible logical clock rate. In slow mode, the logical clock is increased at the
same rate as the hardware clock, and in fast mode the rate of the hardware clock
is multiplied by $1 + \mu$. The value of $\mu$ is bounded from below as a
function of the drift $\drift$, because we must ensure that a node in fast mode
is \emph{always} faster than a node in slow mode, even when the hardware clock
progresses slowly for the node in fast mode and
quickly for the node in fast mode. To ensure that for any $u,v \in V$ we always
have $(1 + \mu) h_u(t) > h_v(t)$ it is sufficient to require $(1 + \mu)(1 -
\drift) > 1 + \drift$, which is equivalent to $\sigma>1$ (see
below). For technical reasons, we require that
\begin{equation}
  \label{eq:mu}
  \mu \leq \frac{1}{10}.
\end{equation}

\paragraph{\underline{\boldmath{$\sigma$}:}} The base of the logarithm in the
desired gradient skew function, which is $\Theta(d \log_{\sigma} (D / d))$.
To obtain the best asymptotic gradient skew bound, we set
\begin{equation}
	\sigma\coloneq \frac{(1-\drift)\mu}{2\drift}>1,\label{eq:sigma}
\end{equation}
and control the base of the logarithm by setting the value of $\mu$
appropriately. Clearly, we must require that $\sigma > 1$, which imposes the
constraint that $\mu>2\drift/(1-\drift)$.

\paragraph{\underline{\boldmath{$\kappa_{\set{u,v}}$}:}} Each edge $\set{u,v}
\in \binom{V}{2}$ is associated with a \emph{weight} $\kappa_{\set{u,v}}$,
corresponding roughly to the uncertainty $\eps_{\set{u,v}}$. The weights must
satisfy
\begin{equation}
\kappa_{\set{u,v}} > 4(\err{\set{u,v}} + \mu\tup_{\set{u,v}}).
\label{eq:def_kappa}
\end{equation}
The term $4\mu \tup_{\set{u,v}}$ compensates for the time during
which the edge $\set{u,v}$ only exists for one of its nodes $u$ and $v$. Otherwise,
the asymmetric behavior could result in one of the nodes erroneously being in
fast mode (or, similarly, slow mode). Therefore, the time uncertainty of
$\tup_{\set{u,v}}$ with respect to the symmetric existence of edges is reflected
in $\kappa$ with a prefactor of $\mu$; we remark that $\mu$ also occurs as a
factor in a term contributing to $\err{\set{u,v}}$ (in any system), as there
must be a non-zero delay for propagating information about $L_u$ to $v$.

\subsubsection{Local Variables}
Each node $u$ maintains the following local variables throughout the execution
of the algorithm.

\paragraph{\underline{\boldmath{$\lc{u}$:}}} the logical clock of node $u$.

\paragraph{\underline{\boldmath{$\rate_u$}:}} the current rate-multiplier for
node $u$'s logical clock. It can take only two values, $1$ or $1 + \mu$; when
$u$ is in slow mode we have $\rate_u = 1$, and when $u$ is in fast mode we have
$\rate_u = 1 + \mu$.

\paragraph{\underline{\boldmath{$M_u$}}:} node $u$'s current estimate for the
maximum logical clock in the network.

\paragraph{\underline{\boldmath{$\N_u = \N_u^0$:}}} the set of all neighbors
node $u$ is aware of.

\paragraph{\underline{\boldmath{$\N_u^1, \N_u^2, \ldots$}:}} the neighbor sets
of node $u$ for each of the levels it maintains.

\paragraph{\underline{\boldmath{$\Ge_u$}}:} the nodes' current global skew
estimate. We assume that at all times $t$, $\Ge_u(t)$ is an upper
bound on the actual global skew at time $t$. Prior to
Section~\ref{sec:dynamic_est}, we assume that simply $\Ge_u(t)=\Ge$ for all
nodes and times.

\paragraph{\underline{\boldmath{${T_0}^{\set{u,v}}<{T_1}^{\set{u,v}}<{T_2}^{\set{u,v}}<
      \ldots$}:}} the logical times for adding edge $\set{u,v}$. For
each edge $\set{u,v}$ and each $s = 1,2,\dots$, nodes $u$ and $v$
decide on logical times $T_1^{\set{u,v}},T_2^{\set{u,v}},\dots$ when
they add the respective neighbor to their respective level-$s$
neighbor set. That is, node $u$ adds $v$ to $\N_u^s$ when its logical
clock reaches $\lc{u}(t) = T_s^{\set{u,v}}$. For convenience, each
node also maintains a logical time $T_0^{\set{u,v}}$ that is used to
define the times $T_s^{\set{u,v}}$ for $s\geq 1$. Note that the nodes
$u$ and $v$ use the same values for $T_0^{\set{u,v}}, T_1^{\set{u,v}},
\ldots$, but because their logical clocks are not perfectly
synchronized, they may update their neighbor sets at different times.
For each edge $\set{u,v}$, the times $T_1^{\set{u,v}},
T_2^{\set{u,v}}, \ldots$ define a converging sequence, so that edges
can be added on all infinite levels in finite time. In fact, in the
analysis we assume that all nodes update all sets, that is, they use
infinite levels. Note however, that only a $\BO(\log\Ge)$ levels are
needed in the algorithm. Further, the times $T_s^{\set{u,v}}$ only depend on
$T_0^{\set{u,v}}$ and on $\Ge_{\set{u,v}}$. All neighbor sets $N_u^s(t)$ are
therefore implicitly given by $\lc{u}(t)$ and some bounded additional
information for each edge $\set{u,v}$, so that the sets $N_u^s(t)$ could be
maintained by only managing a constant number of values per edge.

\subsubsection{Rules for Updating the Local Variables}

\begin{algorithm}[t!]
\caption{Responses to other events at node $u$}
\label{algo:events}
$\Delta:=\frac{(1+\rho)(1+\mu)(\delay_{\{u,v\}}+\tup_{\{u,v\}})}{1-\drift}+\tup_{\{u,v\}}$\;
\When(\textbf{\boldmath formation of an edge $\set{u,v}$ to node $v$ is discovered:}){
  $N_u^0 := N_u^0 \cup \set{v}$\;
  \If{$u$ is the leader of the edge $\set{u,v}$}{
    wait for at least $\Delta$ time\;
    \If{$v\in N_u^0(t')$ for all $t'$ with $L_u(t')\in
      [L_u(t)-(1+\rho)(1+\mu)\Delta,
      L_u(t)]$}{ $\Ge_{\set{u,v}} := \Ge_u$ \hspace{4.3cm} // we
      assume $\Ge_u = \Ge$, except in Section~\ref{sec:dynamic_est}\;
      $L_{\mathit{ins}} := \lc{u} + \Ge_{\set{u,v}} +
      (1+\rho)(1+\mu)\delay_{\set{u,v}}$\;
      \textbf{send} $\mathrm{insertedge}\left(\set{u,v}, L_{\mathit{ins}}, \Ge_{\set{u,v}} \right)$ to $v$\;
      \textbf{call} $\mathrm{computeInsertionTimes}\left(\set{u,v},
      L_{\mathit{ins}},\Ge_{\set{u,v}} \right)$\;
    }
  }
}
\When(\textbf{\boldmath receiving message $\mathrm{insertedge}\left(\set{u,v}, L_{\mathit{ins}}, \Ge \right)$
from node $v$:}){
  wait for at least $\delay_{\{u,v\}}+\tup_{\{u,v\}}$, but at most
  $\Delta-\tup_{\{u,v\}}$ time\;
  \If{$v\in N_u^0(t')$ for all $t'$ with
  $L_u(t')\in [L_u(t)-(1+\rho)(1+\mu)(\delay_{\{u,v\}}+\tup_{\{u,v\}}), L_u(t)]$}{
    \textbf{call} $\mathrm{computeInsertionTimes}\left(
    \set{u,v},L_{\mathit{ins}}, \Ge \right)$\;
  }
}
\When(\textbf{\boldmath failure of an edge to node $v$ is discovered:}){
  \ForEach{$s \in \set{0,1,\ldots}$}{
    $\N_u^s \coloneq \N_u^s \setminus \set{v}$\; \label{line:remove}
    $T_s^{\{u,v\}}:=\bot$\;
  }
}
\When(\textbf{\boldmath $\lc{u} = T_s^{\set{u,v}}$ (for some $s$ and $v$)}){
  $N_u^s := N_u^s \cup \set{v}$\;
}
\end{algorithm}

The algorithm makes three kinds of discrete transitions: the first
kind occurs when a node discovers the formation or failure of a
communication link. The second kind occurs when a node $u$'s logical clock
reaches an update time $T_s^{\set{u,v}}$ for some $s \in \nat$ and an
incident new edge. The responses to these events are given in
Listing~\ref{algo:events}.

The third and final kind of transition is triggered when the slow mode trigger,
the fast mode trigger, or the max estimate triggers, which correspond to \SC,
\FC, and \MC and will be stated shortly, require the node to change its mode;
the logic governing a node's mode is shown in Listing~\ref{algo:rate}.  When no
trigger holds and $\MC$ does not hold, the node is free to choose its mode
nondeterministically; for example, it can stay in its current mode until it is
required to switch modes.\footnote{For simplicity it is
  assumed that the code in Listing~\ref{algo:rate} is evaluated
  \emph{continuously}, so that, for example, as soon as the fast mode trigger
  holds for some node, that node is \emph{already} in fast mode. An
  implementation of the algorithm can achieve this by adding a small ``guard
  region'' to the conditions, and changing mode \emph{before} the triggers hold.
  All the triggers are strictly separated from each other, so such regions can
  be added to each trigger.}

Between discrete transitions the value of each node $u$'s logical clock
increases at a rate of $l_u = \rate_u \cdot h_u(t)$.
In the remainder of the section we describe the algorithm's discrete
transitions.

\paragraph{Coordinating with new neighbors and calculating the
  insertion times.}

When a new edge is formed, the two nodes start a simple protocol during which
they agree on the logical times for adding each other to the respective neighbor
sets. For simplicity, we assume that for each potential edge $\set{u,v}$, one of
the two nodes $u$ and $v$ is the leader of the edge. This can for example be
determined by assuming that nodes $u$ and $v$ have unique
identifiers.\footnote{If we drop the assumption that we can predefine a leader
for each potential edge, it would be possible to use a more complicated
handshake protocol to coordinate between the two nodes of an edge.} Assume that
node $u$ is the leader of an edge $\set{u,v}$. As soon as node $u$ discovers the
edge $\set{u,v}$, it starts the protocol for adding $\set{u,v}$. In order to
make sure that also node $v$ has discovered the edge, node $u$ first waits for
at least $\tup_{\set{u,v}}$ time units (w.r.t.\ real time). If the edge exists
throughout that waiting period, node $u$ decides on a global skew estimate
$\Ge_{\set{u,v}}$ for the edge insertion (which is just node $u$'s current
global skew estimate $\Ge_u$) and a logical time to start adding the edge.
Node $u$ sends this information to node $v$. If node $v$ sees the edge when
receiving the information, it computes the edge insertion times based on the
received information. The protocol guarantees that a) either both nodes insert
the edge or they do not start inserting or cancel the insertion within
$\tup_{\set{u,v}}$ time units of each other, and b) if both nodes insert the
edge, they use the same insertion times and global skew estimate for the
insertions. For further details and a formal argument, we refer to Lemma
\ref{lemma:insertion}. Pseudo-code of the coordination protocol is given in
Listing \ref{algo:events}. The computation of the insertion times based on a
logical time for start inserting and a given global skew estimate is given in
Listing \ref{algo:insertiontimes}. When inserting an edge $\set{u,v}$, $u$ and $v$ compute a time interval of length  $\instime_{\set{u,v}}$ during which the edge $\set{u,v}$ is inserted on all levels. The duration $\instime_{\set{u,v}}$ depends on the global skeq estimate $\Ge_{\set{u,v}}$ of the edge and it is computed differently depending on whether we work with a fixed, static global skew estimate $\Ge$ or whether the global skew estimate is allowed to be dynamically adapted. Outside Section~\ref{sec:dynamic_est}, we assume the global skew estimate to be a fixed value $\Ge$. The insertion duration $\instime_{\set{u,v}}$ of an edge $\set{u,v}$ is then computed as
\begin{equation}\label{eq:instime_static}
	\instime_{\set{u,v}} := \instime(\Ge) := \left(\frac{20(1+\mu)}{(1-\rho)}+56\mu+\frac{8+56\mu}{\sigma}\right)\cdot\frac{\Ge}{\mu}.
\end{equation}
 In Section~\ref{sec:dynamic_est}, we show how our clock synchronization algorithm can adapt to a changing global skew. In this case, the time for inserting an edge has to be chosen larger mainly because we need to make sure that the time is chosen such that it is based on a global skew estimate that holds during the complete insertion process. The insertion time is further increased because the insertions of different edges might use different global skew estimates and thus, the times of inserting the edges on different levels are harder to coordinate (and separate) properly. For details, we refer to Section~\ref{sec:dynamic_est}. In the case of a dynamic global skew, the insertion duration $\instime_{\set{u,v}}$ of an edge $\set{u,v}$ is computed as
\begin{align}\label{eq:instime_dynamic}
	\instime_{\set{u,v}} & := \instime(\Ge_{\set{u,v}}) := 2^{\lceil\log_2 \ell_{\set{u,v}}\rceil},\\
	& \text{where } \nonumber
	\ell_{\set{u,v}} :=  (1+\drift)(1+\mu)(\delta_{\set{u,v}}+2\tup_{\set{u,v}})+8\insconst\cdot\frac{\Ge_{\set{u,v}}}{\mu}.
\end{align}
The parameter $\insconst$ is a constant that is introduced
for convenience and which has to satisfy the following conditions:
\begin{equation}\label{eq:insconst}
  \frac{\mu}{2\drift}\geq \insconst \geq \frac{320\cdot
      2^7}{(1-\drift)^2}.
\end{equation}
We note that together with \eqref{eq:mu}, the above inequality
directly implies that for the dynamic global skew analysis in Section~\ref{sec:dynamic_est}, we can assume that
\begin{equation}\label{eq:rho_dynamic}
  \frac{\drift}{(1-\drift)^2} \leq \frac{1}{6400\cdot 2^7}.
\end{equation}

\begin{algorithm}[t]
\caption{Calculating insertion times}
\label{algo:insertiontimes}
\Proc(\textbf{\boldmath$\mathrm{computeInsertionTimes}\left(\set{u,v},
    L, \Ge\right):$}){
Compute $\instime_{\set{u,v}}:=\instime(\Ge_{\set{u,v}})$ according to \eqref{eq:instime_static} or \eqref{eq:instime_dynamic}\;
$T_0^{\set{u,v}} := \min\set{T\geq L\,:\, \frac{T}{\instime_{\set{u,v}}(\Ge)} \in \mathbb{Z}}$\;
  \For{$s \in \set{1,2,\dots}$}{
    $T_s^{\set{u,v}} := T_0^{\set{u,v}}
      + \left(1 - \frac{1}{2^{s-1}}\right)\instime_{\set{u,v}}(\Ge)$\;
  }
}
\end{algorithm}

\noindent In the technical analysis, we sometimes use $\instime$ for
$\instime(\Ge)$, if $\Ge$ is clear from the context. Note that the sequence $T_1^{\set{u,v}}, T_2^{\set{u,v}}, \ldots$ converges to
\begin{equation*}
	T_{\infty}^{\set{u,v}} := T_0^{\set{u,v}} + \instime_{\set{u,v}}.
\end{equation*}
Also note that although the sequence is infinite, it (and also the
sets $N_u^s$) can be implicitly stored using only bounded
information. Further, if an edge $\set{u,v}$ with leader $u$ appears at time
$t$, the total time to insert $\set{u,v}$ on all levels is in the order of
\[
\Theta\left(\delay_{\set{u,v}} + \tup_{\set{u,v}} + \instime_{\set{u,v}}\right)
\subseteq \BO\left(\delay_{\set{u,v}} + \tup_{\set{u,v}} +
\frac{\Ge_u(t)}{\mu}\right).
\]
For convenience, for a given execution and a level $s\geq 1$, we
define $\Tset_s$ to be the set of all level $s$ insertion times
$T_s^{\set{u,v}}$ used for any possible edge $\set{u,v}$ at any time.
Further, we define $\Tset:=\bigcup_{s\geq 1} \Tset_s$ to be the set of
all edge insertion times of a given execution.

\paragraph{The fast and slow mode triggers.}
The rules for deciding when to enter the fast mode or the slow mode correspond to the conditions
from Section~\ref{sec:algoverview}, but they compensate for the uncertainty of
the clock estimates to ensure that the conditions are satisfied. The triggers
for switching modes are as follows.

\begin{definition}[Fast Mode Trigger]
	Node $u$ satisfies the \emph{fast mode trigger} at time $t$ if there exists an integer $s \in \nat$ such that
	\begin{itemize}
		\item For some $w \in \N_u^s(t)$ we have $\lce{u}{w}(t) - \lc{u}(t) \geq s \cdot \kappa_{\set{u,w}} - \err{\set{u,w}}$, and
		\item For all $v \in \N_u^s(t)$ we have $\lc{u}(t) - \lce{u}{v}(t) \leq s \cdot \kappa_{\set{u,v}} + 2\mu \tup_{\set{u,v}} + \err{\set{u,v}}$.
	\end{itemize}
\end{definition}

The slow mode trigger incorporates some slack, which we also encountered in
Definition~\ref{def:sc}; we now define it as a parameter $\delta_e$ for each
edge $e$, and require
\begin{equation*}
	\delta_e \in \left(0, \frac{\kappa_e}{2} - 2\err{e} - 2\mu \tup_e\right).
\end{equation*}
This constraint ensures that the fast mode and the slow mode triggers are
mutually exclusive (see Lemma~\ref{lemma:disjoint}).
We note that $\frac{\kappa_e}{2} - 2\err{e} - 2\mu \tup_e >0$ due to~\eqref{eq:def_kappa}, which
constrains the choice of $\kappa_e$.
\begin{definition}[Slow Mode Trigger] Node $u$ satisfies the \emph{slow mode
trigger} at time $t$ if there exists an integer $s \in \nat$ such that
\begin{itemize}
	\item For some $w \in \N_u^s(t)$ we have
		$\lc{u}(t) - \lce{u}{w}(t)\geq \left(s+\frac{1}{2}\right)\kappa_{\set{u,w}}-\delta_{\set{u,w}} - \err{\set{u,w}}$, and
	\item For all $v \in \N_u^s(t)$ we have
		$\lce{u}{v}(t) - \lc{u}(t)\leq \left(s+\frac{1}{2}\right)\kappa_{\set{u,v}}+\delta_{\set{u,v}} + \err{\set{u,v}} + \mu(1 + \drift)\tau_{\set{u,v}}$.
\end{itemize}
\end{definition}

\begin{algorithm}[t!]
	\caption{Setting the rate of node $u$'s logical clock}
	\label{algo:rate}
	\If{the slow mode trigger is satisfied}
	{
		$\rate_u \coloneq 1$\;
	}
	\ElseIf{the fast mode trigger is satisfied}
	{
		$\rate_u \coloneq 1 + \mu$\;
	}
	\Else(\tcp*[h]{Neither slow nor fast mode trigger are satisfied; check for
	max estimate triggers}) {
		\If{$\lc{u} = \mc{u}$}
		{
			$\rate_u \coloneq 1$\;
		}
		\ElseIf{$\lc{u} \leq \mc{u} - \iota$}
		{
			$\rate_u \coloneq 1 + \mu$\;
		}
	}
\end{algorithm}

The fast and slow mode triggers are disjoint (as we will prove later), and since
both are closed regions, they are strictly separated from each other: there are
some states that satisfy neither condition. In these in-between regions, nodes
choose their mode based on the max-estimate trigger, which ensures that $\MC$
is satisfied.

\begin{definition}[Max Estimate Triggers]
Node $u$ satisfies the \emph{fast max estimate trigger} at time $t$ if the slow
mode trigger is not satisfied and $L_u(t)\leq M_u(t)-\iota$. It satisfies the
\emph{slow max estimate trigger} at time $t$ if the fast mode trigger is not
satisfied and $L_u(t)=M_u(t)$.
\end{definition}

The code implementing these triggers is shown in Listing~\ref{algo:rate}.

%%% Local Variables: 
%%% mode: latex
%%% TeX-master: "main"
%%% End: 

\section{Analysis}\label{sec:analysis}

In this section, we analyze the algorithm described in Section
\ref{sec:algorithm} and bound its worst-case global and dynamic
gradient skew.

\subsection{Basic Properties}

We begin with some basic properties which were stated informally in
Section~\ref{sec:algorithm}. Essentially, in this subsection we show that the
algorithm behaves ``as intended,'' which is the foundation for our subsequent
reasoning about skews. The first property states that the neighbor set $\N_u^s$
is a subset of $\N_u^{s-1}$ for all $s\ge 1$ at all times.
\begin{lemma}
  For all $u \in V$, at all times $t\geq 0$ we have $\N_u(t) = \N_u^0(t)$
  and $N_u^{s}(t)\subseteq N_u^{s-1}(t)$ for all $s\geq 1$.
  \label{lemma:Nchain}
\end{lemma}
\begin{proof}
  At time $0$, the neighbor sets are initialized to
  $\N_u^s(0)=\N_u(0)$ for all $s\geq 0$. Further, for every edge
  $e=\set{u,v}$ of node $u$, the update times $T_s^{e}$ are reached in
  order $s = 1, 2, \ldots$, and at each such time, we only add to
  $\N_u^s$ nodes that already belong to $\N_u^1, \ldots,
  \N_u^{s-1}$. Therefore node additions preserve the property.

  Nodes are only removed from neighbor sets in Line~\ref{line:remove}
  of Listing~\ref{algo:events}. As a node is removed from all
  the neighbor sets, also node removals preserve the property claimed
  by the lemma. Formally, the claim of the lemma therefore follows by
  induction on node $u$'s discrete transitions.
\end{proof}

In Section~\ref{sec:algoverview} we introduced the fast and slow mode
conditions (\FC\ and \SC), and claimed that the algorithm implements
these conditions; now we prove this claim. Our goal is to show that
when the inaccuracy of the estimates is taken into account, the
fast and the slow mode triggers hold whenever the fast and the slow
mode conditions apply, respectively. Likewise, the max estimate condition \MC
is satisfied by the algorithm.

\begin{lemma}\label{lemma:algo_correct}
Algorithm~\Aopt satisfies the fast and slow mode conditions,
as well as the max estimate condition.
\end{lemma}
\begin{proof}
Let us start with $\FC$. Suppose that the antecedent of $\FC$ holds
at node $u$: that is, there is some $s \in \nat$ such that for some
$w \in \N_u^s(t)$ we have
\begin{equation}
\lc{w}(t) - \lc{u}(t) \geq s \cdot \kappa_{\set{u,w}},
\label{eq:fc_leading}
\end{equation}
and for all $v \in \N_u^s(t)$ we have
\begin{equation}
\lc{u}(t) - \lc{v}(t) \leq s \cdot \kappa_{\set{u,v}} + 2\mu\tup_{\set{u,v}}.
\label{eq:fc_trailing}
\end{equation}
The estimate $\lce{u}{w}(t)$ that node $u$ has for node $w$
satisfies $\lc{w}(t) \leq \lce{u}{w}(t) + \err{\set{u,w}}$, and
combined with~\eqref{eq:fc_leading} we obtain
\begin{equation*}
\lce{u}{w}(t) - \lc{u}(t) \geq s \cdot \kappa_{\set{u,w}} - \err{\set{u,w}}.
\end{equation*}
Similarly, for all $v \in \N_u^s(t)$ we have $\lc{v}(t) \leq
\lce{u}{v}(t) + \err{\set{u,v}}$, so from~\eqref{eq:fc_trailing},
\begin{equation*}
\lc{u}(t) - \lce{u}{v}(t) \leq s \cdot \kappa_{\set{u,v}} + 2\mu\tup_{\set{u,v}} + \err{\set{u,v}}.
\end{equation*}
Hence the fast mode trigger is satisfied and node $u$ is in fast
mode.

Now consider $\SC$. Define $\delta \coloneq \min_{e\in
E}\{\delta_e\}>0$, and suppose that for this value of $\delta$ the
antecedent of $\SC$ holds at node $u$: there is some $s \in \nat$
such that for some $w \in \N_u^s(t)$ we have
\begin{equation}
\lc{u}(t) - \lc{w}(t) \geq \left( s + \frac{1}{2} \right) \kappa_{\set{u,w}} - \delta,
\label{eq:sc_trailing}
\end{equation}
and for all $v \in \N_u^s(t)$ we have
\begin{equation}
\lc{v}(t) - \lc{u}(t) \leq \left( s + \frac{1}{2} \right) \kappa_{\set{u,v}} + \delta + \mu(1 + \drift)\tau_{\set{u,v}}.
\label{eq:sc_leading}
\end{equation}
Now we use the other direction of the estimate accuracy guarantee:
for all $w \in \N_u^s(t)$ we have $\lc{w}(t) \geq \lce{u}{w}(t) -
\err{\set{u,w}}$. In particular, since $\delta \leq
\delta_{\set{u,w}}$ for any $w \in \N_u^s(t)$,
from~\eqref{eq:sc_trailing} we obtain
\begin{equation*}
\lc{u}(t) - \lce{u}{w}(t) \geq \left( s + \frac{1}{2} \right) \kappa_{\set{u,w}} - \delta - \err{\set{u,w}}
\geq \left( s + \frac{1}{2} \right) \kappa_{\set{u,w}} - \delta_{\set{u,w}} - \err{\set{u,w}}.
\end{equation*}
Moreover, when replacing $w$ with $v$ we get that $\lc{v}(t) \geq \lce{u}{v}(t) -
\err{\set{u,v}}$, which together with~\eqref{eq:sc_leading} yields
\begin{align*}
\lce{u}{v}(t) - \lc{u}(t)
&\leq
\left( s + \frac{1}{2} \right) \kappa_{\set{u,v}} + \delta + \err{\set{u,v}} + \mu(1 + \drift)\tau_{\set{u,v}}
\\
&
\leq \left( s + \frac{1}{2} \right) \kappa_{\set{u,v}} + \delta_{\set{u,v}} + \err{\set{u,v}} + \mu(1 + \drift)\tau_{\set{u,v}}.
\end{align*}
Therefore the slow mode trigger is satisfied, and node $u$ is in slow mode.

It remains to show that the algorithm also satisfies the max
estimate condition \MC. Suppose first that \MC requires the node to be in slow
mode. Then $M_u(t)=\lc{u}(t)$ and the fast mode trigger cannot be satisfied, as
there is no neighbor $v\in \Nb{u}$ with $\lc{v}>\lc{u}$. Thus, $u$ is in slow
mode either because the slow mode trigger applies or because neither the slow
nor fast mode trigger applies and $\lc{u}(t)=M_u(t)$, cf.~Listing
\ref{algo:rate}.

Similarly, if \MC requires the node to be in fast mode, $M_u(t)\geq
\lc{u}(t)-\iota$ and there is no neighbor $v\in \Nb{u}$ with $\lc{v}>\lc{u}$.
Hence, the slow mode trigger is not satisfied and $u$ will be in fast mode.
Hence, the max estimate condition \MC is satisfied.
\end{proof}

Next we show that the slow and fast mode triggers are never in conflict. While
this statement is not needed for deriving the guarantees of the algorithm, we
could not actually \emph{implement} the algorithm if it did not hold.

\begin{lemma}
For all $u \in V$, the slow and fast mode triggers are never satisfied at the same time.
\label{lemma:disjoint}
\end{lemma}
\begin{proof}
Suppose for the sake of contradiction that for some node $u \in V$ and time $t$, the fast
mode trigger is satisfied for an integer $s$ and the slow mode trigger is
satisfied for an integer $s'$. We consider two cases.
\begin{enumerate}[I.]
  \item $s \leq s'$. Due to Lemma~\ref{lemma:Nchain}, we have that
  $\N_u^{s'}(t) \subseteq \N_u^s(t)$ in this case.
  Because the slow mode trigger is satisfied for $s'$, there is a node $w \in
  \N_u^{s'}$ such that
  \begin{equation}
\lc{u}(t) - \lce{u}{w}(t) \geq \left( s' + \frac{1}{2} \right)\kappa_{\set{u,w}} - \delta_{\set{u,w}} - \err{\set{u,w}}.
\label{eq:sc1}
\end{equation}
However, since $w \in \N_u^{s'}(t) \subseteq \N_u^s(t)$, the second part of the fast mode condition applies to $w$, and it states that
\begin{equation}
\lc{u}(t) - \lce{u}{w}(t) \leq s \cdot \kappa_{\set{u,w}} + \err{\set{u,w}} + 2\mu\tup_{\set{u,w}}.
\label{eq:fc1}
\end{equation}
Combining~\eqref{eq:sc1} and~\eqref{eq:fc1} yields
\begin{align*}
\left( s' + \frac{1}{2} \right) \kappa_{\set{u,w}}
-
\delta_{\set{u,w}} - \err{\set{u,w}}
&\leq
s \cdot \kappa + \err{\set{u,w}} + 2\mu\tup_{\set{u,w}}
\\
&\leq
s' \cdot \kappa + \err{\set{u,w}} + 2\mu\tup_{\set{u,w}}.
\end{align*}
By re-arranging the terms, we obtain
\begin{equation}
\kappa_{\set{u,w}} \leq 4\err{\set{u,w}} + 4\mu\tup_{\set{u,w}} + 2\delta_{\set{u,w}}.
\label{eq:kappa1}
\end{equation}
However, recall that $\delta_e$ is chosen in the range
$(0, \kappa_e / 2 - 2\err{e} - 2\mu\tup_e)$
for each edge $e$.
%so $2\delta_{\set{u,w}} \leq \kappa_{\set{u,w}} - 4\err{\set{u,w}} - 8\mu\tup_{\set{u,w}}$.
Therefore $4\err{\set{u,w}} + 4\mu\tup_{\set{u,w}} + 2\delta_{\set{u,w}} < \kappa_{\set{u,w}}$, contradicting~\eqref{eq:kappa1}.

\item $s > s'$, that is, $s \geq s' + 1$. In this case
Lemma~\ref{lemma:Nchain} states that $\N_u^s(t) \subseteq \N_u^{s'}(t)$.
Because the fast mode trigger is satisfied for $s$, there is some node $w \in
\N_u^s(t)$ such that
\begin{equation}
\lce{u}{w}(t) - \lc{u}(t) \geq s \cdot \kappa_{\set{u,w}} - \err{\set{u,w}}.
\label{eq:fc2}
\end{equation}
Because $w \in \N_u^s(t) \subseteq \N_u^{s'}(t)$, the second part of the slow mode trigger also applies to $w$:
\begin{equation}
\lce{u}{w}(t) - \lc{u}(t)  \leq \left( s' + \frac{1}{2} \right)\kappa_{\set{u,w}} + \delta_{\set{u,w}} - \err{\set{u,w}}
+ \mu(1 + \drift)\tau_{\set{u,w}}
.
\label{eq:sc2}
\end{equation}
As before, we combine~\eqref{eq:fc2} and~\eqref{eq:sc2} and obtain
\begin{align*}
\left( s' + \frac{1}{2} \right)\kappa_{\set{u,w}} + \delta_{\set{u,w}} + \err{\set{u,w}}
+ \mu(1 + \drift)\tau_{\set{u,w}}
&\geq
s \cdot \kappa_{\set{u,w}} - \err{\set{u,w}}
\\
& \geq (s' + 1)\kappa_{\set{u,w}} - \err{\set{u,w}}.
\end{align*}
Re-arranging the terms yields
\begin{equation}
\kappa_{\set{u,w}} \leq 4\err{\set{u,w}} + 2\delta_{\set{u,w}} + 2\mu(1 +
\drift)\tau_{\set{u,w}}.
\label{eq:kappa2}
\end{equation}
However, since $\drift < 1$, we have $2\mu(1 + \drift)\tau_{\set{u,w}} < 4\mu\tau_{\set{u,w}}$,
and since $\delta_{\set{u,w}} < \kappa_{\set{u,w}}/2 - 2\err{\set{u,w}} - 2\mu\tup_{\set{u,w}}$,
we have $4\err{\set{u,w}} + 2\delta_{\set{u,w}} + 2\mu(1 + \drift)\tau_{\set{u,w}} < \kappa_{\set{u,w}}$,
contradicting~\eqref{eq:kappa2}.
\end{enumerate}

\end{proof}

Finally, we characterize the behavior of the neighbor coordination
mechanism introduced in Section~\ref{sec:algformal} to ensure that
nodes add each other as neighbors in a roughly symmetric manner. Intuitively,
node $u$ is trying to insert (or has inserted) edge $\{u,v\}$ at time $t$ if and
only if its variables $T_s^{\{v,w\}}\neq \bot$ at time $t$. If the edge
disappears in the view of one of the nodes, at the latest $\tup_{\{v,w\}}$
time later this happens also for the other, and both stop considering the edge
for evaluating their mode until it is reinserted. However, for the sake of the
analysis, we want to discard the initial time period in which it is possible
that not both $u$ and $v$ agree on the values $T_s^{\{v,w\}}$; this time period
is irrelevant, since neither node will actually add the edge to one of its
neighbor sets for $s\in \nat$ before this agreement is established. This is
captured by the following definition and lemma. We state them for the general
case that $\Ge_u(t)$ is not constant, as this does not affect the proof of the
lemma and will be of use in Section~\ref{sec:dynamic_est}.
\begin{definition}
For any node $u$ and any time $t$, we define the set
$\mathbf{T}_u^{\set{u,v}}(t):=\set{T_0^{\set{u,v}},\dots,T_\infty^{\set{u,v}}}$,
where $T_s^{\set{u,v}}$, $s\in \nat_0$ refers to the state of $u$'s
respective variable at time $t$. We simply write
$\mathbf{T}_u^{\set{u,v}}(t)=\bot$ if $u$ never computed these variables or set
them to $\bot$ according to Algorithm~\ref{algo:events}.

If $\mathbf{T}_u^{\set{u,v}}(t)\neq \bot$, denote by $T_{u,0}^{\set{u,v}}(t)$
the time $T_0^{\set{u,v}}\in \mathbf{T}_u^{\set{u,v}}(t)$; otherwise,
$T_{u,0}^{\set{u,v}}(t):=\infty$. We define the Boolean variable
\begin{equation*}
\mathcal{A}_u^{\set{u,v}}(t):=\left(
\mathbf{T}_u^{\set{u,v}}(t)\neq \bot \wedge t\geq 
\min\left\{T_{u,0}^{\set{u,v}}(t),T_{v,0}^{\set{u,v}}(t)\right\}\right).
\end{equation*}
\end{definition}

\begin{lemma}\label{lemma:insertion}
  For every (potential) edge $\set{u,v}$ and for all $t\geq 0$, the following three statements are true:
  \begin{align*}
    \mathrm{(I)} && 
    \big(\mathcal{A}_u^{\set{u,v}}(t) \land
    \mathcal{A}_v^{\set{u,v}}(t)\big) & \Longrightarrow
    \mathbf{T}_u^{\set{u,v}}(t) = \mathbf{T}_v^{\set{u,v}}(t),\\
    \mathrm{(II)} && \big(\mathcal{A}_u^{\set{u,v}}(t) \land
    \lnot\mathcal{A}_v^{\set{u,v}}(t)\big) &\Longrightarrow
    \big(\lnot\mathcal{A}_u^{\set{u,v}}(t+\tup_{\set{u,v}}) \land
    \lnot\mathcal{A}_v^{\set{u,v}}(t+\tup_{\set{u,v}})\big),\\
    \mathrm{(III)} && \big(\lnot\mathcal{A}_u^{\set{u,v}}(t) \land
    \mathcal{A}_v^{\set{u,v}}(t)\big) &\Longrightarrow
    \big(\lnot\mathcal{A}_u^{\set{u,v}}(t+\tup_{\set{u,v}}) \land \lnot\mathcal{A}_v^{\set{u,v}}(t+\tup_{\set{u,v}})\big).
  \end{align*}
\end{lemma}
\begin{proof}

Without loss of generality, assume that node $u$ is the leader of edge
$\set{u,v}$. If node $u$ has never discovered neighbor $v$ before time $t$, we
have $\mathcal{A}_u^{\set{u,v}}(t)=\false$ and because $v$ can only start
inserting the edge after receiving a message from $u$, we also have
$\mathcal{A}_v^{\set{u,v}}(t)=\false$. Hence, we can assume that there has been
a time $t'\leq t$ when $u$ discovered neighbor $v$. Let $\underline{t}$ be the
last such time before $t$.

\smallskip
\noindent\textbf{Case 1: \boldmath$\mathcal{A}_u^{\set{u,v}}(t)=\true$.}
Thus, $v\in N_u(t')$ for all $t'\in [\underline{t},t]$. Consequently, $u\in
N_v(t')$ for all $t'\in [\underline{t}+\tup_{\{u,v\}},t-\tup_{\{u,v\}}]$. 
Abbreviate
\begin{equation*}
\Delta:=\frac{(1+\rho)(1+\mu)(\delay_{\{u,v\}}+\tup_{\{u,v\}})}{1-\drift}
+\tup_{\{u,v\}}.
\end{equation*}
At the logical time $t_S$ when
$L_u(t_S)-L_u(\underline{t})=(1+\rho)(1+\mu)\Delta$, $u$ sent a message to $v$,
informing it about the times $\mathbf{T}_u^{\set{u,v}}(t)$, by communicating the
global skew estimate $\Ge_{\set{u,v}}=\Ge_u(t_S)$ and
$L_{\mathit{ins}}=\lc{u}(t_S) + \Ge_u(t_S) +
(1+\drift)(1+\mu)\delay_{\set{u,v}}$. Based on these parameters, in the call to
$\mathrm{computeInsertionTimes}$, the logical time $T_{u,0}^{\set{u,v}}$ is set
to a value of at least $L_{\mathit{ins}}$. Let $t_0$ be the time such that
$\min\set{\lc{u}(t_0),\lc{v}(t_0)}=T_{u,0}^{\set{u,v}}$. We claim that $t_0\geq
t_S+\delay_{\{u,v\}}$. To see this, consider any node $w\in \{u,v\}$ and bound
\begin{align*}
L_w(t_S+\delay_{\{u,v\}})&\leq
L_w(t_S)+(1+\rho)(1+\mu)\delay_{\{u,v\}}\\
&\leq L_v(t_S)+\G(t_S)+(1+\rho)(1+\mu)\delay_{\{u,v\}}\\
&\leq L_v(t_S)+\Ge_u(t_S)+(1+\rho)(1+\mu)\delay_{\{u,v\}}\\
&=  L_{\mathit{ins}}\\
&\leq T_{u,0}^{\set{u,v}}.
\end{align*}
In particular, the definition of $\underline{t}$ implies that $t\geq t_0\geq
t_S+\delay_{\{v,w\}}$.

\smallskip
\noindent\textbf{Case 1a:
\boldmath$\mathcal{A}_v^{\set{u,v}}(t)=\true$.}
In this case we need to show that $\mathbf{T}_u^{\set{u,v}}(t) =
\mathbf{T}_v^{\set{u,v}}(t)$ in order to establish (I). Denote by $t_R$ the
maximal time in $[0,t]$ when $v$ received an
$\mathrm{insertedge}(\set{u,v},\cdot, \cdot)$ message from $u$; such a time must
exist, as otherwise $\mathcal{A}_v^{\set{u,v}}(t)=\false$. Because $u$ waits for
at least $\Delta$ time after an edge has formed (from $u$'s perspective) before
sending a message, $v$ cannot receive any other
$\mathrm{insertedge}(\set{u,v},\cdot, \cdot)$ message from $u$ during
$[t_S-\Delta+\delay_{\{u,v\}},t_R]\supseteq
[t_R-\delay_{\{u,v\}}-\tup_{\{u,v\}},t_R]$. We claim that $u\in N_v(t')$ for all
$t'\in [t_R-\delay_{\{u,v\}}-\tup_{\{u,v\}},t]$. Otherwise, $v$ would satisfy
$T_s^{\{u,v\}}=\bot$ for all $s\in \{0,1,\ldots\}$ at some time $t'\in
[t_R-\delay_{\{u,v\}}-\tup_{\{u,v\}},t]$, and if there is only such a $t'<
t_R$, it would have ignored the message received at time $t_R$ because
\begin{equation*}
L_v(t_R-\delay_{\{u,v\}}-\tup_{\{u,v\}})\geq
L_v(t_R)-(1+\rho)(1+\mu)(\delay_{\{u,v\}}+\tup_{\{u,v\}}).
\end{equation*}
We conclude that the message
sent by $u$ at time $t_S$ is received by $v$; therefore, $t_R\in
[t_S,t_S+\delay_{\{u,v\}}]$ is actually the time when this message is received.
(I) now follows because $t_S+\delay_{\{u,v\}}\leq t_0\leq t$, $v$ computes
$\mathbf{T}_v^{\set{u,v}}(t_R) = \mathbf{T}_u^{\set{u,v}}(t_S)$ upon reception
of the message, and $v$ does not change these variables again until time $t$.

\smallskip
\noindent\textbf{Case 1b:
\boldmath$\mathcal{A}_v^{\set{u,v}}(t)=\false$.}
In this case we need to show that
$\big(\lnot\mathcal{A}_u^{\set{u,v}}(t+\tup_{\set{u,v}}) \land
\lnot\mathcal{A}_v^{\set{u,v}}(t+\tup_{\set{u,v}})\big)$ in order to establish
(II). We claim that there is a time $t'\in
[t_S-\Delta+\tup_{\{u,v\}},t]$ so that $u\notin N_v(t')$. Otherwise, we had
$u\in N_v$ throughout $[t_S-\Delta+\tup_{\{u,v\}},t]$ and $v$ would receive and
not discard the message by $u$, as
\begin{equation*}
\lc{v}(t_R)\geq \lc{v}(t_S)\geq
\lc{v}\left(t_S-\Delta+\tup_{\{u,v\}}\right)+
(1+\drift)(1+\mu)(\delay_{\{u,v\}}+\tup_{\{u,v\}}).
\end{equation*}
However, this would entail that $\mathcal{A}_v^{\set{u,v}}(t)=\true$, so indeed
such a time $t'$ must exist.

Denote by $t''\in [t'-\tup_{\{u,v\}},t'+\tup_{\{u,v\}}]\subseteq
[t_S-\Delta,t+\tup_{\{u,v\}}]$ a time so that $v\notin N_u(t'')$; by the
communication model, such a time exists. In fact, we have that $t''> t_S$, as
the edge must be continuously present from the perspective of $u$ for
$(1+\drift)(1+\mu)\Delta$ local time, i.e., since at least $\Delta$ real time,
before it sends a message. Therefore, $u$ will set $T_s^{\{u,v\}}:=\bot$ for all
$s\in \{0,1,\ldots\}$ at time $t''$. While the definition of $\underline{t}$
admits that $u$ may observe the reappearance of the edge at a time larger than
$t$, $u$ will not recompute the values $T_s^{\{u,v\}}$ or send another message
to $v$ by time $t+\tup_{\{u,v\}}$. This implies that also $v$ does not recompute
its values $T_s^{\{u,v\}}$ during $[t'',t+\tup_{\{u,v\}}]$, and it follows that
$\big(\lnot\mathcal{A}_u^{\set{u,v}}(t+\tup_{\set{u,v}}) \land
\lnot\mathcal{A}_v^{\set{u,v}}(t+\tup_{\set{u,v}})\big)$, as claimed.

\smallskip
\noindent\textbf{Case 2: \boldmath$\mathcal{A}_u^{\set{u,v}}(t)=\false
\wedge \mathcal{A}_v^{\set{u,v}}(t)=\true$.} In this case we need to show that
$\big(\lnot\mathcal{A}_u^{\set{u,v}}(t+\tup_{\set{u,v}}) \land
\lnot\mathcal{A}_v^{\set{u,v}}(t+\tup_{\set{u,v}})\big)$ in order to establish
(III). Denote by $t_R$ the latest time before $t$ when $v$ received an
$\mathrm{insertedge}(\set{u,v},\cdot, \cdot)$ message from $u$; as
$\mathcal{A}_v^{\set{u,v}}(t)=\true$, such a time must exist. Denote by $t_S\in
[t_R-\delay_{\{u,v\}},t_R]$ the time when it was sent. Note that
$\mathcal{A}_v^{\set{u,v}}(t)=\true$ also implies that $v$ neither discarded the
message, i.e., it recomputed the times $T_s^{\{u,v\}}$ at time $t_R$, nor did it
set $T_s^{\{u,v\}}:=\bot$ for any $s$ during $[t_R,t]$. Hence, $u\in N_v(t')$
for all $t'\in [t_R-\delay_{\{u,v\}}-\tup_{\{u,v\}},t]$. We infer that $v\in
N_u(t')$ for all $t'\in [t_R-\delay_{\{u,v\}},t-\tup_{\{u,v\}}]\subseteq
[t_S,t-\tup_{\{u,v\}}]$. This implies that $u$ cannot detect the reappearance of
the edge during this interval and thus will not call
$\mathrm{computeInsertionTimes}$ during $(t_S,t+\tup_{\{u,v\}}]$; it follows
that $\mathcal{A}_u^{\set{u,v}}(t+\tup_{\set{u,v}})=\false$.

To see that also $\mathcal{A}_v^{\set{u,v}}(t+\tup_{\set{u,v}})=\false$, note
that since $\mathcal{A}_u^{\set{u,v}}(t)=\false$, it must hold that $v\notin
N_u(t')$ for some $t'\in [t_S,t]$. Therefore, $u\notin N_v(t'')$ for some
$t''\in [t_S-\tup_{\{u,v\}},t+\tup_{\{u,v\}}]\subseteq
[t_R-\delay_{\{u,v\}}-\tup_{\{u,v\}},t+\tup_{\{u,v\}}]$. As $u\in N_v(t')$
for all $t'\in [t_R-\delay_{\{u,v\}}-\tup_{\{u,v\}},t]$, we obtain that
$t''\in (t,t+\tup_{\{u,v\}}]\subseteq (t_R,t+\tup_{\{u,v\}}]$.
Consequently, $v$ resets $T_s^{\{u,v\}}:=\bot$ for all $s\in \{0,1,\ldots\}$ at
time $t''\in (t_R,t+\tup_{\{u,v\}}]$, yielding that
$\mathcal{A}_v^{\set{u,v}}(t+\tup_{\set{u,v}})=\false$. This proves
Statement~(III), concluding the proof.
\end{proof}

\subsection{The Global Skew}
\label{sec:global}

Like its predecessors in~\cite{kuhn09b,lenzen10}, our algorithm
achieves an asymptotically optimal global skew. In static networks of diameter
$D$ where each message has an uncertainty $U$ in its transit time,
the best possible global skew guarantee is $\Theta((\rho+U)
D)$~\cite{kuhn09b,lenzen10}, and the dynamic estimate diameter satisfies
$D(t)\in \Omega((\rho+U)D)$ at all times. Let $D:=\max_t D(t)$ be the maximal
network uncertainty of a given execution of our algorithm. In the following, we
show that the algorithm always guarantees a global skew of $\BO(D)$. In fact, we
show the following stronger statement:

\begin{theorem}\label{thm:global}
Let $\iota$ be defined as in Definition \ref{def:MC}.
\begin{enumerate}[I.]
  \item On any dynamic graph executing $\Aopt$, at any time the global skew
  increases at rate at most $2\rho$.
  \item On any dynamic graph executing $\Aopt$, at any time $t$ when the global
  skew exceeds $D(t)+\iota$, it decreases at rate at least
  $\mu(1-\drift)-2\drift > 0$.
\end{enumerate}
\end{theorem}
\begin{proof}
It suffices to show that (i) any node with the largest clock value throughout
the network is in slow mode and (ii) whenever the global skew exceeds
$D(t)+\iota$, any node with the smallest clock value in the network is in fast
mode. The theorem then follows because any clock in slow mode is at most
$2\rho$ faster than any other clock and
\[
(1+\mu)(1-\drift) - (1+\drift) =
\mu(1-\drift) - 2\drift \stackrel{\eqref{eq:sigma}}{>}0.
\]

Let $u$ and $v$ be nodes with the largest and smallest logical clock values at
an arbitrary time $t$, respectively. By Inequalities \eqref{eq:mc_upper} and
\eqref{eq:mc_self}, we have $\mc{u}(t)=\lc{u}(t)$. Provided that ${\cal
G}(t)=L_u(t)-L_v(t)> D(t)+\iota$, from Inequality \eqref{eq:mc_lower} it follows
that $\mc{v}(t)\geq\lc{u}(t)-D(t) > \lc{v}(t)-\iota$. Therefore, due to the max
estimate condition \MC, we conclude that Statements (i) and (ii) are true,
yielding the claims of the theorem.
\hspace*{\fill}\end{proof}

Note that Statement~$II.$ of the theorem implies that the global skew is
self-stabilizing in the sense that it is reduced at an asymptotically optimal
rate of $\mu(1-\drift) - 2\drift\in \Omega(\mu)$ when it exceeds the best
possible guarantee. One could add a simple consistency check mechanism to the
algorithm that forces logical clocks to be instantaneously set to close values
(i.e., violate the progress bound of $(1+\drift)(1+\mu)$) whenever a global skew
exceeding a certain multiple of $D(t)$ is detected. While (together with
self-stabilizing implementations of the algorithms rules and neighbor sets) this
would ensure quick stabilization from arbitrarily corrupted states, such
behavior might be undesired if $D(t)$ decreases rapidly in a fault-free
execution, as clock values would change too rapidly; this could also cause
violations of the gradient skew bound.

\subsection{Analysis of the Gradient Skew}

\subsubsection{Preliminary Definitions and Statements}

As described in Section \ref{sec:algstrategy}, the gradient skew
requirement states that paths of certain lengths cannot have an
average skew that exceeds a certain bound. In particular, for every
positive integer $s$, there is a length $C_s$ such that paths $p$ of
length $\kappa_p\geq C_s$ have an average skew of at most
$\BO(s\kappa_e)$ per edge $e$, where the values $C_s$ are
exponentially decreasing in $s$. Since in the process of the
analysis we will have to use different such sequences $C_s$,
we do not explicitly define the values here, but will work with an
abstract \emph{gradient sequence} as defined in the following
Definition \ref{def:sequence} for most of the analysis.

\begin{definition}[Gradient Sequences]\label{def:sequence}
  A \emph{gradient sequence} is a non-increasing sequence of values
  $C=\{C_s\}_{s\in \nat}$.
\end{definition}

The specific sequences that will later be used to prove the gradient
skew properties of the algorithm roughly look as follows. At all times
$t$, we have $C_1\geq 2\G(t)$. As a consequence, the gradient skew
property for level $1$ will follow directly from the bound on the
global skew. Further, for most levels $s$, we have
$C_{s+1}=C_s/\sigma$ such that also for short paths, we obtain a
sufficiently strong requirement on the skew.

We have seen that the different values of $s$ correspond to different
average skew bounds. Throughout the proof, we will mostly make
arguments for a particular such level $s$. In the following, we define
the sets of edges and paths used when arguing about level $s$.

\begin{definition}[Level-$s$ Edge Set]\label{def:sedges}
  For all $s\in\nat$, we define
  \begin{equation*}
    E^s(t) := \set{\set{u,v}\in\binom{V}{2}\,\Big|\, v\in N_u^s(t) \wedge u\in
      N_v^s(t)}.
  \end{equation*}
\end{definition}

\begin{definition}[Level-$s$ Paths]\label{def:paths}
  We define for all $u_0\in V$, $s\in \nat$, and all times $t$ the set
  of \emph{level-$s$ paths starting at node $u_0$ at time $t$} to be
  \begin{equation*}
    P_{u_0}^s(t)\coloneq \{p=(u_0,\ldots,u_k)\,|\, \forall
    i\in
    \{0,\ldots,k-1\}:~\{u_i,u_{i+1}\}\in E^s(t)\}.
  \end{equation*}
\end{definition}

For convenience of notation, we also define reversed and concatenated paths.
\begin{definition}[Path Reversal and Concatenation]
  Given the path $p=(u_0,\ldots,u_k)$, the corresponding
  \emph{reversed path} is $\bar{p}:=(u_k,\ldots,u_0)$. Given two paths
  $p=(u_0,\ldots,u_k)$ and $q=(u_k,\ldots,u_{\ell})$, their
  \emph{concatenation} is $p\circ
  q:=(u_0,\ldots,u_k,\ldots,u_{\ell})$.
\end{definition}
Note that for all levels $s$ and all times $t$, reversal and
concatenation of level-$s$ paths again yield level-$s$ paths.

In addition, we define two notions of ``weighted skew'', essentially
capturing how far away certain paths are from the level-$s$ skew
bound. The multiplicative factors in the two conditions
correspond to the factors in the fast and slow mode condition \FC\ and
\SC.

\begin{definition}\label{def:xi_local}
  For all paths $p = (u, \ldots, v)$, all $s\in \nat$, and all times
  $t$ we define
  \begin{equation*}
    \xi_p^s(t) \coloneq \lc{u}(t) - \lc{v}(t) - s\kappa_p.
  \end{equation*}
  We further define for all $u\in V$ that
  \begin{equation*}
    \Xi_u^s(t) \coloneq \max_{p\in P_u^s(t)}\{\xi_p^s(t)\}.
  \end{equation*}
\end{definition}

\begin{definition}\label{def:psi_local}
  For all paths $p = (u, \ldots, v)$, all $s\in \nat$, and all times $t$ we
  define
  \begin{equation*}
    \psi_p^s(t) \coloneq \lc{v}(t) - \lc{u}(t) - \left(s+\frac{1}{2}\right)\kappa_p.
  \end{equation*}
  Furthermore, we set for all $u\in V$
  \begin{equation*}
    \Psi_u^s(t) \coloneq \max_{p\in P_u^s(t)} \{\psi_p^s(t)\}.
  \end{equation*}
\end{definition}

A gradient sequence $C$ defines what clock skew is allowed on
different paths.  Rather than directly defining a gradient skew
requirement as described in Section \ref{sec:algstrategy}, we define a
condition that is based on the value introduced in Definition
\ref{def:psi_local}. If this requirement is satisfied, we say that the system is
\emph{legal}. Legality is formally defined as follows.

\begin{definition}[Legality]\label{def:legality}
  Given a weighted, dynamic graph $G$ and a gradient sequence $C$, for
  each $s\in \nat$ the system is \emph{$(C,s)$-legal at time $t$ and
    node $u\in V$}, if and only if it holds that
  \begin{equation*}
    \Psi_u^s(t) < \frac{C_s}{2}
  \end{equation*}
  The system is \emph{$C$-legal at $t$ and $u$} if it is $(C,s)$-legal
  for all $s\in \nat$ at node $u$ and time $t$.
\end{definition}

The legality definition can be used to derive an upper bound on the
clock skew between any two nodes $u$ and $v$.

\begin{lemma}\label{lemma:clocks}
  Assume that for some $s\in \nat$ and a path $p=(u,\dots,v)\in
  P_u^s(t)$, the system is $(C,s)$-legal at nodes $u$ and $v$ at time
  $t$. Then, $|\lc{v}(t)-\lc{u}(t)|< (s+1/2)\kappa_p + C_s/2$.
\end{lemma}
\begin{proof}
  Since the system is $(C,s)$-legal at $u$ at time $t$, we have that
  \begin{equation*}
    \lc{v}(t)-\lc{u}(t)-\left(s+\frac{1}{2}\right)\kappa_p=\psi_p^s(t)\leq
    \Psi_u^s(t)<\frac{C_s}{2}.
  \end{equation*}
  We therefore get that $\lc{v}(t)-\lc{u}(t)< (s+1/2)\kappa_p +
  C_s/2$. From legality at $v$ and because $p\in P_u^s(t)$ implies
  that the reversed path is in $P_v^s(t)$, we get the same upper bound
  on $\lc{u}(t)-\lc{v}(t)$.
\end{proof}

We now show a few simple properties that follow from the various
definitions and the basic structure of the algorithm.

\begin{lemma}\label{lemma:simple_stuff}
  The following statements hold at all times $t$ and all nodes $u\in
  V$.
  \begin{itemize}
  \item [(i)] $\forall s, s'\in \nat$, $s'\leq s$: $P_u^s(t)\subseteq
    P_u^{s'}(t)$.
  \item [(ii)] $\forall s, s'\in \nat$, $s'\leq s$: $\Psi_u^s(t)\leq
    \Psi_u^{s'}(t)$.
  \item [(iii)] If for some $s\in \nat$ the system is $(C,s)$-legal and
    $C_s=C_{s+1}$, then the system is also $(C,s+1)$-legal.
  \end{itemize}
\end{lemma}
\begin{proof}
  Statement~(i) is a direct consequence of Lemma~\ref{lemma:Nchain}
  and Definition~\ref{def:paths}. Statement~(ii) follows from
  Definition \ref{def:psi_local} together with Statement~(i). Finally,
  Statement~(iii) follows from Statement~(ii) together with Definition
  \ref{def:legality}.
\end{proof}

\subsubsection{Stabilization Condition and Convergence to Small Skews}

In order to prove the claimed bound on the stabilization time, we
require a \emph{stabilization condition}, which depends on the
$s$-legality of the system for a certain $s \in \nat$. For ease of
presentation, for every $s\in\nat$, we define a parameter that
will be used in the definition of the stabilization condition, as well
as throughout the remainder of the proof.
\begin{equation}\label{eq:stabint}
  \forall s\geq 2: \stabint{s} := \frac{C_{s-1}}{(1+\drift)\mu}.
\end{equation}

\begin{definition}[Stabilization Condition]
  \label{def:stabcondition}
  For a node $u\in V$, a gradient sequence $C$, an integer $s>1$, and
  a time $t$, we say that $u$ satisfies the
  $(C,s)$-\emph{stabilization condition} at time $t$ if and only if
  \begin{itemize}
  \item[\bf (S0)] For all $t'\in[t-\stabint{2},t]$, we have
    $C_1\geq 2\G(t')$.
  \item[\bf (S1)] For all $s'\in\set{2,\dots,s-1}$ and all nodes $v\in
    V$ for which $|\lc{u}(t)-\lc{v}(t)|\leq s'C_{s'-1} + \big(2\drift
    + \mu(1 + \drift)\big)\stabint{s'}$, the system is
    $(C,s')$-legal at node $v$ at all times in $[t-\stabint{s'},t]$.
    \label{stabcond:smallskew}
  \item[\bf (S2)] We have:
    \begin{equation*}
      \forall T_s\in \Tset_s: |\lc{u}(t)- T_s| \geq
      (1+\mu)(1+\drift)\stabint{s} + sC_{s-1}
    \end{equation*}
    \label{stabcond:time}
  \end{itemize}
\end{definition}

A simple observation is that the stabilization condition directly implies
level-$1$ legality.
\begin{lemma}\label{lemma:1legality}
  Let $t\geq 0$ be a time and assume that condition (S0) of the
  stabilization condition holds at some node $u\in V$ at some time
  $t'\in[t,t+\stabint{2}]$. Then the system is $(C,1)$-legal
  at all nodes $v\in V$ at time $t$.
\end{lemma}
\begin{proof}
  Consider arbitrary nodes $v,w\in V$. For any path
  $p=(v,\dots,w)$, we have
  \[
  \psi_p^1(t) = \lc{w}(t) -\lc{v}(t) - \frac{3}{2}\kappa_p <
  |\lc{w}(t)-\lc{v}(t)| \leq \G(t) \stackrel{(S0)}{\leq} \frac{C_1}{2}.
  \]
  Given Definition \ref{def:legality}, this implies that $v$ is
  $(C,1)$-legal at time $t$.
\end{proof}
This can be seen as an induction anchor, starting from which increasingly
stronger bounds can be established for higher levels $s$. The core theorem of
our analysis, stated next, provides the matching induction step.
\begin{theorem}\label{thm:catch_up}
  Fix a level $s > 1$, a node $u \in V$ and an interval $[t^-, t^+]$.
  Let $\Lambda_s \coloneq C_{s-1} / ( 2(1 - \drift)\mu)$, and suppose
  that for each $t \in [t^-, t^+]$ and for each path $(u, \ldots, v)
  \in P_u^s(t)$, if $\kappa_{(u, \ldots, v)} \leq C_{s-1}$, then the
  endpoint $v$ satisfies the $(C,s)$-stabilization condition at time
  $t$. In this case, for all
  \begin{equation*}
    t \in \left[t^-+\Lambda_s+ 2\stabint{s}, t^+\right]
  \end{equation*}
  we have
  \begin{equation*}
    \Psi^s_u(t) < 2\drift \Lambda_s 
    = \frac{\drift C_{s-1}}{(1-\drift)\mu}=\frac{C_{s-1}}{2\sigma}.
  \end{equation*}
\end{theorem}
Proving this theorem is technically challenging, but self-contained. Therefore,
we postpone its proof to Section~\ref{sec:convergence}, in order to show how it
is used to establish the desired gradient skew first.

\subsubsection{Derivation of Skew Bounds}

We now have all the necessary technical tools to prove the gradient skew bound
of our algorithm. As pointed out earlier, our algorithm has some
self-stabilization~\cite{dijkstra74} properties in the following sense: Even if
we start the algorithm from a configuration in which no non-trivial gradient
skew bound holds, the system adapts and converges to a state in which the
desired gradient skew bound holds. Edge insertion exploits this by adding edges
level by level, every time waiting until the respective level and all higher
levels have stabilized to small skews again. This entails that, at any given
time, for level-$s$ paths all but at most one level contributes a reduction by
factor $\sigma$ to the skew bound given by level-$s$ legality. Taking into
account the connection between logical clock values and real time, this
motivates the following definition.

\begin{definition}[Gradient Sequences for Static Global
Skew Estimate]\label{def:levelSeq}
Set
\begin{equation}\label{eq:def_delta}
  \Delta_s:= \left(1-\frac{1}{2^{s-1}}\right)\instime(\Ge)
  +\left(\frac{5(1+\mu)}{2(1-\rho)\mu}+2s\right)C_{s-1}.
\end{equation}
Fix a time $t$ and let $L(t)$ be maximal satisfying that $L(t)\leq L_u(t)$ for
all $u\in V$ and that $L(t)/\instime(\Ge)\in \mathbb{Z}$.
Then, for $s\in \mathbb{N}$, $u\in V$, and a parameter $\Gu$, define
\begin{equation*}
  C^{(t,u)}_s:=\left\{\begin{matrix}
  \frac{2\Gu}{\sigma^{s-1}}& \mbox{if }L_u(t)\geq
  L(t)+\Delta_s\\
  \frac{2\Gu}{\sigma^{\max\{s-2,0\}}}& \mbox{else.}\\
  \end{matrix}\right.
\end{equation*}
\end{definition}

The next lemma shows that the system is legal at all nodes and times with
respect to the above gradient sequences, granted that $\Gu$ is an upper bound on
the global skew and there is some initial time period of length
$\Lambda_2+3\Theta_2\in O(\Gu)$ during which the system is legal. For
simplicity, we also make the assumption that $\sigma \geq 3$. However, by
modifying the insertion times computed in Algorithm~\ref{algo:insertiontimes}
(using a base depending on $\sigma$), one can handle any $\sigma>1$; the
downside is that $\instime(\Ge)$ goes to infinity as $\sigma$ approaches $1$.
\begin{lemma}\label{lemma:legality_static}
Assume that $\sigma\geq 3$ and that the system is $C^{(t,u)}$-legal at all
times $t\in [t_0,t_0+\Lambda_2+3\Theta_2]$ and nodes $u$. If $\Gu\geq \G(t)$ for
all $t\geq t_0$, then the system is $C^{(t,u)}$-legal at all times $t\geq t_0$
and nodes~$u$.
\end{lemma}
\begin{proof}
Assume for contradiction that there is a node $u\in V$ and a minimal time
$\bar{t}> t_0+\Lambda_2+3\Theta_2$ violating
$C^{(\bar{t},u)}$-legality.\footnote{As logical clocks are continuous, the set
of times when the system is not legal at some node is closed. Thus, in any
execution in which the claim of the lemma does not hold, such a minimal time
exists.} Because $C^{(t,u)}_1\geq 2\Gu\geq 2\G(t)$ for any $t\geq t_0$, (S0) is
satisfied at time $t$ at any node $v\in V$. In particular, this is true for node
$u$ and time $\bar{t}$, yielding that the system is $(C^{(\bar{t},u)},1)$-legal
at node $u$ and time $\bar{t}$ by Lemma~\ref{lemma:1legality}. Let $\bar{s}>1$
be the minimal level on which legality is violated at node $u$ and time
$\bar{t}$. By Lemma~\ref{lemma:simple_stuff}, this implies that
$C_{\bar{s}}^{(\bar{t},u)}\neq C_{\bar{s}-1}^{(\bar{t},u)}$. Therefore,
$L_u(\bar{t})\notin
[L(\bar{t})+\Delta_{\bar{s}-1},L(\bar{t})+\Delta_{\bar{s}})$.

We show that the preconditions of Theorem~\ref{thm:catch_up} are satisfied at
node $u$ for level $\bar{s}$, times $t^+=\bar{t}$ and
$t^-=\bar{t}-\Lambda_s-2\Theta_{\bar{s}}$, and the gradient sequence given by
\begin{equation*}
C_s:=
\begin{cases}
\frac{2\Gu}{\sigma^{s-1}} & \mbox{if }
C^{(\bar{t},u)}_{\bar{s}}=\frac{2\Gu}{\sigma^{\bar{s}-1}}\\
\frac{2\Gu}{\sigma^{\max\{s-2,0\}}} & \mbox{if }
C^{(\bar{t},u)}_{\bar{s}}=\frac{2\Gu}{\sigma^{\bar{s}-2}}.
\end{cases}
\end{equation*}
This leads to the desired contradiction, as then
\begin{equation*}
\Psi^{\bar{s}}_u(\bar{t})<\frac{C_{\bar{s}-1}}{2\sigma}
=\frac{C^{(\bar{t},u)}_{\bar{s}}}{2}.
\end{equation*}

Hence, it remains to show that for any $t\in
[\bar{t}-\Lambda_s-2\Theta_{\bar{s}},\bar{t}]$ and any path $(u,\ldots,v)\in
P^{\bar{s}}_u(t)$ with $\kappa_{(u,\ldots,v)}\leq C^{(\bar{t},u)}_{\bar{s}-1}$,
$v$ satisfies the $(C^{(\bar{t},u)},\bar{s})$-stabilization condition at time
$t$. As $C_1= C_1^{(\bar{t},u)}\geq 2\G(t)$, (S0) holds at all times.

For (S1) and (S2), we will make a case distinction. However, both cases will use
the following observation. The system is $(C^{(t,v)},\bar{s})$-legal at
node $v$ and $(C^{(t,u)},\bar{s})$-legal at node $u$ at any time $t\in
[t_0,\bar{t})$ by the minimality of $\bar{t}$. As $C_{\bar{s}-1}\leq
\min\{C^{(t,u)}_{\bar{s}},C^{(t,v)}_{\bar{s}}\}$, this entails
$(C,{\bar{s}-1})$-legality at both nodes and we can apply
Lemma~\ref{lemma:clocks} to bound
\begin{equation}\label{eq:diff_u_v}
|L_v(t)-L_u(t)|\leq 
\left(\bar{s}-\frac{1}{2}\right)\kappa_{(u,\ldots,v)}+\frac{C_{\bar{s}-1}}{2}
\leq  \bar{s}C_{\bar{s}-1}.
\end{equation}
As clocks are continuous, this bound also applies for $t=\bar{t}$. We now
proceed to the case distinction.

\textbf{Case 1:} $L(\bar{t})\geq L(\bar{t})+\Delta_{\bar{s}}$, i.e.,
$C^{(\bar{t},u)}_{\bar{s}}=\frac{2\Gu}{\sigma^{\bar{s}-1}}$. Then
\begin{eqnarray}
L_v(t)&\sr{eq:diff_u_v}{\geq}& L_u(t)-\bar{s}C_{\bar{s}-1}\nonumber\\
&\geq & L_u(\bar{t})-(1+\rho)(1+\mu)(\bar{t}-t)-\bar{s}C_{\bar{s}-1}\nonumber\\
&\geq & L_u(\bar{t})
-(1+\rho)(1+\mu)\left(\frac{1}{2(1-\rho)\mu}+\frac{2}{(1+\rho)\mu}+\bar{s}\right)
C_{\bar{s}-1}\nonumber\\
&= & 
L_u(\bar{t})-\left(\frac{(5-3\rho)(1+\mu)}{2(1-\rho)\mu}+\bar{s}\right)C_{\bar{s}-1}\nonumber\\
&\geq & 
L(\bar{t})+\Delta_{\bar{s}}-\left(\frac{(5-3\rho)(1+\mu)}{2(1-\rho)\mu}+\bar{s}\right)C_{\bar{s}-1}
\label{eq:L_v_large}\\
&>&
L(\bar{t})+\left(1-\frac{1}{2^{\bar{s}-1}}\right)\instime(\Ge)
-\left(\frac{(1+\rho)(1+\mu)}{(1-\rho)\mu}+\bar{s}\right)C_{\bar{s}-1}
\label{eq:L_v_large}\nonumber\\
&=& L(\bar{t})+\left(1-\frac{1}{2^{\bar{s}-1}}\right)\instime(\Ge)
+(1+\mu)(1+\rho)\Theta_{\bar{s}}+\bar{s}C_{\bar{s}-1}.\nonumber
\end{eqnarray}
Hence, if (S2) is violated for some $T_{\bar{s}}\in \mathbb{T}_{\bar{s}}$, then
$T_{\bar{s}}>L(\bar{t})+\left(1-\frac{1}{2^{\bar{s}-1}}\right)\instime(\Ge)$.
The smallest possible such $T_{\bar{s}}$ is
$L(\bar{t})+\left(2-\frac{1}{2^{\bar{s}-1}}\right)\instime(\Ge)$. However, by
definition of $L(\bar{t})$,
\begin{eqnarray*}
L_v(t)&\leq & L_v(\bar{t})\\
&\leq & \min_{w\in V}\{L_w(\bar{t})\} + \G(\bar{t})\\
&< & L(\bar{t})+\instime(\Ge) + \Ge\\
&\leq & L(\bar{t})+\frac{3\instime(\Ge)}{2}
-\left(\frac{1+\mu}{\mu}+2\right)2\Gu\\
&= & L(\bar{t})+\frac{3\instime(\Ge)}{2}
-(1+\mu)(1+\rho)\Theta_2 - 2C_1\\
&\stackrel{\sigma\geq 3}{\leq} &
L(\bar{t})+\left(2-\frac{1}{2^{\bar{s}-1}}\right)\instime(\Ge)
-(1+\mu)(1+\rho)\Theta_{\bar{s}} - \bar{s}C_{\bar{s}-1},
\end{eqnarray*}
as $\bar{s}\geq 2$. Therefore, (S2) is satisfied.

Next, consider any $s\in \{2,\ldots,\bar{s}-1\}$ and node $w\in V$ with
$|L_v(t)-L_w(t)|\leq sC_{s-1}+(2\rho +\mu(1+\rho))\Theta_s)$. We have that
\begin{eqnarray*}
L_w(t)&\geq & L_v(t)-sC_{s-1}-(2\rho+\mu(1+\rho))\Theta_s\\
&\sr{eq:L_v_large}{\geq} &
L(\bar{t})+\Delta_{\bar{s}}-\left(\frac{(5-3\rho)(1+\mu)}{2(1-\rho)\mu}+\bar{s}\right)C_{\bar{s}-1}
-\left(\frac{2\rho}{(1+\rho)\mu}+s+1\right)C_{s-1}\\
&> &
L(\bar{t})+\Delta_{\bar{s}}-\left(\frac{(5-3\rho)(1+\mu)}{2(1-\rho)\mu}+\bar{s}\right)C_{\bar{s}-1}
-\left(\frac{1}{\sigma}+s+1\right)C_{s-1}\\
&= &
L(\bar{t})+\Delta_{\bar{s}}-\left(\frac{(5-3\rho)(1+\mu)}{2(1-\rho)\mu}+\bar{s}+1\right)C_{\bar{s}-1}
-(s+1)C_{s-1}\\
&\stackrel{\sigma\geq 3}{\geq} &
L(\bar{t})+\Delta_{s+1}-\left(\frac{(5-3\rho)(1+\mu)}{2(1-\rho)\mu}+s+2\right)C_s
-(s+1)C_{s-1}\\
&>& L(\bar{t})+\left(1-\frac{1}{2^s}\right)\instime(\Ge)-(s+1)C_{s-1}\\
&=&
L(\bar{t})+\Delta_s+\frac{\instime(\Ge)}{2^s}
-\left(\frac{5(1+\mu)}{2(1-\rho)\mu}+3s+1\right)\frac{2\Gu}{\sigma^{s-2}}\\
&\stackrel{\sigma\geq 3}{\geq} &
L(\bar{t})+\Delta_s+\frac{1+\mu}{\mu}\cdot C_{s-1}\\
&=& L(\bar{t})+\Delta_s+(1+\mu)(1+\rho)\Theta_s.
\end{eqnarray*}
For any time $t'\in [t-\Theta_s]$, this yields that
\begin{equation*}
L_w(t')\geq L_w(t)-(1+\mu)(1+\rho)\Theta_s \geq L(\bar{t})+\Delta_s.
\end{equation*}
We conclude that $C_s^{(t',w)}=C_s$, which by minimality of $\bar{t}$ and
$\bar{s}$ implies that the system is $(C,s)$-legal at node $w$ and time $t'$,
i.e., (S1) is satisfied. Therefore, all all preconditions of
Theorem~\ref{thm:catch_up} are satisfied and Case 1 leads to a contradiction.

\textbf{Case 2:} $L_u(\bar{t})<L(\bar{t})+\Delta_{\bar{s}-1}$, i.e.,
$C^{(\bar{t},u)}_{\bar{s}}=\frac{2\Gu}{\sigma^{\bar{s}-2}}$. Then, for
any $s\in \{2,\ldots,\bar{s}-1\}$, time $t$, and node $w\in V$, we have that
\begin{equation*}
C_s= \frac{2\Gu}{\sigma^{\max\{s-2,0\}}}\leq C_s^{(t,w)}.
\end{equation*}
Thus, as $\bar{t}\geq t_0+\Lambda_s+3\Theta_2\geq
t_0+\Lambda_s+2\Theta_{\bar{s}}+\Theta_s$, by minimality of $\bar{t}$ (S1) is
satisfied for all $t\in [\bar{t}-\Lambda_s-2\Theta_{\bar{s}},\bar{t}]$.

Concerning (S2), note that
\begin{eqnarray*}
\frac{\instime(\Ge)}{2^{\bar{s}-1}} &\stackrel{\sigma\geq 3}{>}&
\left(\frac{5(1+\mu)}{2(1-\rho)\mu}+2(\bar{s}-1)\right)\frac{2\Gu}{\sigma^{\max\{\bar{s}-3,0\}}}
+\left(\frac{1+\mu}{\mu}+2\bar{s}\right)\frac{2\Gu}{\sigma^{\bar{s}-2}}\\
&=& \left(\frac{5(1+\mu)}{2(1-\rho)\mu}+2(\bar{s}-1)\right)C_{\bar{s}-2}
+(1+\mu)(1+\rho)\Theta_{\bar{s}}+2\bar{s}C_{\bar{s}-1},
\end{eqnarray*}
where we used that the ratio between left- and right-hand side is minimized for
$\bar{s}=3$ (as opposed to minimal $\bar{s}=2$ like in other places).

Together with \eqref{eq:diff_u_v}, this yields that
\begin{eqnarray*}
L_v(t)&\sr{eq:diff_u_v}{\leq} & L_u(t)+\bar{s}C_{\bar{s}-1}\\
&< & L(\bar{t})+\Delta_{\bar{s}-1}+\bar{s}C_{\bar{s}-1}\\
& =& L(\bar{t})+\left(1-\frac{1}{2^{\bar{s}-2}}\right)\instime(\Ge)
+\left(\frac{5(1+\mu)}{2(1-\rho)\mu}+2(\bar{s}-1)\right)C_{\bar{s}-2}+\bar{s}C_{\bar{s}-1}\\
&\leq & L(\bar{t})+\left(1-\frac{1}{2^{\bar{s}-1}}\right)\instime(\Ge)
-(1+\mu)(1+\rho)\Theta_{\bar{s}}-\bar{s}C_{\bar{s}-1}.
\end{eqnarray*}
As $L(\bar{t})+\left(1-\frac{1}{2^{\bar{s}-1}}\right)\instime(\Ge)$
is the smallest possible logical time that is at least $L(\bar{t})$ and in
$\mathbb{T}_{\bar{s}}$, (S2) cannot be violated for any $T_{\bar{s}}\geq
L(\bar{t})$.
On the other hand,
\begin{eqnarray*}
L_v(t)&\geq & L_v(\bar{t})-(1+\mu)(1+\rho)(\bar{t}-t)\\
&\geq & L(\bar{t})-(1+\mu)(1+\rho)(\Lambda_s+2\Theta_{\bar{s}})\\
&> & L(\bar{t})-\frac{5(1+\mu)}{2(1-\rho)\mu}\cdot
\frac{2\Gu}{\sigma^{\bar{s}-2}}\\
&\stackrel{\sigma\geq 3}{\geq} & L(\bar{t})-\frac{\instime(\Ge)}{2^{\bar{s}-1}}
+\left(\frac{1+\mu}{\mu}+\bar{s}\right)\frac{2\Gu}{\sigma^{\bar{s}-2}}\\
&=& L(\bar{t})-\frac{\instime(\Ge)}{2^{\bar{s}-1}}
+(1+\mu)(1+\rho)\Theta_{\bar{s}}+\bar{s}C_{\bar{s}-1}.
\end{eqnarray*}
Hence, (S2) can also not be violated with respect to any
$T_{\bar{s}}<L(\bar{t})$.
Thus, again all preconditions of Theorem~\ref{thm:catch_up} are met and the proof is
complete.
\end{proof}
In order to obtain a bound on the gradient skew, it remains to show that there
is an interval $[t_0,t_0+\Lambda_2+3\Theta_2]$ during which the system is
$C^{(t,u)}$-legal at all times $t\geq t_0$ and nodes~$u$. If $\Gu\approx \Ge$,
the above reasoning is sufficient for this purpose.
\begin{corollary}\label{coro:legality_static}
If $\Gu\geq \G(t)$ at all times $t$ and $\sigma\geq 3$, the system is
$C^{(t,u)}$-legal at all times $t \geq \frac{2\instime(\Ge)}{1-\rho}$ and
nodes~$u$.
\end{corollary}
\begin{proof}
Observe that within $\frac{\instime}{1-\rho}$ time, it must occur that
$L(t)=\min_{u\in V}\{L_u(t)\}$. Shifting the time axis, it is hence sufficient
to show the claim for all times $t\geq \frac{\instime(\Ge)}{1-\rho}$ under the
assumption that $L(0)=\min_{u\in V}\{L_u(0)\}$.

We modify $C^{(t,u)}$ to ``switch on'' its guarantees level by level. That is,
we consider the gradient sequence
\begin{equation*}
\bar{C}_s^{(t,u)}:=\begin{cases}
C_1^{(t,u)} & \mbox{if }s=1\\
C_s^{(t,u)} & \mbox{if }s>1 \mbox{ and }L_u(t)\geq L(0)+\Delta_s\\
\bar{C}_{s-1}^{(t,u)} & \mbox{if }s>1 \mbox{ and }L_u(t)<L(0)+\Delta_s.
\end{cases}
\end{equation*}
The proof is now analogous to the one of Lemma~\ref{lemma:legality_static}, with
the exception that if $L_u(\bar{t})<L(0)+\Delta_{\bar{s}}$, then
$C_{\bar{s}}^{(t,u)}=C_{\bar{s}-1}^{(t,u)}$ and Lemma~\ref{lemma:simple_stuff}
immediately yields a contradiction. As the case $L_u(\bar{t})\geq
L(0)+\Delta_{\bar{s}}$ only uses bounds for levels $s\in\{2,\ldots,\bar{s}-1\}$
at logical times of at least $L(0)+\Delta_s$ (and on level $1$ at times $t\geq
0$), the weaker guarantees offered by $\bar{C}$ are sufficient. Because during
the time interval $[0,0+\Lambda_2+3\Theta_2]$ we have that
$\bar{C}_s^(t,u)=\bar{C}_1^{(t,u)}=2\Gu$, the prerequisite that the system is
legal at all nodes during this interval is satisfied.

Finally, observe that $L\left(\frac{\instime}{1-\rho}\right)\geq L(0)+\instime$,
implying that $\bar{C}^{(t,u)}=C^{(t,u)}$ at times $t\geq
\frac{\instime}{1-\rho}$.
\end{proof}

We can now infer that the system achieves a gradient skew based on $\Gu$,
an upper bound on the global skew that holds at all times, and all edges that
have been present for $O\left(\frac{\Ge}{\mu}\right)$ time, where $\Ge$ is the
a priori upper bound on the global skew that is known to the algorithm.
\begin{theorem}\label{thm:gradient_slow}
Suppose that $\sigma\geq 3$ and $\Gu\geq \G(t)$ for all times $t$. Denote by
$G_{\instime}(t)$ the graph on nodes $V$ with all edges $\{u,v\}$ that have been
continuously present for at least
$\frac{2\instime+\Ge+(1+\rho)(1+\mu)\delay_{\{u,v\}}}{1-\rho}\in
O(\frac{\Ge}{\mu})$ time. If path $p=(u,\ldots,v)$ exists in $G_{\instime}$ at
time $t$, it holds that
\begin{equation*}
|L_v(t)-L_w(t)|=\left(\log_{\sigma}\frac{\Gu}{\kappa_p}+O(1)\right)\kappa_p.
\end{equation*}
\end{theorem}
\begin{proof}
As $\Gu\geq \G(t)$ for all $t$, we can apply
Corollary~\ref{coro:legality_static}. This shows that for all times $t\geq
\frac{2\instime}{1-\rho}$, the system is $C^{(t,u)}$-legal at each node $u$. In
particular, it is legal w.r.t.\ the gradient sequence given by
$C_s=\frac{2\Gu}{\sigma^{\max\{s-2,0\}}}$. At smaller times, $G_{\instime}$
contains no edges. Hence, for any path $p=(u,\ldots,v)$ that exists in
$G_{\instime}$ at such a time $t$ and any $s\in \nat$, Lemma~\ref{lemma:clocks}
yields that
\begin{equation*}
|L_v(t)-L_w(t)|\leq \left(s+\frac{1}{2}\right)\kappa_p+\frac{C_s}{2}.
\end{equation*}
Choosing $s=2+\left\lceil\log_{\sigma}\frac{\Gu}{\kappa_p}\right\rceil$, we have
that $C_s\leq \kappa_s$ and the statement of the theorem follows.
\end{proof}

This theorem has two shortcomings. First, even if $\Gu\ll \Ge$, it does not
provide any guarantee at times $t\ll \frac{\Ge}{\mu}$. Second, if the global
skew is large and then decreases later, the gradient skew bound does not adapt.
In the next section, we address these points.

\subsection{Fast Stabilization in Case of Small Global Skew}
We will again rely on Lemma~\ref{lemma:legality_static}, but need to handle the
case that $\Gu\ll \Ge$ differently than in Corollary~\ref{coro:legality_static}.
First, we prove the precondition of the lemma under the assumption that for
$O(\frac{\Gu}{\mu})$ time no edge is inserted, and then we show that within
$O(\frac{\Gu}{\mu})$ time, such a ``silent'' period must occur.
\begin{lemma}\label{lemma:stability_from_silence}
Set $\Gamma:=\frac{(29+8\mu)\Gu}{2(1-\rho)\mu}\in O(\frac{\Gu}{\mu})$. Suppose
that $\sigma \geq 3$, that $\G(t)\leq \Gu$ for all $t\in [t',t'+\Gamma]$, and
that no node inserts an edge during this interval. Then the system is
$C^{(t,u)}$-legal at all times $t\in [t'+\Gamma-\Lambda_2-3\Theta_2,t'+\Gamma]$
and nodes~$u$.
\end{lemma}
\begin{proof}
Set $C_s:=\frac{2\Gu}{\sigma^{s-1}}$,
$t_2:=t'+\Lambda_2+\left(2+\frac{(1+\mu)(1+\rho)}{1-\rho}\right)\Theta_2+\frac{C_1}{1-\rho}$
and $t_s:=t_{s-1}+\Theta_{s-1}+\Lambda_s+2\Theta_s$ for $3\leq s\in \nat$. Consider the gradient sequence
\begin{equation*}
C_s^{(t)}:=\begin{cases}
2\Gu & \mbox{if }s=1\\
C_{s-1}^{(t)} & \mbox{if }t< t_s\\
C_s & \mbox{if }t\geq t_s.
\end{cases}
\end{equation*}
We claim that the system is legal with respect to this sequence at all nodes and
times $t\in [t',t'+\Gamma]$, which we show by induction on the level $s$.
For $s=1$, the claim is trivial. For step from $s-1$ to $s\geq 2$, note that, by
Lemma~\ref{lemma:simple_stuff}, the claim is immediate from $(s-1)$-legality for
all times $t<t_s$, so assume that $t\geq t_s$. We apply
Theorem~\ref{thm:catch_up} for level $s$ with $t^+=t$,
$t^-=t-\Lambda_s-2\Theta_s$, and the gradient sequence $C_s$, which will
complete the induction step.

It remains to establish the preconditions of the theorem. (S0) holds
because $t^-\geq t'+\Theta_2$. (S1) holds because for any time
$t''\geq t-\Lambda_s-2\Theta_s-\Theta_{s'}\geq t_{s-1}$ we have that
$C_{s-1}^{(t'')}=C_{s-1}$. Concerning (S3), by assumption there are no edge
insertions during $[t',t'+\Gamma]$. W.l.o.g., we may assume that
there are no edge insertions after time $t'+\Gamma$: future events
cannot influence past clock values, and we can always extend the execution such
that no further edges are inserted (e.g.\ by removing all edges at time
$t'+\Gamma$). Hence, the only remaining case is that there is a logical time
$T_s\in \mathbb{T}_s$ and a node $v$ such that $L_v(t')>T_s$, but
$L_v(t^-)<T_s+(1+\mu)(1+\rho)\Theta_s+C_{s-1}$. However, $t^-\geq
t_2-\Lambda_2-2\Theta_2\geq
\left(\frac{(1+\mu)(1+\rho)}{1-\rho}\right)\Theta_2+\frac{2C_1}{1-\rho}$ and
thus
\begin{equation*}
L_v(t^-)\geq L_v(t')+(1+\mu)(1+\rho)\Theta_2+2C_1
\stackrel{\sigma\geq 3/2}{>} T_s+(1+\mu)(1+\rho)\Theta_s+sC_{s-1},
\end{equation*}
showing (S2).

We conclude that the system is $C$-legal at all nodes and times $t\in
[t_{\infty},t'+\Gamma]$, where $t_{\infty}:=\lim_{s\to \infty} t_s$. As
\begin{eqnarray*}
t'+\Gamma-t_{\infty}&=& \Gamma-\frac{(1+\mu)(1+\rho)}{1-\rho}\cdot \Theta_2 +
\frac{C_1}{(1-\rho)}-\sum_{s=2}^{\infty}(\Lambda_s+3\Theta_s)\\
&=&\Gamma-\left(\frac{1+\mu}{(1-\rho)\mu}+1\right)C_1
-\frac{\sigma}{\sigma-1}
\cdot\left(\frac{C_1}{2(1-\rho)\mu}+\frac{C_1}{(1+\rho)\mu}\right)\\
&\stackrel{\sigma\geq 3}{>}& \frac{5\Gu}{(1-\rho)\mu}\\
&>& \frac{C_1}{2(1-\rho)\mu}+\frac{C_1}{(1+\rho)\mu}\\
&=& \Lambda_2+3\Theta_2.
\end{eqnarray*}
As $C_s\leq C_s^{(t,u)}$ for all times $t$ and nodes $u$, the claim of the
lemma follows.
\end{proof}

\begin{lemma}\label{lemma:silence}
Set $\Gamma':=\frac{(2+(1+\rho)\mu)(29+10\mu)\Gu}{(1-\rho)^2\mu}\in
O(\frac{\Gu}{\mu})$ and suppose that $\G(t)\leq \Gu\leq
\frac{20\Ge}{(1+\rho)(29+8\mu)}$ for all $t\in [t',t'+\Gamma']$. Then there is
some time $t''\in [t',t'+\Gamma'-\Gamma]$ such that no edges are inserted during
$t\in [t'',t''+\Gamma]$.
\end{lemma}
\begin{proof}
If the statement does not hold for $t''=t'$, let $t$ be the minimal time when an
edge insertion occurs during $[t',t'+\Gamma]$ at some node $u$. Denoting by $s$
the level on which the edge is inserted, we have that $L_u(t)=T_s$ for some
$T_s\in \mathbb{T}_s$. Observe that the next larger logical time at which some
edge could be inserted is $T_s+\frac{\instime}{2^s}$. For any node $v$, we have
that $L_v\left(t+\frac{\Gu}{1-\rho}\right)\geq L_v(t)+\Gu\geq T_s$ and
$L_v\left(t+\frac{\instime}{2^s(1+\rho)(1+\mu)}+\frac{\Gu}{(1+\rho)(1+\mu)}\right)\leq
T_s+\frac{\instime}{2^s}$. Hence, either the claim holds for
$t''=t+\frac{\Gu}{1-\rho}+\varepsilon$ (for some sufficiently small
$\varepsilon>0$) or
\begin{equation}
\frac{\instime}{2^s}\leq
\frac{(2+(1+\rho)\mu)\Gu}{1-\rho}+(1+\rho)(1+\mu)\Gamma.
\label{eq:beat_zI}
\end{equation}
In the latter case, we use that for any node $v$, we have that
\begin{equation*}
L_v\left(t+\frac{\instime}{(1-\rho)2^{s-1}}+\frac{\Gu}{1-\rho}\right)\geq
T_s+\frac{\instime}{2^{s-1}}=z \instime
\end{equation*}
for some $z\in \mathbb{Z}$. Setting $\bar{t}:=\max_{v\in
V}\{t_v\,|\,L_v(t_v)=z\instime\}$, we then have that
\begin{equation*}
L_v\left(\bar{t}-\frac{\Gu}{1-\rho}\right)\geq L_v(\bar{t})-\Gu \geq z\instime
\end{equation*}
and
\begin{equation*}
L_v\left(\bar{t}+\frac{\instime-2\Gu}{2(1+\rho)(1+\mu)}\right)\leq
L_v(\bar{t})+\frac{\instime}{2}-\Gu \leq
\left(z+\frac{1}{2}\right)\instime,
\end{equation*}
which due to
\begin{equation*}
\frac{\instime-2\Gu}{2(1+\rho)(1+\mu)}>
\frac{20\Ge}{2(1+\rho)(1-\rho)\mu}
\geq\frac{(29+8\mu)\Gu}{2(1-\rho)\mu}=\Gamma
\end{equation*}
yields that no edges are inserted during $[t'',t''+\Gamma]$ if we set
$t'':=\bar{t}+\varepsilon$ for sufficiently small $\varepsilon>0$. As
\begin{eqnarray*}
\bar{t}+\Gamma &\leq &
t+\frac{\instime}{(1-\rho)2^{s-1}}+\frac{\Gu}{1-\rho}+\Gamma \\
&\leq & t'+\frac{\instime}{(1-\rho)2^{s-1}}+\frac{\Gu}{1-\rho}+2\Gamma\\
&\sr{eq:beat_zI}{\leq} & t'+\frac{(2+(1+\rho)\mu)\Gu}{(1-\rho)^2}
+\frac{2(1+\rho)(1+\mu)\Gamma}{1-\rho}+\frac{\Gu}{1-\rho}+2\Gamma\\
&< &
\left(\frac{4+2(1+\rho)\mu}{1-\rho}\right)\left(\frac{\Gu}{1-\rho}+\Gamma\right)\\
&=& t'+\Gamma',
\end{eqnarray*}
this proves the claim of the lemma.
\end{proof}

Together, the above results yield the following theorem.
\begin{theorem}\label{thm:gradient_fast}
Suppose that $\sigma\geq 3$ and $\Gu\geq \G(t)$ for all times $t$. Set
$S:=\Gamma'$ if $\Gu \leq \frac{20\Ge}{(1+\rho)(29+8\mu)}$ and
$S:=\frac{2\instime}{1-\rho}$ otherwise, and let $s\in \nat$. At any time
$t\geq S\in O(\frac{\Gu}{\mu})$, any level-$s$ path $p=(u,\ldots,v)$ satisfies
that
\begin{equation*}
|L_v(t)-L_w(t)|=
\left(s+\frac{1}{2}\right)\kappa_p+\frac{\Gu}{\sigma^{\max\{s-2,0\}}}.
\end{equation*}
\end{theorem}
\begin{proof}
If $S\neq \Gamma'$, we can apply Corollary~\ref{coro:legality_static}.
Otherwise, we apply Lemma~\ref{lemma:silence} to show that the preconditions of
Lemma~\ref{lemma:stability_from_silence} are satisfied for a time $t'\leq
t-\Gamma$. This enabls us to apply Lemma~\ref{lemma:legality_static} for a time
$t_0\leq t$. In both cases, we have shown $C^{(t,u)}$-legality at all nodes
$u\in V$. As $C^{(t,u)}_s\leq \frac{\Gu}{\sigma^{\max\{s-2,0\}}}$ for all $t$,
$u$, and $s$, the claim now follows from Lemma~\ref{lemma:clocks}.
\end{proof}

In particular, if the global skew is bounded by $\Gu$ for $O(\frac{\Gu}{\mu})$
time, on fully inserted edges we have a stable gradient skew that depends on
$\Gu$ only.
\begin{corollary}
Denote by $G_{\infty}(t)$ the graph on nodes $V$ and with all edges that have
been inserted on all levels. Then for any path $(u,\ldots,v)$ in this graph we
have that
\begin{equation*}
|L_v(t)-L_w(t)|=\left(\log_{\sigma}\frac{\Gu}{\kappa_p}+O(1)\right)\kappa_p.
\end{equation*}
\end{corollary}

\subsection{Discussion}

Before proceeding to the remaining more technical sections, we briefly put the
obtained results in context.

\subsubsection*{Optimality}

Our algorithm is simultaneously optimal or asymptotically optimal in terms of
several parameters:

\paragraph{Global Skew.}
The global skew bound given by Theorem~\ref{thm:global} is optimal in the sense
that there are executions with dynamic diameter of $D$ in which a skew of $D$
cannot be avoided~\cite{kuhn09b,lenzen10}.

\paragraph{Clock Rates.}
The algorithm guarantees that no logical clock runs at rate smaller than
$1-\rho$ and the largest clock value increases at most at rate $1+\rho$, which
is clearly optimal. Moreover, the maximum logical clock rate is
$(1+\mu)(1+\rho)$. The above proofs assumed that $\sigma\geq 3$ and thus
$\mu\geq \frac{6\rho}{1-\rho}$. However, more careful reasoning would show that
any $\mu>\frac{2\rho}{1-\rho}$ is feasible (at the expense of diverging
insertion times as $\mu$ approaches this bound). This is optimal, as it is
necessary that $(1-\rho)(1+\mu)>1+\rho$ so that nodes with slow hardware clocks
can catch up to those with fast hardware clocks.

\paragraph{Gradient Skew.}
Provided that the global skew has been bounded by $\Gu$ for sufficiently
long, the stable gradient skew between nodes $u$ and $v$ connected by a path $p$
with $\kappa_p$ is
$\left(\log_{\sigma}\frac{\Gu}{\kappa_p}+O(1)\right)\kappa_p$, where
$\kappa_e=4(\err{e}+\mu \tup_e)$. The first term is matched by a lower bound of
$\err{p}\log_{\Theta(\sigma)}\frac{\Gu}{\err{p}}$ for the static
case~\cite{lenzen10}, i.e., if $\mu \tup_e \ll \err{e}$, this bound is optimal
up to factor $4+o(1)$.

The additional slack of $\mu \tup_e$ compensates for the amount by which logical
clock values can drift apart while only one endpoint is aware of the edge. It
seems plausible that such a term is needed, but no matching lower bound is
known. However, it is worth pointing out that $\err{e}\geq \mu d_e$, where $d_e$
is the time it takes for information on clock values to propagate along the edge
$e$. One can thus expect that in most systems $\tup_e\in O(\err{e})$.

\paragraph{Stabilization Time.}
For a given global skew bound $\Gu$, Theorem~\ref{thm:gradient_fast} shows that
within $O(\frac{\Gu}{\mu})$ time the gradient skew bound holds for all inserted
edges. A lower bound of $\Omega(D)$ for establishing the gradient skew bound is
given in Section~\ref{sec:lower}, where $D$ is the (current) diameter of the
graph. As pointed out above, there are executions in which a global skew of $D$
cannot be avoided; in fact a skew of $\Omega(D)$ can be hidden from the
algorithm entirely. Therefore, for $\mu\in \Theta(1)$, this bound would be
asymptotically optimal (recall that we assumed $\mu \in O(1)$; for larger
values of $\mu$, we would obtain a bound of $\Theta(\Gu)$ in
Theorem~\ref{thm:gradient_fast}). Note also that $O(\frac{\Gu}{\mu})$ time is
trivially necessary, as a newly inserted edge $e$ may exhibit a skew of $\Gu$,
regardless of its uncertainty $\err{e}$.

Unfortunately, our algorithm takes $\Theta(\frac{\Ge}{\mu}+\tup_e+\delay_e)$
time to \emph{insert} a newly discovered edge $e$,
cf.~Theorem~\ref{thm:gradient_slow}. The additive terms of $\tup_e$ and
$\delay_e$ are necessary for guaranteeing \emph{any} communication between its
endpoints, and therefore must be present either explicitly or via the
abstraction of the estimate layer (by making stronger assumptions on the
interface it provides). However, it may be the case that $\Ge \gg \Gu$,
especially since we assume that $\Ge$ is an upper bound on the global skew that
holds at all times. In Section~\ref{sec:dynamic_est}, we discuss how to insert
edges based on local, time-dependent global skew estimates $\Ge_u(t)$,
overcoming this issue.

The leading constants in Theorem~\ref{thm:gradient_slow} and
Theorem~\ref{thm:gradient_fast} are moderate. For example, if $\mu\leq
\frac{1}{100}$ and thus $\rho \leq \frac{\mu}{100}$ ($\rho\leq 10^{-5}$ for a
typical quartz oscillator), then
\begin{equation*}
\frac{2\instime+\Ge}{1-\rho}<
\frac{43\Ge}{\mu}
\end{equation*}
and
\begin{equation*}
\Gamma'<\frac{58\Gu}{\mu}.
\end{equation*}
We remark that in the interest of a more streamlined presentation, we did not
attempt to optimize constants. We conjecture that the leading constants can be
reduced to less than $10$ without introducing additional techniques.

\subsection*{Comparison to Simultaneous Insertion on all Levels}

In~\cite{lenzen10}, we presented a simpler insertion strategy and analysis that
inserts edges on all levels right after discovering them. The idea is to
initially give the edge a very large weight $\kappa$, so that the gradient skew
bound is trivially satisfied due to the global skew bound, and then reduce the
weight exponentially until the final value is reached. Adding some additional
slack to the (final) $\kappa$, it is then shown that the gradient property holds
on all existing edges w.r.t.\ the time-dependent values of $\kappa$.

Compared to the solution present so far, the gradient skew bound we achieve here
is slightly better (as no additional slack is needed). More importantly,
the insertion time is asymptotically optimal in case $\Ge=\Gu$, in contrast to a
multiplicative overhead of $\Theta(\frac{\Ge}{\min_e\{\err{e}\}})$ for the
simpler strategy. Given that the simpler strategy has a leading constant of at
least $24$ in its insertion time bound, the above bounds compare favorably with
this approach.

A big advantage of the simpler approach, however, is that it can readily use
local and time-dependent estimates $\Ge_u(t)$ for edge insertion, thus
performing better in case of a large gap between $\Ge$ and $\Gu$. As mentioned
earlier, we adapt our insertion strategy to account for such local estimates in
Section~\ref{sec:dynamic_est}. Unfortunately, this results in a very large
leading constant, meaning that the simpler insertion strategy performs better in
practice, as $\frac{\Ge}{\min_e\{\err{e}\}}$ is extremely unlikely to exceed
$10^3$. We thus prove that an asymptotically optimal insertion time can be
achieved even for time- and node-dependent estimates of the global skew, but
leave whether they can be used to obtain small insertion times in practice open.

%%% Local Variables:
%%% mode: latex
%%% TeX-master: "main"
%%% End:

\section{Proving Convergence}\label{sec:convergence}

On an abstract level, the proof of Theorem~\ref{thm:catch_up} follows the same
strategy as in the static case. Consider the potentials $\Psi^s:=\max_{u\in
V}\{\Psi_u^s\}$ and $\Xi^s:=\max_{u\in V}\{\Xi_u^s\}$, $s\in \nat$. By the
design of the algorithm, in the static case $\Psi^s$ grows at any time at rate
at most $2\drift=1+\drift-(1-\drift)$ because \emph{any} node $v$ that is the
endpoint with larger clock value of a path maximizing $\Psi^s$ must be in slow
mode. The reason is that \SC not being satisfied at the node on level $s$
implied the existence of a neighbor $w$ of $v$ that is more than $s\kappa$
ahead, yielding a path with even larger $\psi^s$-value. On the other hand, the
algorithm is much more aggressive with respect to $\Xi^s$: Whenever $\Xi^s>0$,
it is shown that \emph{any} node $v$ that is the endpoint with the smaller clock
value of a path maximizing $\Xi^s$ will be in fast mode. This holds as \FC not
being satisfied for $s$ at $v$ would permit to extend the path by a neighbor $w$
of $v$ to a new path with larger $\xi^s$-value. These two statements play
together to ensure that the maximal values of $\Psi^s$, i.e., the sequence
$C_s$, decreases exponentially in $s$, which is the essence of the statement of
Theorem~\ref{thm:catch_up}.

\subsection{Existence of Relevant Paths}

Given the level-$s$ stabilization condition, we can conclude that paths with
sufficiently large $\xi$- or $\psi$-values to be of interest cannot have a large
weight, and infer that the system must have been $s'$-legal for all $1<s'<s$ at
one of the path's endpoints for a long time.
\begin{lemma}\label{lemma:shortpaths}
  Let $s>s'\geq 1$ be two integers and let $t\geq 0$ be a time. Assume
  that the $(C,s)$-stabilization condition holds at node $u$ at some
  time in $[t,t+\stabint{s'}]$ (if $s'>1$) or $[t,t+\stabint{2}]$ (if $s'=1$),
  and let $p=(u,\dots,v)$ be a path for which $p\in P_u^{s'}(t)$. Moreover,
  suppose that $\xi_p^{s}(t)\geq 0$ or $\psi_p^{s}(t)\geq
  -\kappa_p/2$. Then
  \[
  \kappa_p\leq C_{s'}
  \]
  and the system is $s'$-legal at $v$ and time $t$.
\end{lemma}
\begin{proof}
  Because $\psi_p^s(t)=\xi_{\bar{p}}^s(t)-\kappa_p/2$, we
  assume, w.l.o.g., that $\xi_p^s(t)\geq 0$. It holds that
  \begin{equation*}
    0\leq \xi_p^s(t)=L_u(t)-L_v(t)-s\kappa_p\leq \G(t)-\frac{\kappa_p}{2}
    \leq \frac{C_1-\kappa_p}{2},
  \end{equation*}
  implying $\kappa_p\leq C_1$. The last inequality follows from (S0)
  because the stabilization condition at node $u$ holds at some time
  $t'\in[t,t+\stabint{2}]$. Together with Lemma \ref{lemma:1legality},
  this implies the claim of the lemma for $s'=1$ (and arbitrary
  $s>1$).

  For $s'\geq2$, we can therefore define $s''$ to be the largest
  integer in $\set{1,\dots,s'}$ for which both $\kappa_p\leq C_{s''}$
  and $v$ is $(C,s'')$-legal at time $t$. We want to show that
  $s''=s'$, so assume for the sake of contradiction that $s''<s'$.
  From the precondition that $p\in P_u^{s'}(t)$, we also get that
  $p\in P_u^{s''}(t)$ (Statement~(i) of
  Lemma~\ref{lemma:simple_stuff}). Since we have $(C,s'')$-legality
  at nodes $u$ and $v$ at time $t$, we can use Lemma~\ref{lemma:clocks}
  to obtain an upper bound on the skew between $u$
  and $v$:
  \[ |\lc{u}(t) - \lc{v}(t)| \leq \left(s''+\frac{1}{2}\right)\kappa_p +
  \frac{C_{s''}}{2} \leq (s''+1)C_{s''}.
  \]
  We know that the $(C,s)$-stabilization condition holds at node $u$
  at some time $t'\in[t,t+\stabint{s'}]$. Because logical clocks
  progress at a rate between $1-\drift$ and $(1+\mu)(1+\drift)$, we
  get
  \begin{align*}
    |\lc{u}(t') - \lc{v}(t')| &\leq
    |\lc{u}(t) - \lc{v}(t)| +
    ((1+\mu)(1+\drift)-(1-\drift))\stabint{s'}\\
    &\leq
    (s''+1)C_{s''} + (2\drift + \mu(1+\drift))\stabint{s'}\\
    &\leq
    (s''+1)C_{s''} + (2\drift + \mu(1+\drift)))\stabint{s''+1}.
  \end{align*}
  From Requirement (S1) of the $(C,s)$-stabilization condition at
  node $u$ at time $t'$ for level $s''+1\leq s'<s$, we therefore get
  that the system is $(C,s''+1)$-legal at node $v$ and time $t$.

  As $p\in P_u^{s''}(t)$ also implies that $\bar{p}\in P_v^{s''}(t)$,
  $(s''+1)$-legality at node $v$ and time $t$ implies that
  \begin{equation*}
    \frac{C_{s''+1}}{2} > \Psi_v^{s''+1}(t) \geq
    \psi_{\bar{p}}^{s''+1}(t) =
    \psi_{\bar{p}}^{s}(t) + (s-(s''+1))\kappa_p \geq
    \psi_{\bar{p}}^{s}(t) + \kappa_p =
    \xi_p^{s}(t) + \frac{\kappa_p}{2}
    \geq \frac{\kappa_p}{2},
  \end{equation*}
  i.e., $\kappa_p< C_{s''+1}$, contradicting the maximality of
  $s''$.  We conclude that $s''=s'$, concluding the proof.
\end{proof}

Next, we show that the level-$s$ stabilization condition guarantees that
paths of relevant skew on level $s$ must have existed for at least $\Theta_s$
time on level $s$. Moreover, if we append an edge $e$ to such a path that exists
only in one direction, we can still argue that the entire path has been a
level-$s$ path until at most $\tau_e$ time ago.
\begin{lemma}\label{lemma:paths}
  Assume that node $u\in V$ satisfies the $(C,s)$-stabilization
  condition for a gradient sequence $C$ and an integer $s>1$
  throughout a time interval $[t^-,t^+]$. Consider a path $p =
  (u,\ldots,v)\in P_u^s(t^+)$. If $\xi_p^s(t^+) \geq 0$,
  $\psi_p^s(t^+)\geq-\kappa_p/2$, $\xi_{\bar{p}}^s(t^+)\geq0$, or
  $\psi_{\bar{p}}^s(t^+)\geq-\kappa_p/2$, then $p\in P_u^s(t)$ for all
  \begin{equation*}
    t\in \left[t^--\stabint{s},t^+\right].
  \end{equation*}
  Furthermore, let $p'=p\circ (v,w)=(u,\dots,v,w)$. If $w\in
  \N_v^s(t^+)$ and $\xi_{p'}^s(t^+-\tup_{\{v,w\}})\geq 0$,
  $\psi_{p'}^s(t^+-\tup_{\{v,w\}})\geq -\kappa_{p'}/2$,
  $\xi_{\bar{p}'}^s(t^+-\tup_{\{v,w\}})\geq 0$, or
  $\psi_{\bar{p}'}^s(t^+-\tup_{\{v,w\}})\geq -\kappa_{p'}/2$, then
  $p'\in P_u^s(t)$ for all
  \begin{equation*}
    t\in \left[t^--\stabint{s},t^+-\tup_{\{v,w\}}\right],
  \end{equation*}
  and this time interval is non-empty.
\end{lemma}

\begin{proof}
  For all $t$, we have $\xi_{\bar{p}}^s(t) = \psi_p^s(t) + \kappa_p/2$
  and $\psi_{\bar{p}}^s(t) = \xi_p^s - \kappa_p/2$. Therefore,
  assuming that $\xi_p^s(t^+) \geq 0$, $\psi_p^s(t^+)\geq-\kappa_p/2$,
  $\xi_{\bar{p}}^s(t^+)\geq0$, or
  $\psi_{\bar{p}}^s(t^+)\geq-\kappa_p/2$ implies that either
  $\xi_p^s(t^+)\geq0$ or $\psi_p^s(t^+)\geq-\kappa_p/2$. Because the
  system satisfies the $(C,s)$-stabilization condition at node $u$ at
  time $t^+$, using Lemma \ref{lemma:shortpaths}, together with $p\in
  P_u^s(t^+)\supseteq P_u^{s-1}(t^+)$ both cases imply that
  $\kappa_p\leq C_{s-1}$. 
  
  In the following, we define 
  \[
  T_s^+ :=\min\set{T_s\in\Tset_s: T_s\geq \lc{u}(t^+)}
  \quad\text{and}\quad
  \forall s'\leq s:
  T_{s'}^-:=\max\set{T_{s'}\in\Tset_{s'}: T_{s'}\leq \lc{u}(t^-)}.
  \]
  For all $t\in[t^-,t^+]$, the stabilization condition yields that
  \begin{equation}\label{eq:clock_u}
    T_s^- +  (1+\drift)(1+\mu)\stabint{s} - sC_{s-1} < \lc{u}(t) <
    T_s^+ - (1+\drift)(1+\mu)\stabint{s} + sC_{s-1}.
  \end{equation}
  Let $X$ be the set of nodes of the path $p$. In order to prove the
  first part of the lemma, we show that
  \begin{equation}\label{eq:lemmapart1}
    \forall t\in[t^--\stabint{s},t^+],\ \forall x\in X :
    \text{ $x$ is $(C,s-1)$-legal at time $t$ and }
    N_x^s(t)\supseteq N_x^s(t^+).
  \end{equation}
  Note that this directly implies the first part of the lemma. Before
  showing \eqref{eq:lemmapart1}, we first show that 
  \begin{equation}\label{eq:Xisstable}
    \forall t\in[t^--\stabint{s},t^+]\text{ for which
      \eqref{eq:lemmapart1} holds, } \forall x\in X:~T_s^- < \lc{x}(t) < T_s^+.
  \end{equation}
  As at time $t$ \eqref{eq:lemmapart1} holds, the sub-path
  $(u,\ldots,x)\in P_u^s(t)\supseteq P_u^{s-1}(t)$ and the system is
  $(C,s-1)$-legal at both $u$ and $x$ at time $t$.  Applying
  Lemma~\ref{lemma:clocks}, we see that
  \[
  |\lc{u}(t) - \lc{x}(t)| \leq s\cdot C_{s-1}
  \]
  for all $x\in X$. Because $t> t^--\stabint{s}$, we have
  $\lc{u}(t)> \lc{u}(t^-)-(1+\mu)(1+\drift)\stabint{s}$ and therefore
  \[
  |\lc{u}(t^-) - \lc{x}(t)| < (1+\mu)(1+\drift)\stabint{s} + s\cdot C_{s-1}.
  \]
  Together with Inequality \eqref{eq:clock_u}, this implies that
  $T_s^-<\lc{x}(t)<T_s^+$ and thus \eqref{eq:Xisstable} for all $x\in
  X$.

  Let us now show that \eqref{eq:lemmapart1} holds.  Trivially, it
  holds that $N_x^s(t^+)\supseteq N_x^s(t^+)$. In addition, due to
  Statement~(i) of Lemma~\ref{lemma:simple_stuff} we have that
  $p\in P_u^{s-1}(t^+)$ and Lemma~\ref{lemma:shortpaths} therefore
  establishes \eqref{eq:lemmapart1} for time $t=t^+$.  Now suppose for
  the sake of contradiction that the non-empty maximal time interval
  $[t',t^+]\subseteq [t^--\stabint{s},t^+]$ for which both statements
  of \eqref{eq:lemmapart1} are satisfied is not equal to
  $[t^--\stabint{s},t^+]$, i.e., $t'>t^--\stabint{s}$. At time $t'$,
  \eqref{eq:lemmapart1} still holds and therefore from
  \eqref{eq:Xisstable}, we get that $T_s^-<\lc{x}(t')<T_s^+$. As nodes
  add level-$s$ neighbors only at times $T_s\in \Tset_s$, it follows
  that there is some time $t''\in [t^--\stabint{s},t')$ such that
  $N_x^s(t)\supseteq N_x^s(t')\supseteq N_x^s(t^+)$ for all $t\in
  [t'',t']$. Therefore, $p\in P_u^{s-1}(t)\subseteq P_u^s(t)$ for all
  such $t$ and because $u$ satisfies the $(C,s)$-stabilization
  condition throughout interval $[t^-,t^+]$ and thus at some time in
  $[t,t+\stabint{s}]\subseteq [t,t+\stabint{s-1}]$, we can apply
  Lemma~\ref{lemma:shortpaths} to infer that the system is
  $(s-1)$-legal at $x$ at time $t$. This is a contradiction to the
  maximality of the interval $[t',t^+]\subseteq [t^--\stabint{s},t^+]$
  (i.e., the minimality of $t'$). We conclude that indeed
  $t'=t^--\stabint{s}$, in particular showing the first statement of
  the lemma.

  It therefore remains to prove the second part of the lemma. For
  the path $p'$ used in the second part of the claim, we know that
  either
  \begin{equation}\label{eq:pprimepos}
    \xi_{p'}^s(t^+-\tup_{\set{v,w}}) \geq 0
    \quad\text{or}\quad
    \psi_{p'}^s(t^+-\tup_{\set{v,w}}) \geq -\frac{\kappa_{p'}}{2}.
  \end{equation}
  Let us first assume that $\tup_{\set{v,w}}< \stabint{s}$ and that
  $v\in N_w^s(t^+-\tup_{\set{v,w}})$. From $w\in N_v^s(t^+)$ and
  \eqref{eq:lemmapart1}, this implies that $w\in
  N_v^s(t^+-\tup_{\set{v,w}})$ and thus $p'\in
  P_u^s(t^+-\tup_{\set{v,w}})$. Further, by the stabilization condition, the
  system is $(C,s-1)$-legal at $u$ at time
  $t^+-\tup_{\set{v,w}}$. Because by \eqref{eq:pprimepos}, either
  $\xi_{p'}^s(t^+-\tup_{\set{v,w}})\geq0$ or
  $\psi_{p'}^s(t^+-\tup_{\set{v,w}})\geq -\kappa_{p'}/2$, Lemma
  \ref{lemma:shortpaths} then implies that $\kappa_{p'}\leq
  C_{s-1}$. Therefore, in that case exactly the same argument as for
  the first part of the lemma also shows that for all
  $t\in[t^--\stabint{s},t^+-\tup_{\set{v,w}}]$ we have $p'\in
  P_u^s(t)$. To also prove the second claim, it therefore suffices to
  show that $\tup_{\set{v,w}}< \stabint{s}$ and that $v\in
  N_w^s(t^+-\tup_{\set{v,w}})$.

  To this end, we will prove by induction that
  \begin{itemize}
  \item [(i)] $\forall s'\in \{1,\ldots,s\}: p'\in
    P_u^{s'}(t^+-\tup_{\set{v,w}})$, and
  \item [(ii)] $\forall s'\in \{2,\ldots,s\}:
    \tup_{\set{v,w}}<\Theta_{s'}$.
  \end{itemize}
  For $s'=s$, these statements imply the above and hence complete the
  proof. 

  Before we continue with the induction, we make the following
  observation. In the first part of the lemma, we proved
  \eqref{eq:Xisstable}, implying that no node $x\in X$ adds a new
  neighbor to $N_x^s$ during the interval
  $[t^--\stabint{s},t^+]$. Thus, for all $x\in X$ and all $y\in
  N_x^s(t^+)$, $y$ has been added to $N_x^{s}$ at the latest at time
  $t<t^--\stabint{s}$ for which $\lc{x}(t)=T_s^-$. As $x$ adds $y$ to
  set $N_x^s$ after adding $y$ to sets $N_x^{s'}$ for all $s'<s$, this
  also implies that for all $x\in X$, all $y\in N_x^s(t^+)$, and all
  $s'\leq s$, node $x$ adds $y$ to $N_x^{s'}$ at the latest at time
  $t$ for which $\lc{x}(t)=T_{s'}^-$. Note that this includes node $w$
  which is in the set $N_v^s(t^+)$. By Lemma~\ref{lemma:insertion}, both
  nodes on an edge add the edge to the respective neighbor sets at the
  same logical times. We therefore also know that $w$ adds $v$ to
  $N_w^{s'}$ at the latest at time $t_w^{s'}$ such that
  $\lc{w}(t_w^{s'})=T_{s'}^-$. To prove (i) for a specific $s'$, it is
  thus sufficient to prove for all $x\in X\cup\{w\}$ (note the
  inclusion of $w$!) that $\lc{x}(t^+-\tup_{\set{v,w}})\geq T_{s'}^-$.

  We now proceed with the induction. We anchor it at $s'=2$. Let us
  first consider (ii) for $s'=2$. From (S0) of the
  $(C,s)$-stabilization condition at node $u$ at time $t^+$, we have
  $C_1 > (1+\drift)\mu\tup_{\set{v,w}}$ and therefore 
  \begin{equation}\label{eq:iibase}
    \tup_{\set{v,w}} < \frac{C_1}{(1+\drift)\mu} = \stabint{2}.
  \end{equation}
  Further, as node $u$ satisfies the
  $(C,s)$-stabilization condition at time $t^+$ (for $s\geq 2$), from
  (S2), we have
  \begin{eqnarray*}
    \lc{u}(t^+-\tup_{\set{v,w}}) & \geq & \lc{u}(t^+) -
    (1+\mu)(1+\drift)\tup_{\set{v,w}}\\
    & \stackrel{(S2)}{\geq} & T_2^- + (1+\mu)(1+\drift)\stabint{2} +
    2C_1 - (1+\mu)(1+\drift)\tup_{\set{v,w}}\\
    & \stackrel{\eqref{eq:iibase}}{>} & T_2^- +2C_1\\
    & \stackrel{(S0)}{\geq} & T_2^- + \G(t^+-\tup_{\set{v,w}}).
  \end{eqnarray*}
  The last inequality follows from $(S0)$ because we already know that
  $\tup_{\set{v,w}}<\stabint{2}$. For every node $x\in V$ (and
  therefore in particular for every node $x\in X\cup\set{w}$), we
  thus have
  \[
  \lc{x}(t^+-\tup_{\set{v,w}}) > T_2^-.
  \]
  We have already seen that this implies also Statement~(i) for $s'=2$.

  The induction step comes in two parts. First, we prove for $s'\in
  \{1,\ldots,s-1\}$ that Statement~(i) for $s'$ and Statement~(ii) for
  $s'$ imply Statement~(ii) for $s'+1$.  From Statement~(i) for $s'$,
  we know that $p'\in P_u^{s'}(t)$ for $t^+-\tup_{\set{v,w}}$. From
  Statement~(ii), we also know that
  $t^+-\tup_{\set{v,w}}>t^+-\stabint{s'}$. Since the
  $(C,s)$-stabilization condition holds at node $u$ at time $t^+$, we
  therefore know that $u$ is $(C,s')$-legal at time
  $t^+-\tup_{\set{v,w}}$. Together with \eqref{eq:pprimepos},
  Lemma~\ref{lemma:shortpaths} then implies that $\kappa_{p'}\leq
  C_{s'}$. We thus get that
  \begin{equation*}
    \tup_{\{v,w\}}\stackrel{\eqref{eq:def_kappa}}{<}
    \frac{\kappa_{\{v,w\}}}{(1+\drift)\mu} 
    \stackrel{(\kappa_{p'}\leq C_{s'})} \leq
    \frac{C_{s'}}{(1+\drift)\mu} = \stabint{s'+1},
  \end{equation*}
  
  To conclude the induction step, we now also show that Statement~(i)
  for $s'$ and Statement~(ii) for $s'+1$ imply Statement~(i) for
  $s'+1$. We lower bound $\lc{u}(t^+-\tup_{\set{v,w}})$ as follows:
  \begin{eqnarray*}
    \lc{u}(t^+-\tup_{\set{v,w}}) &\geq &
    \lc{u}(t^+) - (1+\mu)(1+\drift)\tup_{\set{v,w}}\\
    &\stackrel{\text{I.H.}}{>} & \lc{u}(t^+) - (1+\mu)(1+\drift)\stabint{s'+1}\\
    &\stackrel{(S2)}{\geq} &
    T_{s'+1}^{-} + (1+\mu)(1+\drift)\stabint{s'+1}  + (s'+1)C_{s'} -
    (1+\mu)(1+\drift)\stabint{s'+1}\\
    &=&
    T_{s'+1}^{-} + (s'+1)C_{s'}.
  \end{eqnarray*}

  As $p'\in P_u^{s'}(t)$, for each $x\in X\cup \{w\}$
  Lemma~\ref{lemma:shortpaths} yields that the system is $s'$-legal at
  $x$ and time $t^+-\tup_{\set{v,w}}$ as well as $\kappa_{p'}\leq
  C_{s'}$. By Lemma~\ref{lemma:clocks}, it follows that
  $L_x(t^+-\tup_{\set{v,w}})>T_{s'+1}^{-}$. As already noted, this
  implies that $p'\in P_u^{s'+1}(t^+-\tup_{\set{v,w}})$, as required.
\end{proof}

\subsection{Properties of \texorpdfstring{\boldmath{$\Xi$}}{Xi}}

The next lemma shows, roughly speaking, that when a node $u$ is too far ahead on
some level-$s$ paths---that is, it has a positive $\Xi_u^s$ value---then all
endpoints of paths $p$ satisfying $\xi_p^s=\Xi_u^s$ are in fast mode, trying to
catch up to $u$. Their clocks increase at a rate of at least $(1 - \drift)(1 +
\mu)$, the slowest possible fast rate. However, we do not know what node $u$
itself does in this situation: because of the local nature of the algorithm, $u$
does not necessarily realize that it has a large skew, and it can be in either
slow mode or fast mode. In the latter case, the skew on the path might actually
increase, due to hardware clock drift; thus we cannot necessarily show that the
skew decreases. The lemma states that over an interval, the weighted skew
$\Xi_u^s$ increases by at most $u$'s logical clock increase, \emph{minus} the
catching-up that nodes trailing behind $u$ achieve at a rate of $(1 - \drift)(1
+ \mu)$. Finally, it may be the case that an edge $e$ is not present throughout
the entire relevent time period, causing us to ``jump'' back in by $\tau_e$
time; in this case, we use some additional slack in \FC to gain a ``reserve
term'' accounting for the resulting time difference later on.
\begin{lemma}
  \label{lemma:pull}
  Assume that a node $u\in V$ satisfies the $(C,s)$-stabilization for
  $s > 1$ throughout a time interval $[t^-, t^+]$, and that for all $t
  \in (t^-, t^+)$ we have $\Xi_u^s(t) > 0$. Then there exists a time
  \begin{equation*}
    t'\in \left[ t^- - \stabint{s},t^- \right]
  \end{equation*}
  such that
  \begin{equation*}
    \Xi_u^s(t^+)-\Xi_u^s(t')\leq L_u(t^+)-L_u(t')-(1-\drift)(1+\mu)(t^+-t')
    -(1+\drift)\mu(t^--t').
  \end{equation*}
\end{lemma}
\begin{proof}
  Set $u_0\coloneq u$, and consider an arbitrary time $t \in (t^-, t^+)$.
  Let $p = (u_0,\ldots,u_k) \in P_u^s(t)$ be any path such that $\Xi_u^s(t) =
  \xi_p^s(t)$ (that is, a path that maximizes the value of $\xi_p^s(t)$ at time
  $t$). By assumption, $\xi_p^s(t) > 0$. We will show that for the endpoint $u_k$,
  the first condition of $\FC$ is satisfied; specifically, the next node $u_{k-1}$
  on the path satisfies $\lc{u_{k-1}}(t) - \lc{u_k}(t) \geq s \cdot
  \kappa_{\set{u_{k-1}, u_k}}$. Note that since $\xi_p^s(t) = \lc{u_0}(t) -
  \lc{u_k}(t) - s \cdot \kappa_p > 0$, we cannot have $u_0 = u_k$, i.e., $u_{k-1}$
  must exist.

  Consider the sub-path $(u_0, \ldots, u_{k-1})$. Because it is a sub-path of $p$,
  and $p \in P_u^s(t)$, we also have $(u_0, \ldots, u_{k-1}) \in P_u^s(t)$.
  Further, by choice of $p$ we know that $\xi_p^s(t) = \Xi_u^s(t) \geq \xi_{(u_0,
    \ldots, u_{k-1})}^s(t)$; that is,
  \begin{equation*}
    \lc{u_0}(t) - \lc{u_k}(t) - s \cdot \kappa_{(u_0, \ldots, u_k)} \geq \lc{u_0}(t) - \lc{u_{k-1}}(t) - s \cdot \kappa_{(u_0, \ldots, u_{k-1})},
  \end{equation*}
  which we can re-arrange to obtain
  \begin{equation*}
    \lc{u_{k-1}}(t) - \lc{u_k}(t) \geq s \cdot \left( \kappa_{(u_0, \ldots, u_k)} - \kappa_{(u_0, \ldots, u_{k-1})} \right) =
    s \cdot \kappa_{\set{u_{k-1}, u_k}}.
  \end{equation*}
  This is the first condition of $\FC$ at node $u$. We cannot guarantee that the
  second condition holds, but if it does, that is, if at time $t$ we also have
  \begin{equation}
    \label{eq:not_blocked}
    \forall v \in \N_{u_k}^s(t)  : \quad
    \lc{u_k}(t) - \lc{v}(t) \leq s \cdot \kappa_{\set{u_k,v}}+2\mu\tup_{\set{u_k,v}},
  \end{equation}
  then $\FC$ is satisfied at node $u_k$ at time $t$, and $u_k$ must be in fast
  mode. In that case we have $l_{u_k}(t) \geq (1 - \drift)(1 + \mu)$, which is
  the rate needed for the statement of the lemma. We proceed by considering the
  longest suffix of $[t^-, t^+)$ such that $\FC$ holds for all nodes that maximize
  $\Xi_u^s$ during the interval. At the point where $\FC$ stops holding for some
  node $u_k$ that maximizes $\Xi_u^s$, we can show that some other node is ``to
  blame for this'', and that node has an even larger skew to $u$.

  Let $\theta \in [t^-, t^+]$ be the infimal time such that for all $t \in
  (\theta, t^+)$, condition~\eqref{eq:not_blocked} above holds for all paths
  $p$ (where $p$ is a path maximizing $\xi_p^s(t)$, as defined above), where
  $\theta \coloneq t^+$ if no such time exists.

  By definition, $\Xi_u^s(t) = \max_{p \in P_u^s(t)} \{\xi_p^s(t)\}$.
  Each $\xi_p^s(t)$ is continuous and left-differentiable, since it is
  obtained by taking the difference of logical clocks, which are themselves
  continuous and left-differentiable. Therefore $\Xi_u^s(t)$ is also continuous
  and left-differentiable. By choice of $\theta$, for all $t \in (\theta, t^+)$
  we have 
  \begin{equation*}
  d/dt^- \xi_{p(t)}^s(t) \leq d/dt^- \lc{u}(t) - (1 - \drift)(1 + \mu)(t^+ -
  \theta)
  \end{equation*}
  (where $p(t)$ is the path such that $\xi_{p(t)}^s(t) = \Xi_u^s(t)$),
  because the other endpoint of $p(t)$ is in fast mode and its logical clock
  increases at a rate of at least $(1 - \drift)(1 + \mu)$. Consequently also $d
  /dt^- \Xi_u^s(t) \leq d/dt^- \lc{u}(t) - (1 - \drift)(1 + \mu)(t^+ - \theta)$.
  Using the mean-value theorem (which generalizes to the case where the function
  is only semi-differentiable), we see that over any interval $[\theta_1,
  \theta_2] \subseteq [\theta, t^+]$ where $P_u^s$ does not change,
  \begin{equation*}
    \Xi_u^s(\theta_2) - \Xi_u^s(\theta_1) \leq \lc{u}(\theta_2) - \lc{u}(\theta_1) - (1 - \drift)(1 + \mu)(\theta_2 - \theta_1).
  \end{equation*}
  Now consider points in time when $P_u^s$ changes. If a path is removed from
  $P_u^s$ at time $t$ then the value of $\Xi_u^s(t)$ can only decrease. If a path
  $q$ is added to $P_u^s$ at time $t$, then Lemma~\ref{lemma:paths} shows that
  $\xi_q^s(t) < 0$, (otherwise $q$ must be in $P_u^s$ throughout $[t^-, t]$).
  By the conditions of the current lemma we know that $\Xi_u^s(t) > 0$, so
  $\xi_q^s(t) < \Xi_u^s(t)$, and again $\Xi_u^s(t)$ is not increased by the
  addition of $q$.
  It follows that over the entire interval $[\theta, t^+]$,
  \begin{equation}\label{eq:smoothly}
    \Xi_u^s(t^+) - \Xi_u^s(\theta) \leq \lc{u}(t^+) - \lc{u}(\theta) - (1 - \drift)(1 + \mu)(t^+ - \theta).
  \end{equation}
  Therefore, if $\theta = t^-$, then we set $t' \coloneq t^-$ and we are done.

  Suppose that $\theta > t^-$. Let $p =  (u_0, \ldots, u_k)$ be some path
  such that $\Xi_u^s(\theta) = \xi_p^s(\theta)$ and
  Condition~\eqref{eq:not_blocked} does not hold for $u_k$ (by choice of
  $\theta$ such a path exists); that is, there is some $v \in \N_{u_k}^s(\theta)$
  such that
  \begin{equation}
    \label{eq:blocked}
    \lc{u_k}(\theta) - \lc{v}(\theta) > s \cdot \kappa_{\set{u_k,v}} + 2\mu\tup_{\set{u_k,v}}.
  \end{equation}
  Let $p' \coloneq p\circ (u_k,v)=(u_0, \ldots, u_k, v)$. We use
  Lemma~\ref{lemma:paths} to ``switch'' from path $p$ to path $p'$ and go back in
  time to time $\theta - \tau_{\set{u_k, v}}$, \emph{increasing the weighted skew}
  as we go back in time.
  We have
  \begin{eqnarray}
    \xi_{p'}^s(\theta - \tup_{\{u_k,v\}})
    &=&
    \xi_{p'}^s(\theta) -
    (L_u(\theta)-L_u(\theta-\tup_{\{u_k,v\}})) +
    L_v(\theta)-L_v(\theta-\tup_{\{u_k,v\}})\nonumber\\
    &\geq&
    \xi_p^s(\theta)
    -(L_u(\theta)-L_u(\theta-\tup_{\{u_k,v\}}))+(1-\drift)\tup_{\{u_k,v\}}
    \nonumber\\
    &&
    + (L_{u_k}(\theta)-L_v(\theta)-s \cdot \kappa_{\{u_k,v\}})
    \nonumber\\
    &\sr{eq:blocked}{>}& \xi_p^s(\theta)
    -(L_u(\theta)-L_u(\theta-\tup_{\{u_k,v\}}))
    +(1-\drift+2\mu)\tup_{\{u_k,v\}}\nonumber\\
    &=& \Xi_u^s(\theta)-(L_u(\theta)-L_u(\theta-\tup_{\{u_k,v\}}))
    +(1-\drift+2\mu)\tup_{\{u_k,v\}}\label{eq:xi_pprime}\\
    &
    \geq
    & \Xi_u^s(\theta) - (1 + \drift)(1 + \mu)\tup_{\set{u_k, v}}
    + (1 - \drift + 2\mu)\tup_{\set{u_k, v}}\nonumber
    \\
    &>& 0,\nonumber
  \end{eqnarray}
  where in the last step we used the fact that $(1-\drift)\mu>2\drift$. Hence,
  Lemma~\ref{lemma:paths} shows that $p'\in P_u^s(\theta-\tup_{\{u_k,v\}})$,
  giving
  \begin{eqnarray*}
    \Xi_u^s(\theta)-\Xi_u^s(\theta-\tup_{\{u_k,v\}}) &\leq&
    \Xi_u^s(\theta)-\xi_{p'}^s(\theta-\tup_{\{u_k,v\}})\\
    &\sr{eq:xi_pprime}{<}& L_u(\theta)-L_u(\theta-\tup_{\{u_k,v\}})
    -(1-\drift+2\mu)\tup_{\{u_k,v\}}.
  \end{eqnarray*}
  We conclude that
  \begin{equation}
    \Xi_u^s(t^+) - \Xi_u^s(\theta-\tup_{\{u_k,v\}}) \sr{eq:smoothly}{\leq}
    L_u(t^+)-L_u(\theta-\tup_{\{u_k,v\}})
    -(1-\drift)(1+\mu)(t^+-(\theta-\tup_{\{u_k,v\}}))
    -(1+\drift)\mu\tup_{\{u_k,v\}}.
    \label{eq:Xi_ih}
  \end{equation}

  From Lemma~\ref{lemma:paths} we know that $\theta - \tup_{\set{u,v}} \geq t^- -
  \Theta_s$, that is, we did not go back too far in time. Thus, if $\theta -
  \tup_{\{u_k,v\}} \leq t^-$, the statement follows by setting $t'\coloneq \theta
  - \tup_{\{u_k,v\}}$: Recall that by definition, $\theta\geq t^-$, and hence $(1
  + \drift)\mu (t^- - (\theta - \tup_{\set{u,v}})) \leq (1 +
  \drift)\mu\tup_{\set{u,v}}$; therefore~\eqref{eq:Xi_ih} shows that
  \begin{align*}
    \Xi_u^s(t^+) - \Xi_u^s(t') &\leq \lc{u}(t^+) - \lc{u}(t') - (1 - \drift)(1 + \mu)(t^+ - t')
    - (1 + \drift)\mu\tup_{\set{u_k, v}}
    \\
    & \leq \lc{u}(t^+) - \lc{u}(t') - (1 - \drift)(1 + \mu)(t^+ - t')
    - (1 + \drift)\mu(t^- - t').
  \end{align*}

  Otherwise, if $\theta - \tup_{\set{u_k, v}} > t^-$, we drop the term $-(1 +
  \drift)\mu\tup_{\set{u_k, v}}$ from~\eqref{eq:Xi_ih} to obtain
  \begin{equation*}
    \Xi_u^s(t^+) - \Xi_u^s(\theta-\tup_{\{u_k,v\}}) \sr{eq:smoothly}{\leq}
    L_u(t^+)-L_u(\theta-\tup_{\{u_k,v\}})
    -(1-\drift)(1+\mu)(t^+-(\theta-\tup_{\{u_k,v\}})).
  \end{equation*}
  To prove the claim, it is sufficient to find a time $t' \in [t^- - \stabint{s},
  \theta - \tup_{\set{u_k, v}}]$ for which
  \begin{equation}
    \Xi_u^s(\theta - \tup_{\set{u_k, v}}) - \Xi_u^s(t') \leq
    \lc{u}(\theta - \tup_{\set{u_k, v}}) - \lc{u}(t')
    - (1 - \drift)(1 + \mu)( (\theta - \tup_{\set{u_k, v}}) - t')
    - (1 + \drift)\mu(t^- - t').
    \label{eq:Xi_step}
  \end{equation}
  In other words, we need to show the original statement of the lemma, but only
  for the sub-interval $[t^-, \theta - \tup_{\set{u_k, v}}] \subset [t^-, t^+]$.
  The lemma then follows by summing~\eqref{eq:Xi_ih} and~\eqref{eq:Xi_step}.

  To this end, we continue inductively, applying the entire argument over again to
  the interval $[t^-, \theta - \tup_{\set{u_k, v}}]$. At each step we go back in
  time at least $\min_{x \neq y \in V} \tup_{\set{x,y}} > 0$, and we never go
  further back than $t^- - \stabint{s}$; therefore the induction halts after a
  finite number of steps, at a time
  \begin{equation*}
    t'\in \left[ t^- - \stabint{s},t^- \right],
  \end{equation*}
  for which it holds that
  \begin{equation*}
    \Xi_u(t^+) - \Xi_u(t') \leq \lc{u}(t^+) - \lc{u}(t') - (1 - \drift)(1 + \mu)(t^+ - t')
    -(1 + \drift)\mu(t^- - t'),
  \end{equation*}
  as required.
\end{proof}

We will also need a technical helper lemma about $\Xi_u^s$ that guarantees that
$\Xi_u^s$ remains positive under certain circumstances, enabling to apply
Lemma~\ref{lemma:pull}.
\begin{lemma}
  Let $[t^-, t^+]$ be an interval such that node $u$ satisfies the
  $(C,s)$-stabilization condition throughout $[t^-, t^+]$, where
  $s>1$. If
  \begin{equation}
    \Xi_u^s(t^+) \geq 2\drift(t^+ - t^-)
    \label{eq:ass_u_Xi}
  \end{equation}
  and
  \begin{equation}
    \lc{u}(t^+) - \lc{u}(t^-) \leq (1 + \drift)(t^+ - t^-),
    \label{eq:ass_avg_rate}
  \end{equation}
  then for all $t \in (t^-, t^+]$ we have $\Xi_u^s(t) > 0$.
  \label{lemma:Xi_pos}
\end{lemma}
\begin{proof}
  Let $p = (u, \ldots, v)$ be a path such that $\xi_p^s(t^+) =
  \Xi_u^s(t^+) \geq 2\drift(t^+ - t^-)$ and let $t \in [t^-,t^+]$ Lemma~\ref{lemma:paths} states
  that $p \in P_u^s(t)$ and hence $\Xi_u^s(t) \geq \xi_p^s(t)$. How
  much can $\xi_p^s$ decrease when we go back from time $t^+$ to time
  $t$? We have
  \begin{equation*}
    \lc{u}(t^+) - \lc{u}(t)
    =
    \lc{u}(t^+) - \lc{u}(t^-)
    -
    \left( \lc{u}(t) - \lc{u}(t^-)\right)
    \leq
    (1 + \drift)(t^+ - t^-)
    -
    (1 - \drift)(t - t^-).
  \end{equation*}
  Therefore,
  \begin{align*}
    \xi_p^s(t)
    &=
    \xi_p^s(t^+)
    - \left( \lc{u}(t^+) - \lc{u}(t) \right)
    + \left( \lc{v}(t^+) - \lc{v}(t) \right)
    \\
    &
    >
    2\drift(t^+ - t^-)
    - (1 + \drift)(t^+ - t^-) + (1 - \drift)(t - t^-)
    + (1 - \drift)(t^+ - t)
    \\
    &=
    2\drift(t^+ - t^-) - (1 + \drift)(t^+ - t^-) + (1 - \drift)(t^+ - t^-)
    \\
    &= 2\drift(t^+ - t^-) - 2\drift(t^+ - t^-) = 0.\qedhere
  \end{align*}
\end{proof}

\subsection{Decrease of \texorpdfstring{\boldmath{$\Psi$}}{Psi} and
  Stability}

Next, we show a simple lemma that roughly states that if a node is in
fast mode, then a condition slightly weaker than the slow mode
condition \SC is false.  We already know that \SC itself cannot
hold when a node is in fast mode from Lemma~\ref{lemma:algo_correct}.
The weaker condition is almost the same as \SC, except that the slack
$\delta$ is removed.

\begin{lemma}\label{lemma:neighbors_clocks}
  Assume that for some node $u\in V$, a level $s\in \nat$, and times
  $t^-<t^+$ we have $(L_u(t^-),L_u(t^+)]\cap \Tset_s =\emptyset$. If
  \begin{equation*}
    t^-<t_0 \coloneq \min \set{ t \in [t^-,t^+] \st L_u(t^+)-L_u(t) \leq (1 +
      \drift)(t^+ - t)},
  \end{equation*}
  then
  \begin{equation}\label{eq:fast_neighbor}
    \exists w\in \N_u^s(t_0):~L_w(t_0)-L_u(t_0) >
    \left(s+\frac{1}{2}\right)\kappa_{\set{u,w}} + \mu(1+\drift)\tup_{\set{u,w}}
  \end{equation}
  or
  \begin{equation}\label{eq:slow_neighbors}
    \forall v\in \N_u^s(t_0):~L_u(t_0)-L_v(t_0)<
    \left(s+\frac{1}{2}\right)\kappa_{\set{u,v}}.
  \end{equation}
\end{lemma}
\begin{proof}
  Assuming the contrary, the logical negation of $\eqref{eq:fast_neighbor} \vee
  \eqref{eq:slow_neighbors}$ is
  \begin{eqnarray*}
    &&\forall v\in \N_u^s(t_0):~L_v(t_0)-L_u(t_0)\leq
    \left(s+\frac{1}{2}\right)\kappa_{\{u,v\}}+ \mu(1+\drift)\tup_{\set{u,w}}\\
    \wedge&&\exists w\in \N_u^s(t_0):~L_u(t_0)-L_w(t_0)\geq
    \left(s+\frac{1}{2}\right)\kappa_{\{u,w\}}.
  \end{eqnarray*}
  Because $(L_u(t^-),L_u(t^+)]\cap \Tset_s =\emptyset$, no new neighbors
  are added to $\N_u^s$ in the interval $(t^-, t^+]$. Neighbors can be
  removed during the interval, but there are only finitely many
  neighbors to remove ($\N_u^s(t^-)$ is finite). Recall that a node is
  considered being in $N_u^s$ at both the times of its insertion and
  removal. Hence there is some $\tilde{t}\in [t^-,t_0)$ such that
  $\N_u^s(t)=\N_u^s(t')$ for all $t,t'\in [\tilde{t},t_0]$.
  Furthermore, since logical clocks are continuous and $\delta>0$, there
  is a sub-interval $[t_0', t_0] \subseteq [\tilde{t},t_0]$ such that
  $t_0'<t_0$ and for all $t \in [t_0', t_0]$ we have
  \begin{eqnarray*}
    &&\forall v\in \N_u^s(t_0)=\N_u^s(t):~L_v(t)-L_u(t)\leq
    \left(s+\frac{1}{2}\right)\kappa_{\{u,v\}}+
    \mu(1+\drift)\tup_{\set{u,w}}+\delta\\
    \wedge&&\exists w\in \N_u^s(t_0)=\N_u^s(t):~L_u(t)-L_w(t)\geq
    \left(s+\frac{1}{2}\right)\kappa_{\{u,w\}}-\delta.
  \end{eqnarray*}
  Thus, $\SC$ applies at $u$ for all $t \in [t_0', t_0]$, and $u$ must be in slow
  mode during this interval, so $\lc{u}(t_0) - \lc{u}(t_0') \leq (1 +
  \drift)(t_0 - t_0')$. But we also know by choice of $t_0$ that $\lc{u}(t^+) -
  \lc{u}(t_0) \leq (1 + \drift)(t^+ - t_0)$; therefore,
  \begin{equation*}
    L_u(t^+)-L_u(t_0')=L_u(t^+)-L_u(t_0)+L_u(t_0)-L_u(t_0')\leq
    (1+\drift)(t^+-t_0'),
  \end{equation*}
  contradicting the definition of $t_0$.
\end{proof}

We are now ready to prove our main theorem.
\begin{proof}[Proof of Theorem~\ref{thm:catch_up}]
Assume for the sake of contradiction that at some time $t_0 \in [t^- + \Lambda_s
+ 2\stabint{s}, t^+]$ we have $\Psi_u^s(t_0) \geq 2\drift \Lambda_s$; that is,
there is a path $(u = u_k, \ldots, u_0) \in P_u^s(t_0)$ such that
  \begin{equation}
    \psi_{(u_k, \ldots, u_0)}^s(t_0) = \Psi_u^s(t_0) \geq 2\drift \Lambda_s.
    \label{eq:initial_psi}
  \end{equation}
  For the inverse path $(u_0, \ldots, u_k)$ we have
  \begin{equation}
    \xi_{(u_0, \ldots, u_k)}^s(t_0) = \psi_{(u_0, \ldots, u_k)}^s(t_0)
    + \frac{\kappa_{(u_0, \ldots, u_k)}}{2}
    \geq
    2\drift\Lambda_s
    + \frac{\kappa_{(u_0, \ldots, u_k)}}{2}
    > 0.
    \label{eq:initial_xi}
  \end{equation}
  Lemma~\ref{lemma:shortpaths} shows that
  \begin{equation}
    \kappa_{(u_0, \ldots, u_k)} \leq C_{s-1}.
    \label{eq:initial_kappa}
  \end{equation}
  By the conditions of the lemma, $u$ satisfies the
  $(C,s)$-stabilization condition throughout $[t^-, t^+]$. Hence
  Lemma~\ref{lemma:paths} shows that for all $t \in [t^-, t_0]$ we
  have $(u_0, \ldots, u_k) \in P_u^s(t)$. In particular, each sub-path
  $(u_0, \ldots, u_i)$ is also in $P_u^s$ throughout $[t^-, t_0]$, and
  by the conditions of the lemma, this shows that each node $u_i$ on
  the path satisfies the $(C,s)$-stabilization condition throughout
  $[t^-, t_0]$.

  We construct a sequence of non-increasing times
  $t_0 = t_0' \geq t_1 \geq t_1' \geq \ldots t_{\ell} \geq t_{\ell'}$,
  where $t_{\ell} \geq t_0 - \Lambda_s$ and $t_{\ell'} \in [t^-, t_0 - \Lambda_s]$,
  and where each pair $t_i, t_i'$ for $i < \ell$ is associated with a path $p_i$
  of non-zero length ending at $u_k$. The construction maintains the following
  properties for all $0 \leq i \leq \ell$:
  \begin{enumerate}[(1)]
  \item For all $t \in [t^-, t_i']$ we have $p_i \in P_{u_k}^s(t)$.
    \label{item:path_ih}
  \item We have
    \begin{equation}
      \xi_{p_i}^s(t_i')
      \geq
      2\drift \Lambda_s
      + \frac{\kappa_{p_i}}{2}
      - (1 + \drift)(t_0 - t_i')
      + \lc{u}(t_0) - \lc{u}(t_i')
      .
      \label{eq:xi_ih_prime}
    \end{equation}
    \label{item:xi_ih_prime}
  \item If $i < \ell$ then we have
    \begin{equation}
      \xi_{p_i}^s(t_{i+1})
      \geq
      2\drift \Lambda_s
      + \frac{\kappa_{p_i}}{2}
      - (1 + \drift)(t_0 - t_{i+1})
      + \lc{u}(t_0) - \lc{u}(t_{i+1})
      .
      \label{eq:xi_ih}
    \end{equation}
    \label{item:xi_ih}
  \end{enumerate}

  \paragraph{Constructing the sequence.}
  We construct the sequence as follows: first, we show that we can find an initial
  path $p_0$ satisfying Properties~\eqref{item:path_ih}
  and~\eqref{item:xi_ih_prime}. Then we show that if we have already constructed
  the sequence up to $i$, such that
  Properties~\eqref{item:path_ih} and~\eqref{item:xi_ih_prime} hold for $i$, then
  we can extend the construction by one step, choosing a time $t_{i+1}$ such that
  \eqref{item:xi_ih} holds at $i$ as well, and selecting a new path $p_{i+1}$ and
  time $t_{i+1}'$ for which Properties~\eqref{item:path_ih}
  and~\eqref{item:xi_ih_prime} hold (until we finally reach some time
  $t_{\ell'}\in [t^-, t_0 - \Lambda_s]$ and the construction halts).

  For the base of the construction we set $p_0 \coloneq (u_0, \ldots, u_k)$ and
  $t_0' := t_0$. For this path we have already seen that $(u_0, \ldots, u_k) \in
  P_{u_k}^s(t)$ for all $t \in [t^-,t_0]$, so Property~\eqref{item:path_ih} is
  satisfied. Also, at time $t_0$ we have $\xi_{p_0}^s(t_0) \geq 2\drift \Lambda_s
  + \kappa_{p_0} / 2$ by choice of $p_0$, and since $t_0' = t_0$, this shows that
  Property~\eqref{item:xi_ih_prime} holds.

  Suppose that we have already constructed the sequence up to time $t_i'$ such
  that Properties~\eqref{item:path_ih} and \eqref{item:xi_ih_prime} hold. In
  particular, at time $t_i'$ we have
  \begin{equation}
    \xi_{p_i}^s(t_i') \geq
    2\drift \Lambda_s
    + \frac{\kappa_{p_i}}{2}
    - (1 + \drift)(t_0 - t_i')
    + \lc{u}(t_0) - \lc{u}(t_i')
    .
    \label{eq:inst_xi_ih}
  \end{equation}
  Suppose further that $t_i' > t_0 - \Lambda_s$ (otherwise the construction halts
  at $i$). Let $p_i = (w = w_0, \ldots, w_m = u_k)$ be the path associated with
  the $i$-th step. We define
  \begin{equation}
    t_{i+1} \coloneq \min \set{ t \in [t_0 - \Lambda_s, t_i'] \st \lc{w}(t_i') - \lc{w}(t) \leq (1 + \drift)(t_i' - t)}.
    \label{eq:t_i_plus_1}
  \end{equation}
  The minimum is taken over a non-empty set because $\lc{w}(t_i') - \lc{w}(t_i')
  \leq (1 + \drift)(t_i' - t_i')$.

  From~\eqref{eq:inst_xi_ih}, we get that
  \begin{eqnarray}
    \xi_{p_i}^s(t_{i+1})
    &=&
    \xi_{p_i}^s(t_i')
    - \left(\lc{w}(t_i') - \lc{w}(t_{i+1})\right)
    + \left( \lc{u_k}(t_i') - \lc{u_k}(t_{i+1}) \right)
    \nonumber
    \\
    &\stackrel{\eqref{eq:inst_xi_ih}}{\geq}&
    2\drift \Lambda_s
    + \frac{\kappa_{p_i}}{2}
    - (1 + \drift)(t_0 - t_i')
    + \lc{u}(t_0) - \lc{u}(t_i')\\
    &&- (1 + \drift)(t_i' - t_{i+1})
    + \lc{u_k}(t_i') - \lc{u_k}(t_{i+1})
    \nonumber
    \\
    &=&
    2\drift \Lambda_s
    + \frac{\kappa_{p_i}}{2}
    - (1 + \drift)(t_0 - t_{i+1})
    + \lc{u}(t_0) - \lc{u}(t_{i+1})
    .
    \label{eq:xi_at_t_i_1}
  \end{eqnarray}
  This shows that Property~\eqref{item:xi_ih} holds at $i$.

  If $t_{i+1} = t_0 - \Lambda_s$, then we define $t_{i+1}' \coloneq
  t_{i+1}$ and $p_{i+1} \coloneq p_i$, and we are done. Thus, suppose
  that $t_{i+1} > t_0 - \Lambda_s$. In this case~\eqref{eq:xi_at_t_i_1}
  shows that $\xi_{p_i}^s(t_{i+1}) > 0$. Consequently, since $u_k$
  satisfies the level $s$ stabilization condition throughout $[t^-,
  t_0]$, we can use Lemma~\ref{lemma:shortpaths} to show that
  $\kappa_{p_i} \leq C_{s-1}$.

  For the $(i+1)$-th step, we choose path $p_{i+1}$ using
  Lemma~\ref{lemma:neighbors_clocks}, but first we must establish the conditions
  of the lemma. From I.H.~\eqref{item:path_ih}, for all $t \in [t^-, t_{i+1}]
  \subseteq [t^-, t_i']$ we have $p_i \in P_u^s(t)$.
  Hence, by the conditions of the lemma, node $w$ satisfies the level $s$
  stabilization condition  throughout $[t^-, t_{i+1}]$. In particular,
  this implies that
  \begin{equation*}
    \Tset_s \cap [\lc{w}(t^-),\lc{w}(t_{i+1})] = \emptyset,
  \end{equation*}
  because the level $s$ stabilization condition implies that node $w$'s logical
  clock does not cross any update point $T_s^{e}$ for any edge $e$ at any
  point throughout the interval $[t^-, t_{i+1}]$.

  This allows us to apply Lemma~\ref{lemma:neighbors_clocks}, which shows that
  \begin{equation}
    \exists w' \in N_w^s(t_{i+1})
    \medspace : \medspace
    \lc{w'}(t_{i+1}) - \lc{w}(t_{i+1})
    >
    \left( s + \frac{1}{2} \right)\kappa_{\set{w, w'}} + \mu(1 + \drift)\tau_{\set{w, w'}},
    \label{eq:switch_fwd}
  \end{equation}
  or
  \begin{equation}
    \forall v \in N_w^s(t_{i+1})
    \medspace : \medspace
    \lc{w}(t_{i+1}) - \lc{v}(t_{i+1})
    <
    \left( s + \frac{1}{2} \right)
    \kappa_{\set{v,w}}.
    \label{eq:switch_back}
  \end{equation}
  We consider each case separately.

  First, suppose that~\eqref{eq:switch_fwd} holds for some node
  $w'$ and define $p_{i+1} \coloneq (w',w) \circ
  p_i = (w', w = w_0, \ldots, w_m = u_k)$ and
  $t_{i+1}' = t_{i+1} - \tau_{\set{w, w'}}$.
  We call this a \emph{forward step}, as the distance to $u_k$ increases.

  For the new path $p_{i+1}$ we have
  \begin{align}
    \xi_{p_{i+1}}^s(t_{i+1})
    &\medspace = \medspace
    \xi_{p_i}^s(t_{i+1})
    + \lc{w'}(t_{i+1})
    - \lc{w}(t_{i+1})
    - s \cdot \kappa_{\set{w, w'}}
    \nonumber
    \\
    & \medspace \stackrel{\mathclap{\eqref{eq:switch_fwd}}}{>} \medspace
    \xi_{p_i}^s(t_{i+1})
    + \frac{\kappa_{\set{w,w'}}}{2}
    + \mu(1 + \drift)\tau_{\set{w, w'}}
    \nonumber
    \\
    & \medspace \stackrel{\mathclap{\eqref{eq:xi_at_t_i_1}}}{\geq} \medspace
    2\drift \Lambda_s
    + \frac{\kappa_{p_{i+1}}}{2}
    - (1 + \drift)(t_0 - t_{i+1})
    + \lc{u}(t_0) - \lc{u}(t_{i+1})
    + \mu(1 + \drift)\tau_{\set{w, w'}}
    \nonumber
    \\
    & \medspace > 2\drift \Lambda_s-(1+\drift)(t_0 - t_{i+1})+ (1 - \drift)(t_0 -
    t_{i+1})
    \nonumber
    \\
    & \medspace > 0.
    \label{eq:xi_after_fwd_switch}
  \end{align}

  By assumption we have $t_{i+1} > t_0 - \Lambda_s$, and hence $t_{i+1} > t^- +
  \stabint{s}$. Because $u_k$ satisfies the $(C,s)$-stabilization condition
  throughout $[t^-, t_{i+1}]$, we can apply Lemma~\ref{lemma:paths} to the
  interval $[t^- + \stabint{s}, t_{i+1}]$ to show that for all $t \in [t^-,
  t_{i+1}']$ we have $p_{i+1} \in P_{u_k}^s(t)$, so Property~\eqref{item:path_ih}
  holds. The lemma also shows that $t_{i+1}'\geq t^-$.

  Going back to time $t_{i+1}' $, we have
  \begin{align*}
    \xi_{p_{i+1}}^s(t_{i+1}')
    &\medspace = \medspace
    \xi_{p_{i+1}}^s(t_{i+1})
    - \left( \lc{w'}(t_{i+1}) - \lc{w'}(t_{i+1}') \right)
    + \left( \lc{u_k}(t_{i+1}) - \lc{u_k}(t_{i+1}') \right)
    \\
    &\medspace \geq \medspace
    \xi_{p_{i+1}}^s(t_{i+1})
    - (1 + \mu)(1 + \drift)(t_{i+1} - t_{i+1}')
    + \lc{u_k}(t_{i+1}) - \lc{u_k}(t_{i+1}')
    \\
    &\medspace \stackrel{\mathclap{\eqref{eq:xi_after_fwd_switch}}}{>} \medspace
    2\drift \Lambda_s
    + \frac{\kappa_{p_{i+1}}}{2}
    - (1 + \drift)(t_0 - t_{i+1})
    + \lc{u}(t_0) - \lc{u}(t_{i+1})
    + \mu(1 + \drift)\tau_{\set{w, w'}}
    \\
    &\qquad
    - (1 + \mu)(1 + \drift)\tau_{\set{w, w'}}
    + \lc{u_k}(t_{i+1}) - \lc{u_k}(t_{i+1}')
    \\
    &\medspace = \medspace
    2\drift \Lambda_s
    + \frac{\kappa_{p_{i+1}}}{2}
    - (1 + \drift)(t_0 - t_{i+1} + \tau_{\set{w, w'}})
    + \lc{u}(t_0) - \lc{u}(t_{i+1}')
    \\
    &\medspace = \medspace
    2\drift \Lambda_s
    + \frac{\kappa_{p_{i+1}}}{2}
    - (1 + \drift)(t_0 - t_{i+1}')
    + \lc{u}(t_0) - \lc{u}(t_{i+1}')
    .
  \end{align*}
  This shows that Property~\eqref{item:xi_ih_prime} holds for $i + 1$.

  Now let us turn to the other case, in which~\eqref{eq:switch_back} holds.
  From I.H.~\eqref{item:path_ih} we know that $p_i \in P_{u_k}^s(t_{i+1})$, and in
  particular, for the next node $w_1$ on the path $(w = w_0, w_1, \ldots, w_m =
  u_k)$, we have
  $w_1 \in \N_w^s(t_{i+1})$.
  Thus,~\eqref{eq:switch_back} shows that
  \begin{equation}
    \lc{w}(t_{i+1}) - \lc{w_1}(t_{i+1})
    <
    \left( s + \frac{1}{2} \right)\kappa_{\set{w, w_1}}.
    \label{eq:backtrace}
  \end{equation}

  In this case, we define $p_{i+1} \coloneq (w_1, \ldots, w_m = u_k)$, that is, we
  remove $w$ from the head of the path, and $t_{i+1}' \coloneq t_{i+1}$. We call
  this a \emph{backward step}. Property~\eqref{item:path_ih} for $i + 1$ follows
  immediately from I.H.~\eqref{item:path_ih} (for $i$). As for
  Property~\eqref{item:xi_ih_prime}, we have
  \begin{align*}
    \xi_{p_{i+1}}^s(t_{i+1})
    &\medspace = \medspace
    \xi_{p_i}^s(t_{i+1})
    - \lc{w}(t_{i+1})
    + \lc{w_1}(t_{i+1})
    + s \cdot \kappa_{\set{w, w_1}}
    \\
    &\medspace \stackrel{\mathclap{\eqref{eq:backtrace}}}{>} \medspace
    \xi_{p_i}^s(t_{i+1})
    - \frac{\kappa_{\set{w, w_1}}}{2}
    \\
    &\medspace \stackrel{\mathclap{\eqref{eq:xi_at_t_i_1}}}{\geq} \medspace
    2 \drift \Lambda_s
    + \frac{\kappa_{p_i}}{2}
    - \frac{\kappa_{\set{w, w_1}}}{2}
    - (1 + \drift)(t_0 - t_{i+1})
    + \lc{u}(t_0) - \lc{u}(t_{i+1})
    \\
    &\medspace = \medspace
    2 \drift \Lambda_s
    + \frac{\kappa_{p_{i+1}}}{2}
    - (1 + \drift)(t_0 - t_{i+1})
    + \lc{u}(t_0) - \lc{u}(t_{i+1})
    .
  \end{align*}
  Since $t_{i+1}' = t_{i+1}$, Property~\eqref{item:xi_ih_prime} is satisfied for $i + 1$.
  Note in particular that we have $\xi_{p_{i+1}}^s(t_{i+1}) > 0$, and hence $p_{i+1} \neq (u_k)$, because
  $\xi_{(u_k)}^s(t) = 0$ for all times $t$.
  This concludes the induction.

  We note that the sequence we constructed is finite, that is, there is some $\ell \in \nat$ such that
  $t_{\ell}' \in [t^-, t_0 - \Delta]$.
  This is because every time we make a forward step we have $t_{i+1}' \leq t_{i+1} - \min_{e \in \binom{V}{2}} \tau_e$,
  so after finitely many such steps we reach time $t_0 - \Delta$;
  as for backward steps, each such step shortens the path, so only finitely many backward steps
  can occur between two forward steps.

  \paragraph{Properties of the chain construction.}
  Let $v_0, \ldots, v_{\ell - 1}$ denote the first node on each path $p_0, \ldots, p_{\ell - 1}$ in the sequence above.
  Before proceeding, we establish the following additional properties of the chain $v_0, \ldots, v_{\ell - 1}$.
  \begin{enumerate}[(1)]
    \setcounter{enumi}{3}
  \item For all $i = 0, \ldots, \ell - 1$ we have $\kappa_{p_i} \leq C_{s-1}$.
    \label{item:kappa}
  \item For all $i = 0, \ldots, \ell - 1$, node $v_i$ satisfies the level $s$ stabilization condition throughout $[t^-, t_i']$.
    \label{item:stab}
  \item For all $i = 0, \ldots, \ell - 1$ we have $\Xi_{v_i}(t_i') \geq
  2\drift(t_i' - (t_0 - \Lambda_s))$ and for all $t \in (t_{i+1}, t_i']$ we have
  $\Xi_{v_i}^s(t) > 0$.
    \label{item:Xi}
  \item There is an index $m \in \set{0, \ldots, \ell - 1}$ such that in the construction
    all steps prior to index $m$ are backward steps and all steps starting from index $m$ are forward steps.
    \label{item:fwd_suffix}
  \item
    For all $m \leq i < j \leq \ell$, we have
    \begin{equation}
      \lc{v_i}(t_{i+1})
      -
      \lc{v_j}(t_j')
      \leq
      (1 + \drift)(t_{i+1} - t_j')
      -
      \left( s + \frac{1}{2} \right)\kappa_{(v_i,\ldots, v_j)}
      .
      \label{eq:chain}
    \end{equation}
    \label{item:span}
  \end{enumerate}
  The proof of these properties follows.

  Fix $i \leq \ell - 1$. We know that node $u$ satisfies the level $s$
  stabilization condition throughout $[t^-, t^+]$. From
  Property~\eqref{item:path_ih} for index $i$ we have $(v_i, \ldots, u_k) = p_i
  \in P_{u_k}^s(t)$ for all $t \in [t^-, t_i]$. Also,
  Property~\eqref{item:xi_ih_prime} states that
  \begin{align*}
    \xi_{p_i}^s(t_i') &\geq 2\drift \Lambda_s - (1 + \drift)(t_0 - t_i') +
    \lc{u_k}(t_0) - \lc{u_k}(t_i')
    \\
    &
    \geq 2\drift(t_i' - (t_0 - \Lambda_s)) \geq 0
  \end{align*}
  Consequently, Lemma~\ref{lemma:shortpaths} shows that $\kappa_{p_i} \leq C_{s -
    1}$, so Property~\eqref{item:kappa} holds, and Property~\eqref{item:stab}
  follows from the conditions of the current lemma.

  Because $p_i \in P_{u_k}^s(t_i')$ we have $\Xi_{v_i}^s(t_i') \geq
  \xi_{p_i}^s(t_i')$, and hence $\Xi_{v_i}^s(t_i') \geq 2\drift( t_i' - (t_0 -
  \Lambda_s))$. Also, because $t_i' - (t_0 - \Lambda_s) \geq t_i' - t_{i+1}$, we can
  apply Lemma~\ref{lemma:Xi_pos} to show that $\Xi_{v_i}^s(t) > 0$ throughout
  $(t_{i+1}, t_i']$. This shows that Property~\eqref{item:Xi} holds.

  To show Property~\eqref{item:fwd_suffix}, we show that if a forward step occurs
  at index $i < \ell - 1$ of the construction, then at index $i + 1$ we also take
  a forward step. The property follows.

  Suppose that this is not the case, that is, at some index $i < \ell - 1$ we have
  \begin{equation*}
    \lc{v_i}(t_i) - \lc{v_{i-1}}(t_i) > 
    \left( s + \frac{1}{2} \right)\kappa_{\set{v_{i-1}, v_i}} + \mu(1 + \drift)\tau_{\set{v_{i-1}, v_i}}
  \end{equation*}
  and at index $i + 1$ we have $v_{i+1} = v_{i-1}$ and
  \begin{equation*}
    \lc{v_i}(t_{i+1}) - \lc{v_{i-1}}(t_{i+1}) < \left( s + \frac{1}{2} \right)\kappa_{\set{v_{i-1}, v_i}}.
  \end{equation*}
  We show that this implies that node $v_{i-1}$'s average rate over the interval $[t_{i+1}, t_i]$ was no greater than $1 + \drift$,
  contradicting the choice of $t_i$ as the \emph{minimal} time such that $v_{i-1}$'s average rate over $[t_i, t_{i-1}']$ did not exceed $1 + \drift$.

  Summing the two inequalities above yields
  \begin{equation}
    \lc{v_{i-1}}(t_i) - \lc{v_{i-1}}(t_{i+1})
    <
    \lc{v_i}(t_i) - \lc{v_i}(t_{i+1})
    - \mu(1 + \drift)\tau_{\set{v_{i-1}, v_i}}.
    \label{eq:untitled1}
  \end{equation}
  By definition of $t_{i+1}$ we have $\lc{v_i}(t_i') - \lc{v_i}(t_{i+1}) \leq (1 + \drift)(t_i' - t_{i+1})$,
  and since at index $i$ we took a forward step, we defined $t_i' = t_i - \tau_{\set{v_{i-1}, v_i}}$.
  Hence
  \begin{align*}
    \lc{v_i}(t_i) - \lc{v_i}(t_{i+1})
    &=
    \lc{v_i}(t_i) - \lc{v_i}(t_i')
    + \lc{v_i}(t_i') - \lc{v_i}(t_{i+1})
    \\
    &\leq
    (1 + \drift)(1 + \mu)\tau_{\set{v_{i-1}, v_i}}
    + (1 + \drift)(t_i - \tau_{\set{v_{i-1}, v_i}} - t_{i+1})
    \\
    &=
    (1 + \drift)(t_i - t_{i+1}) + \mu(1 + \drift)\tau_{\set{v_{i-1}, v_i}}
    .
  \end{align*}
  Combining with~\eqref{eq:untitled1} yields
  \begin{align*}
    \lc{v_{i-1}}(t_i) - \lc{v_{i-1}}(t_{i+1})
    &<
    (1 + \drift)(t_i - t_{i+1}) + \mu(1 + \drift)\tau_{\set{v_{i-1}, v_i}}
    -
    \mu(1 + \drift)\tau_{\set{v_{i-1}, v_i}}
    \\
    &\leq
    (1 + \drift)(t_i - t_{i+1}).
  \end{align*}
  This is a contradiction.

  This shows that after the first forward step in the construction (if one occurs),
  no backward steps can occur. Thus, there is some index $m \in \set{ 0, \ldots, \ell - 1}$ such that for all
  $i = m, \ldots, \ell - 2$,
  node $v_{i+1}$ is obtained from node $v_i$ by a forward step at time $t_{i+1}$.
  (If no forward steps occur in the construction then we set $m \coloneq \ell - 1$.)

  Finally we show Property~\eqref{item:span}. Fix $i, j$ such that $m \leq i < j \leq \ell - 1$.
  All steps between index $m$ and index $\ell - 1$ are forward steps, and hence
  for each $k = i, \ldots, j - 1$ we have
  \begin{equation*}
    \lc{v_{k+1}}(t_{k+1}) - \lc{v_k}(t_{k+1})
    >
    \left( s + \frac{1}{2} \right) \kappa_{\set{v_k, v_{k+1}}}
    + \mu(1 + \drift)\tau_{\set{v_k, v_{k+1}}}
    .
  \end{equation*}
  Also, because $t_{k+1}' = t_{k+1} - \tau_{\set{v_k, v_{k+1}}}$ we have
  \begin{equation*}
    \lc{v_{k+1}}(t_{k+1})
    -
    \lc{v_{k+1}}(t_{k+1}')
    \leq
    (1 + \drift)(1 + \mu)\tau_{\set{v_k, v_{k+1}}},
  \end{equation*}
  and by definition of $t_{k+2}$,
  \begin{equation*}
    \lc{v_{k+1}}(t_{k+1}') - \lc{v_{k+1}}(t_{k+2})
    \leq
    (1 + \drift)(t_{k+1}' - t_{k+2}).
  \end{equation*}
  Summing the three inequalities above yields, for each $k = i, \ldots, j-2$,
  \begin{equation*}
    \lc{v_k}(t_{k+1}) - \lc{v_{k+1}}(t_{k+2})
    <
    (1 + \drift)(t_{k+1} - t_{k+2})
    -
    \left( s + \frac{1}{2} \right)\kappa_{\set{v_k, v_{k+1}}}
    .
  \end{equation*}
  Summing over $k = i, \ldots, j - 2$, we obtain
  \begin{equation*}
    \lc{v_i}(t_{i+1}) - \lc{v_{j - 1}}(t_j)
    \leq
    (1 + \drift)(t_i - t_j)
    -
    \left( s + \frac{1}{2} \right)\kappa_{(v_i, \ldots, v_{j-1})}
    .
  \end{equation*}
  For the final step, from $j - 1$ at time $t_j$ to $j$ at time $t_j'$, we use
  only the first two inequalities, which show that
  \begin{equation*}
    \lc{v_{j-1}}(t_j) - \lc{v_j}(t_j') \leq
    (1 + \drift)\tau_{\set{v_{j-1}, v_j}}
    -
    \left( s + \frac{1}{2} \right)\kappa_{\set{v_{j-1}, v_j}}
    .
  \end{equation*}
  Since $t_j - t_j' = \tau_{\set{v_{j-1}, v_j}}$, we combine the two inequalities above to obtain~\eqref{eq:chain}, as desired.

  \paragraph{Bounding $\Xi^s$.}
  The chain construction provides us with a sequence of sub-intervals $\set{ [t_0,
    t_m']} \cup \set{ [t_{i+1}, t_i'] \st i = m, \ldots, \ell - 2} \cup \set{ [t_0 -
    \Lambda_s, t_{\ell - 1}']}$, each associated with a node $v_i$ that has a
  non-negative $\Xi^s$-value and a small average clock rate (at most $1 + \drift$)
  over the entire sub-interval. Roughly speaking, over each
  sub-interval, Lemma~\ref{lemma:pull} shows that $\Xi^s$ decreases at an average
  rate of at least $(1 - \drift)(1 + \mu) - (1 + \drift) = (1 - \drift)\mu -
  2\drift$; since at the end of the whole interval (time $t_0$) we had
  $\Xi_{v_0}^s(t_0) \geq 2\drift C_{s-1} / ((1 - \drift)\mu)$, at the beginning
  (time $t_0 - \Lambda_s = t_0 - C_{s-1} / ( (1 - \drift)\mu)$) we will be able to
  show that for some node we have $\Xi^s(t_0 - \Lambda_s) \geq \Xi_{v_0}^s(t_0) +
  ((1 - \drift)\mu - 2\drift)\cdot \Lambda_s > C_{s-1}$, which will result in a
  contradiction to $(s-1)$-legality at time $t_0 - \Lambda_s$. Thus, our strategy
  now is to apply Lemma~\ref{lemma:pull} successively to the sub-intervals we
  obtained in the construction as we went back in time from $t_0$ to $t_0 - \Delta
  t$.

  However, when we apply Lemma~\ref{lemma:pull} to a sub-interval $[t_{i+1},
  t_i']$ there is some ``overshoot'': the lemma yields a time $t' \in [t_{i+1} -
  \stabint{s}, t_{i+1}]$ such that over the interval $[t', t_i']$, $\Xi^s$
  decreases quickly, but $t'$ falls in some other sub-interval $[t_{j+1}, t_j]$.
  Before we can apply Lemma~\ref{lemma:pull} again, we must deal with two
  concerns:
  \begin{itemize}
  \item The new interval $[t_{j+1}, t_j]$ is not necessarily contiguous to the
    interval $[t_{i+1}, t_i']$ to which we applied Lemma~\ref{lemma:pull}; that
    is, we can have $j > i + 1$. Thus we define a sequence of indices $i_0 \leq
    \ldots \leq i_h$ representing the indices of the sub-intervals in which we
    find ourselves after each application of Lemma~\ref{lemma:pull}.
  \item We have $t' \in [t_{j+1}, t_j]$, but we do not know whether $t' \in
    [t_{j+1}, t_j']$ or $t' \in (t_j', t_j]$.
    Our construction in the previous part only ensures that $v_j$ has an average
    rate of at most $1 + \drift$ over the interval $[t_{j+1}, t_j']$; hence, if we
    have $t' \in (t_j', t_j]$, we must first ``go back in time'' to $t_j'$ before
    we can usefully apply Lemma~\ref{lemma:pull} to $v_j$.
  \end{itemize}
  In other words, in the current part of the proof, we successively apply two kinds
  of steps: the first is an application of Lemma~\ref{lemma:pull} to obtain a time
  $t'$ for which we have a large average $\Xi^s$-increase rate; the second is a
  step back in time to the nearest preceding time $t_j' \leq t'$, in preparation
  for the next application of Lemma~\ref{lemma:pull}.

  Accordingly, we define two sequences of times, $\theta_0 \geq \ldots \geq
  \theta_h$ and $\phi_0 \geq \ldots \geq \phi_h$, such that
  \begin{equation*}
    t_{m+1} \geq \phi_0 \geq \theta_1 \geq \phi_1 \geq \ldots \geq \theta_h \geq t_0 - \Lambda_s \geq  \phi_h \geq t^-.
  \end{equation*}
  The first sequence $\phi_0, \ldots, \phi_h$ represents the times obtained by
  successively applying Lemma~\ref{lemma:pull}; the second sequence $\theta_1,
  \ldots, \theta_h$ represents the second step back in time. Each $\theta_j$ is
  chosen such that for some $k \leq \ell - 1$ we have $\theta_j \in [t_{k+1},
  t_k']$. Note that since at each index $j \geq m$ we take a forward step, we
  always have $t_{j+1} \leq t_j' < t_j$ (this is a property of the construction),
  and hence the index $k$ is unique. We define a sequence of indices $i_0 \leq
  \ldots \leq i_h$ as follows:
  \begin{equation}
    i_j \coloneq
    \begin{cases}
      m
      & \text{ if $j = 0$,}\\
      \text{the unique index k such that $\theta_j \in [t_{k+1}, t_k']$}
      & \text{ if $j > 0$.}\\
    \end{cases}
    \label{eq:i_j}
  \end{equation}
  Finally, each $\phi_j$ is chosen such that $\phi_j \in [t_{i_j+1} - \stabint{s},
  t_{i_j+1}]$ (as $\phi_j$ is obtained by applying Lemma~\ref{lemma:pull} to the
  interval $[t_{i_j+1}, \theta_j]$).

  We maintain the following properties:
  \begin{enumerate}[(i)]
  \item For all $j = 0, \ldots, h$,
    \begin{equation}
      \Xi_{v_{i_j}}^s(\phi_j)
      \geq
      2 \drift \Lambda_s
      + ( (1 - \drift)\mu - 2\drift)(t_0 - t_{i_j+1})
      + (1 - \drift + 2\mu)(t_{i_j+1} - \phi_j)
      - \left( \lc{v_{i_j}}(t_{i_j+1}) - \lc{v_{i_j}}(\phi_j) \right)
      .
      \label{eq:pull}
    \end{equation}
    \label{item:phi}
  \item For all $j = 1, \ldots, h$,
    \begin{equation}
      \Xi_{v_{i_j}}^s(\theta_j)
      \geq
      2 \drift \Lambda_s
      + ( (1 - \drift)\mu - 2\drift)(t_0 - t_{i_j}')
      + (1 - \drift)(1 + \mu)(t_{i_j}' - \theta_j)
      - \left( \lc{v_{i_j}}(t_{i_j}') - \lc{v_{i_j}}(\theta_j) \right)
      % + (1 + \drift)\mu(t_{i_{j-1} + 1} - \theta_j)
      .
      \label{eq:xi_ih2}
    \end{equation}
    \label{item:theta}
  \end{enumerate}
  Intuitively, at each point $\phi_j$ we claim an average increase rate of $( (1 -
  \drift)\mu - 2\drift)$ for $\Xi^s$ as we go back in time over the interval
  $[\phi_j, t_0]$: this follows from~\eqref{eq:pull}, because
  $\lc{v_{i_j}}(t_{i_j+1}) - \lc{v_{i_j}}(\phi_j) \leq ( 1 + \drift)(1 +
  \mu)(t_{i_j+1} - \phi_j)$. However, the statement of Property~\eqref{item:phi}
  is more precise, and keeps track of the exact clock increase of $v_{i_j}$.
  We need this additional information because the chain construction relates the
  clock values of the different nodes at points where we switch from one to the
  other (Property~\eqref{item:span}), not at arbitrary times such as $\phi_j$.
  Thus, we keep track of $v_{i_j}$'s increase from time $\phi_j$ to time
  $t_{i_j+1}$, at which the switch occurs. This leads us to show that because
  $\Xi_{v_{i_j}}^s(\phi_j)$ is large, so is $\Xi_{v_{i_{j+1}}}^s(\phi_j)$, and
  allows the induction to go through.

  In Property~\eqref{item:theta} we also keep track of $v_{i_j}$'s exact
  clock increase, but here the reason is different:
  % because $\theta_j$ is some arbitrary time in the interval $[t_{i_j+1}, t_{i_j}']$,
  we simply do not have a good enough bound on $v_{i_j}$'s clock increase over the
  interval $[\theta_j, t_{i_j}']$. The choice of $t_{i_j+1}$ yields an average
  rate of at most $1 + \drift$ over $[t_{i_j+1}, t_{i_j}']$, but tells us nothing
  about arbitrary points in the interval, such as $\theta_j$. Therefore, we keep
  track of $v_{i_j}$'s clock increase over $[\theta_j, t_{i_j}']$, and also of the
  decrease $(1 - \drift)(1 + \mu)(t_{i_j'} - \theta_j)$ provided by
  Lemma~\ref{lemma:pull}, which represents the fact that the other endpoints of
  paths maximizing $\Xi^s$ are acting to catch up over $[\theta_j, t_{i_j}']$.
  Later, when we go back in time to $\phi_{j+1} \leq t_{i_j+1}$, we ``complete''
  the sub-interval, and use the fact that $v_{i_j}$'s average rate over
  $[t_{i_j+1}, t_{i_j}']$ is at most $1 + \drift$ to obtain the average decrease
  rate of $( (1 - \drift)\mu - 2\drift)$ over all of $[\phi_j, t_0]$.

  The definition of the two sequences is mutually-recursive. We begin by showing
  how $\phi_j$ is chosen, assuming that if $j > 0$ then we have already chosen
  $\theta_j$ such that Property~\eqref{item:theta} holds.

  First, consider the base case, $j = 0$: we must find a time $\phi_0 \in [t_{m+1}
  - \stabint{s}, t_{m+1}]$ satisfying Property~\eqref{item:phi}.
  This requires a bit more effort than the step, because we are claiming an
  average decrease rate of $( (1 - \drift)\mu - 2\drift)$ for $\Xi^s$ over an
  interval $[\phi_0, t_0] \supseteq [t_{m+1}, t_0]$; but since we skipped the
  prefix $v_0, \ldots, v_{m-1}$ of the chain construction, the properties of the
  chain tell us nothing about $v_m$ during the sub-interval $[t_m, t_0]$.
  Nevertheless, we show that we can apply Lemma~\ref{lemma:pull} to the entire
  interval $[t_{m+1}, t_0]$.

  Because only backward steps occur prior to index $m$, the path $p_m$ is a
  sub-path of $p_0$; we know that $p_0 \in P_{u_k}^s(t)$ for all $t \in [t^-,
  t_0]$, and hence $p_m \in P_{u_k}^s(t)$ for all $t \in [t^-, t_0]$ as well.
  Also, from Property~\eqref{item:xi_ih} of the construction,
  \begin{equation}
    \xi_{p_m}^s(t_{m+1}) \geq
    2\drift \Lambda_s
    - (1 + \drift)(t_0 - t_{m+1})
    + \lc{u_k}(t_0) - \lc{u_k}(t_{m+1}),
    \label{eq:xi_t_m_1}
  \end{equation}
  so going forward to any time $t \in [t_{m+1}, t_0]$ we have
  \begin{align*}
    \Xi_{v_m}^s(t)
    &\geq
    \xi_{p_m}^s(t)
    =
    \xi_{p_m}^s(t_{m+1})
    + \left(\lc{v_m}(t) - \lc{v_m}(t_{m+1}) \right)
    - \left( \lc{u_k}(t) - \lc{u_k}(t_{m+1}) \right)
    \\
    &\geq
    2\drift \Lambda_s
    - (1 + \drift)(t_0 - t_{m+1})
    + \left( \lc{v_m}(t) - \lc{v_m}(t_{m+1}) \right)
    + \left( \lc{u_k}(t_0) - \lc{u_k}(t) \right)
    \\
    &\geq
    2\drift \Lambda_s
    - (1 + \drift)(t_0 - t_{m+1})
    + (1 - \drift)(t - t_{m+1})
    + (1 - \drift)(t_0 - t)
    \\
    &\geq
    2\drift(\Lambda_s - (t_0 - t_{m+1})),
  \end{align*}
  which is strictly greater than $0$ for $t\neq t_0$. This is sufficient
  to apply Lemma~\ref{lemma:pull} to the interval $[t_{m+1},t_0]$, yielding a
  time $\phi_0 \in [t_{m+1} - \stabint{s}, t_{m+1}]$ such that
  \begin{align*}
    \Xi_{v_m}^s(\phi_0)
    &
    \medspace
    \geq
    \medspace
    \Xi_{v_m}^s(t_0)
    + (1 - \drift)(1 + \mu)(t_0 - \phi_0)
    - \left( \lc{v_m}(t_0) - \lc{v_m}(\phi_0) \right)
    + \mu(1 + \drift)(t_{m+1} - \phi_0)
    \\
    &
    \medspace
    \geq
    \medspace
    \xi_{p_m}^s(t_0)
    + (1 - \drift)(1 + \mu)(t_0 - \phi_0)
    - \left( \lc{v_m}(t_0) - \lc{v_m}(\phi_0) \right)
    + \mu(1 + \drift)(t_{m+1} - \phi_0)
    \\
    &
    \medspace
    \geq
    \medspace
    \xi_{p_m}^s(t_{m+1})
    + \left(\lc{v_m}(t_0) - \lc{v_m}(t_{m+1}) \right)
    - \left( \lc{u_k}(t_0) - \lc{u_k}(t_{m+1}) \right)
    \\
    &\qquad
    + (1 - \drift)(1 + \mu)(t_0 - \phi_0)
    - \left( \lc{v_m}(t_0) - \lc{v_m}(\phi_0) \right)
    + \mu(1 + \drift)(t_{m+1} - \phi_0)
    \\
    &
    \medspace
    \stackrel{\mathclap{\eqref{eq:xi_t_m_1}}}{\geq}
    \medspace
    2\drift \Lambda_s
    - (1 + \drift)(t_0 - t_{m+1})
    + (1 - \drift)(1 + \mu)(t_0 - \phi_0)
    - \left( \lc{v_m}(t_{m+1}) - \lc{v_m}(\phi_0) \right)
    \\
    &\qquad
    + \mu(1 + \drift)(t_{m+1} - \phi_0)
    \\
    &
    \medspace
    =
    \medspace
    2 \drift \Lambda_s
    +
    ((1 - \drift)\mu - 2\drift)(t_0 - t_{m+1})
    + (1 - \drift + 2\mu)(t_{m+1} - \phi_0)\\
    &\qquad- \left( \lc{v_m}(t_{m+1}) - \lc{v_m}(\phi_0) \right)
    .
  \end{align*}
  This shows that Property~\eqref{item:phi} holds for this choice of $\phi_0$.

  For the step, suppose that at $j \geq 1$ we have already chosen a time
  $\theta_j \geq t_0 - \Lambda_s$ such that
  \begin{equation*}
    \Xi_{v_{i_j}}^s(\theta_j)
    \geq
    2 \drift \Lambda_s
    + ( (1 - \drift)\mu - 2\drift)(t_0 - t_{i_j}')
    + (1 - \drift)(1 + \mu)(t_{i_j}' - \theta_j)
    \nonumber
    - \left( \lc{v_{i_j}}(t_{i_j}') - \lc{v_{i_j}}(\theta_j) \right)
    .
  \end{equation*}
  Since $\theta_j \leq t_{i_j}'$ by choice of $i_j$, Properties~\eqref{item:stab}
  and~\eqref{item:Xi} show that we have $\Xi_{v_{i_j}}^s(t) > 0$ for all $t \in
  (t_{i_j + 1}, \theta_j]$ and that $v_{i_j}$ satisfies the
  $(C,s)$-stabilization condition during $[t_{i_j + 1}, \theta_j]$. Thus, we can
  apply Lemma~\ref{lemma:pull} to obtain a time $\phi_j \in [t_{i_j+1} - \stabint{s},
  t_{i_j+1}]$ such that
  \begin{equation*}
    \Xi_{v_{i_j}}^s(\phi_j)
    \geq
    \Xi_{v_{i_j}}^s(\theta_j)
    + (1 - \drift)(1 + \mu)(\theta_j - \phi_j)
    - \left( \lc{v_{i_j}}(\theta_j) - \lc{v_{i_j}}(\phi_j) \right)
    + \mu(1 + \drift)(t_{i_j+1} - \phi_j).
  \end{equation*}
  Together with the induction hypothesis, we obtain
  \begin{align*}
    \Xi_{v_{i_j}}^s(\phi_j)
    &\geq
    2 \drift \Lambda_s
    + ( (1 - \drift)\mu - 2\drift)(t_0 - t_{i_j}')
    + (1 - \drift)(1 + \mu)(t_{i_j}' - \theta_j)
    + (1 - \drift)(1 + \mu)(\theta_j - \phi_j)
    \nonumber
    \\
    &\quad
    % - \left( \lc{v_{i_j}}(t_{i_j}') - \lc{v_{i_j}}(\theta_j) \right)
    % - \left( \lc{v_{i_j}}(\theta_j) - \lc{v_{i_j}}(\phi_{j+1}) \right)
    - \left( \lc{v_{i_j}}(t_{i_j}') - \lc{v_{i_j}}(\phi_j) \right)
    + \mu(1 + \drift)(t_{i_j+1} - \phi_j)
    \\
    &=
    2 \drift \Lambda_s
    + ( (1 - \drift)\mu - 2\drift)(t_0 - t_{i_j}')
    + (1 - \drift)(1 + \mu)(t_{i_j}' - t_{i_j+1} + t_{i_j+1} - \phi_j)
    \\
    &\quad
    - \left( \lc{v_{i_j}}(t_{i_j}') - \lc{v_{i_j}}(t_{i_j+1}) \right)
    - \left( \lc{v_{i_j}}(t_{i_j+1}) - \lc{v_{i_j}}(\phi_j) \right)
    + \mu(1 + \drift)(t_{i_j+1} - \phi_j)
    .
  \end{align*}
  Recall that by definition of $t_{i_j+1}$ we have $\lc{v_{i_j}}(t_{i_j}') -
  \lc{v_{i_j}}(t_{i_j+1}) \leq (1 + \drift)(t_{i_j}' - t_{i_j+1})$. Thus, we have
  \begin{equation*}
    \Xi_{v_{i_j}}^s(\phi_j)
    \geq
    2 \drift \Lambda_s
    + ( (1 - \drift)\mu - 2\drift)(t_0 - t_{i_j+1})
    + (1 - \drift + 2\mu)(t_{i_j+1} - \phi_j)
    - \left( \lc{v_{i_j}}(t_{i_j+1}) - \lc{v_{i_j}}(\phi_j) \right)
    .
  \end{equation*}
  Note that since $\phi_j \geq t_{i_j + 1} - \stabint{s} \geq t_0 - \Lambda_s -
  \stabint{s}$, we have $\phi_j \geq t^-$, as desired. This completes the
  induction step for the sequence $\phi_j$.

  Next we show how $\theta_{j+1}$ is chosen, assuming that we have already chosen
  a time $\phi_j$ satisfying Property~\eqref{item:phi}. Assume also that $\phi_j >
  t_0 - \Lambda_s$, as otherwise the construction halts. Recall that we must choose
  $\theta_{j+1}$ such that $\theta_{j+1} \in [t_{k+1}, t_k']$ for some $k$. Thus,
  if $\phi_j \in [t_{k+1}, t_k']$ for some $k$, then we set $\theta_j \coloneq
  \phi_j$; otherwise it must be that $\phi_j \in (t_k', t_k]$ for some unique $k$,
  and in this case we define $\theta_{j+1} \coloneq t_k'$.

  The choice of $\theta_{j+1}$ induces an index $i_{j+1}$ (which is the minimal
  index $k$ from the definition of $\theta_{j+1}$). The induction hypothesis
  (Property~\eqref{item:phi}) states that $\Xi_{v_{i_j}}^s(\phi_j)$ is large, that
  is, there is some path $q \in P_{v_{i_j}}^s(\phi_j)$ such that $\xi_q^s(\phi_j)$
  is large. To show Property~\eqref{item:theta} at $j$, we first extend $q$ into a
  path $q'$ starting at $v_{i_{j+1}}$, and show that because $\xi_q^s(\phi_j)$ is
  large, so is $\xi_{q'}^s(\phi_j)$. Then we go back in time and show that
  $\xi_{q'}^s(\theta_{j+1})$ is large (that is, not much skew is lost as we step
  back from $\phi_j$ to $\theta_{j+1}$), and finally we show that $q' \in
  P_{v_{i_{j+1}}}^s(\theta_{j+1})$, which implies that
  $\Xi_{v_{i_{j+1}}}^s(\theta_{j+1}) \geq \xi_{q'}^s(\theta_{j+1})$.

  Formally, let $q = (v_{i_j}, \ldots, x) \in P_{v_{i_j}}^s(t)$ be a path such
  that $\xi_q^s(t) = \Xi_{v_{i_j}}^s(t)$. For the extended path $q' \coloneq
  (v_{i_{j+1}}, \ldots, v_{i_j}) \circ q$, we have
  \begin{align*}
    \xi_{q'}^s(\phi_j)
    &
    \medspace
    =
    \medspace
    \xi_q^s(\phi_j)
    - \lc{v_{i_j}}(\phi_j)
    + \lc{v_{i_{j+1}}}(\phi_j)
    - s \cdot \kappa_{(v_{i_j}, \ldots, v_{i_{j+1}})}
    \\
    &
    \medspace
    =
    \medspace
    \Xi_{v_{i_j}}^s(\phi_j)
    - \lc{v_{i_j}}(\phi_j)
    + \lc{v_{i_{j+1}}}(\phi_j)
    - s \cdot \kappa_{(v_{i_j}, \ldots, v_{i_{j+1}})}
    \\
    &
    \medspace
    \stackrel{\mathclap{\eqref{eq:pull}}}{\geq}
    \medspace
    2 \drift \Lambda_s
    + ( (1 - \drift)\mu - 2\drift)(t_0 - t_{i_j+1})
    + (1 - \drift)(1 + \mu)(t_{i_j+1} - \phi_j)
    \\
    &\qquad
    - \left( \lc{v_{i_j}}(t_{i_j+1}) - \lc{v_{i_{j+1}}}(\phi_j) \right)
    - s \cdot \kappa_{(v_{i_j}, \ldots, v_{i_{j+1}})}
    .
  \end{align*}
  (Note that we omit the term $\mu(1 + \drift)(t_{i_j+1} - \phi_j)$
  in~\eqref{eq:pull}, which is non-negative because $\phi_j \leq t_{i_j + 1}$.)

  Next we deal with the gap
  $\left( \lc{v_{i_j}}(t_{i_j+1}) - \lc{v_{i_{j+1}}}(t) \right)$
  using Property~\eqref{item:span}, which shows that
  \begin{equation*}
    \lc{v_{i_j}}(t_{i_j+1}) - \lc{v_{i_{j+1}}}(t_{i_{j+1}}')
    \leq
    (1 + \drift)(t_{i_j+1} - t_{i_{j+1}}')
    -
    \left( s + \frac{1}{2} \right)\kappa_{(v_{i_j}, \ldots, v_{i_{j+1}})};
  \end{equation*}
  thus we have
  \begin{align*}
    \xi_{q'}^s(\phi_j)
    &\geq
    2 \drift \Lambda_s
    + ( (1 - \drift)\mu - 2\drift)(t_0 - t_{i_{j+1}}')
    + (1 - \drift)(1 + \mu)(t_{i_{j+1}}' - \phi_j)
    \\
    &\quad
    - \left( \lc{v_{i_{j+1}}}(t_{i_{j+1}}') - \lc{v_{i_{j+1}}}(\phi_j) \right)
    + \frac{\kappa_{(v_{i_j}, \ldots, v_{i_{j+1}})}}{2}
    .
  \end{align*}
  Now let us go back in time to $\theta_{j+1}$. Note that by definition of
  $\theta_{j+1}$ we have $\phi_j - \theta_{j+1} \leq t_{i_{j+1}} - t_{i_{j+1}}'
  \leq \tau_{\set{v_{i_{j+1} - 1}, v_{i_{j+1}}}}$.
  Therefore,
  \begin{align*}
    \xi_{q'}^s(\theta_{j+1})
    &=
    \xi_{q'}^s(\phi_j)
    - \left( \lc{v_{i_{j+1}}}(\phi_j) - \lc{v_{i_{j+1}}}(\theta_{j+1}) \right)
    + \left( \lc{x}(\phi_j) - \lc{x}(\theta_{j+1}) \right)
    \\
    &\geq
    2 \drift \Lambda_s
    + ( (1 - \drift)\mu - 2\drift)(t_0 - t_{i_{j+1}}')
    + (1 - \drift)(1 + \mu)(t_{i_{j+1}}' - \phi_j)
    \\
    &\quad
    - \left( \lc{v_{i_{j+1}}}(t_{i_{j+1}}') - \lc{v_{i_{j+1}}}(\theta_{j+1}) \right)
    + \frac{\kappa_{(v_{i_j}, \ldots, v_{i_{j+1}})}}{2}
    % + \left( \lc{x}(t) - \lc{x}(\theta_{j+1}) \right)
    + (1 - \drift)(\phi_j - \theta_{j+1})
    \\
    &=
    2 \drift \Lambda_s
    + ( (1 - \drift)\mu - 2\drift)(t_0 - t_{i_{j+1}}')
    + (1 - \drift)(1 + \mu)(t_{i_{j+1}}' - \theta_{j+1})
    \\
    &\quad
    - (1 - \drift)(1 + \mu)(\phi_j - \theta_{j+1})
    + (1 - \drift)(\phi_j - \theta_{j+1})
    \\
    &\quad
    - \left( \lc{v_{i_{j+1}}}(t_{i_{j+1}}') - \lc{v_{i_{j+1}}}(\theta_{j+1}) \right)
    + \frac{\kappa_{(v_{i_j}, \ldots, v_{i_{j+1}})}}{2}
    \\
    &>
    2 \drift \Lambda_s
    + ( (1 - \drift)\mu - 2\drift)(t_0 - t_{i_{j+1}}')
    + (1 - \drift)(1 + \mu)(t_{i_{j+1}}' - \theta_{j+1})
    \\
    &\quad
    - \mu \tau_{\set{v_{i_{j+1} - 1}, v_{i_{j+1}}}}
    - \left( \lc{v_{i_{j+1}}}(t_{i_{j+1}}') - \lc{v_{i_{j+1}}}(\theta_{j+1}) \right)
    + \frac{\kappa_{(v_{i_j}, \ldots, v_{i_{j+1}})}}{2}
    \\
    &\stackrel{\mathclap{\eqref{eq:def_kappa}}}{>}
    2 \drift \Lambda_s
    + ( (1 - \drift)\mu - 2\drift)(t_0 - t_{i_{j+1}}')
    + (1 - \drift)(1 + \mu)(t_{i_{j+1}}' - \theta_{j+1})
    \\
    &\quad
    - \left( \lc{v_{i_{j+1}}}(t_{i_{j+1}}') - \lc{v_{i_{j+1}}}(\theta_{j+1}) \right)
    .
  \end{align*}
  Finally, to show that $q' \in P_{v_{i_{j+1}}}^s(\theta_{j+1})$, recall that $q'
  = (v_{i_{j+1}}, \ldots, v_{i_j}) \circ q$. Because $\theta_{j+1} \leq
  t_{i_{j+1}}'$ (by definition of $i_{j+1}$), Property~\eqref{item:path_ih} of the
  chain shows that $p_{i_{j+1}} \in P_{v_{i_{j+1}}}^s(\theta_{j+1})$, and in
  particular $(v_{i_{j+1}}, \ldots, v_{i_j}) \in P_{v_{i_{j+1}}}^s(\theta_{j+1})$,
  as this is a sub-path of $p_{i_{j+1}}$. Thus, to show that $q' \in
  P_{v_{i_{j+1}}}^s(\theta_{j+1})$, it remains to show that $q \in
  P_{v_{i_j}}^s(\theta_{j+1})$.

  By definition of $q$ we have $q \in P_{v_{i_j}}^s(\phi_j)$, and from~\eqref{eq:pull}
  it follows that
  \begin{align*}
    \xi_q^s(\phi_j)
    &= \Xi_{v_{i_j}}^s(\phi_j)
    \geq
    2 \drift \Lambda_s
    + ( (1 - \drift)\mu - 2\drift)(t_0 - t_{i_j+1})
    + (1 - \drift + 2\mu)(t_{i_j+1} - \phi_j)
    \\
    &\quad
    - (1 + \drift)(1 + \mu)(t_{i_j + 1} - \phi_j)
    \\
    &=
    2 \drift \Lambda_s
    + 2( (1 - \drift) \mu - 2\drift)(t_0 - \phi_j)
    \stackrel{\mathclap{\eqref{eq:sigma}}}{>} 0.
  \end{align*}
  Also, from Property~\eqref{item:stab}, $v_{i_j}$ satisfies the
  $(C,s)$-stabilization condition at time $\phi_j \leq t_{i_j}'$. Therefore,
  Lemma~\ref{lemma:paths} shows that for all $t' \in [t^-, \phi_j]$ we have $q \in
  P_{v_{i_j}}^s(t')$. In particular, then, $q \in P_{v_{i_j}}^s(\theta_{j+1})$.
  This shows that $\Xi_{v_{i_{j+1}}}^s(\theta_{j+1}) \geq
  \xi_{q'}^s(\theta_{j+1})$ and completes the induction.

  The induction ends at a time $\phi_h \in [t^-, t_0 - \Lambda_s]$ and a node
  $v_{i_h}$ satisfying
  \begin{align*}
    \Xi_{v_{i_h}}^s(\phi_h)
    &\geq
    2 \drift \Lambda_s
    + ( (1 - \drift)\mu - 2\drift)(t_0 - t_{i_h+1})
    + (1 - \drift)(1 + \mu)(t_{i_h+1} - \phi_h)
    \nonumber
    \\
    &\quad
    - \left( \lc{v_{i_h}}(t_{i_h+1}) - \lc{v_{i_h}}(\phi_h) \right)
    + (1 + \drift)\mu(t_{i_h + 1} - \phi_h)
    \\
    &\geq
    2 \drift \Lambda_s
    + ( (1 - \drift)\mu - 2\drift)(t_0 - t_{i_h+1})
    + (1 - \drift)(1 + \mu)(t_{i_h+1} - \phi_h)
    \nonumber
    \\
    &\quad
    - (1 + \drift)(1 + \mu)(t_{i_h+1}) - \phi_h)
    + (1 + \drift)\mu(t_{i_h + 1} - \phi_h)
    \\
    &=
    2 \drift \Lambda_s
    + ( (1 - \drift)\mu - 2\drift)(t_0 - \phi_h)
    \\
    &\geq
    2 \drift \Lambda_s
    + ( (1 - \drift)\mu - 2\drift)\Lambda_s
    \\
    &=
    (1 - \drift)\mu \cdot \frac{C_{s-1}}{2(1 - \drift)\mu}
    = \frac{C_{s-1}}{2}
    ,
  \end{align*}
  where in the second to last step we again used that $(1-\drift)\mu-2\drift>0$ due
  to Inequality~\eqref{eq:mu}. Let $p = (v_{i_h}, \ldots, y) \in
  P_{v_{i_h}}^s(\phi_{h})$ be a path such that $\xi_p^s(\phi_h) =
  \Xi_{v_{i_h}}^s(\phi_h) \geq C_{s-1}/2$. Thus,
  \begin{equation}
    \Psi_y^{s-1}(\phi_h) \geq \psi_{\bar{p}}^{s-1}(\phi_h) = \xi_p^s(\phi_h) +
    \frac{\kappa_p}{2} > \frac{C_{s-1}}{2}.
    \label{eq:Psi_contra}
  \end{equation}
  From Property~\eqref{item:stab}, node $v_{i_h}$ satisfies the
  $(C,s)$-stabilization condition at time $\phi_{h}$. Hence, we can apply
  Lemma~\ref{lemma:shortpaths} to $p$, yielding that $y$ is $(C,s-1)$-legal at
  time $\phi_h$, contradicting~\eqref{eq:Psi_contra}.
\end{proof}

\section{Dynamic and Local Global Skew Estimates}\label{sec:dynamic_est}

In this section, we extend the analysis to handling and adapting to dynamic global skew estimates. From here on, we therefore assume that when inserting an edge $\set{u,v}$, the insertion duration $\instime_{\set{u,v}}$ is computed according to \eqref{eq:instime_dynamic} in Algorithm \ref{algo:insertiontimes}. In the next lemma, we show that with this assumption the logical times of edge insertions
on different levels are well separated, even if the insertions are for
different edges and if different global skew estimates are used for
inserting the different edges. 

\begin{lemma}\label{lemma:instime_separation}
  Let $e$ and $e'$ be two edges that are inserted with global skew
  estimates $\Ge_e$ and $\Ge_{e'}$, respectively. Further, let $s\geq 1$ and
  $s'\geq 1$ be two levels and consider the logical insertion times
  $T_s^e$ and $T_{s'}^{e'}$. If $s\neq s'$, it holds that
  \begin{equation}\label{eq:instime_separation}
  |T_s^e - T_{s'}^{e'}| \geq
  \frac{\min\set{\instime_e,\instime_{e'}}}{
  2^7\cdot 4^{\min\set{s,s'}-2}}.
  \end{equation}
  If $s=s'$, either \eqref{eq:instime_separation} holds or $T_s^e=T_{s'}^{e'}$.
\end{lemma}
\begin{proof}
  For convenience, we define $\ell_e:=\left\lceil\log_2 (\Ge_e/\mu +
    \delay_e + \tup_e)\right\rceil$ and $\ell_{e'} :=
  \left\lceil\log_2 (\Ge_{e'}/\mu + \delay_{e'} + \tup_{e'})\right\rceil$.
  The lengths $\instime_e$ and $\instime_{e'}$ of the time insertion
  intervals are then $\instime_e=\insconst\cdot 2^{3+\ell_e}$ and
  $\instime_{e'}=\insconst\cdot 2^{3+\ell_{e'}}$, respectively.
  Without loss of generality, assume that $\ell_e\leq \ell_{e'}$ and
  thus also $\instime_e \leq \instime_{e'}$. Consider the insertion
  times $T_s^e$ and $T_{s'}^{e'}$ of the edges $e$ and $e'$ on levels
  $s$ and $s'$, respectively. We have $T_s^e = T_0^e +
  \instime_e\big(1-\frac{1}{2^{s+1}-1}\big)$ and similarly
  $T_{s'}^{e'} = T_0^{e'} +
  \instime_{e'}\big(1-\frac{1}{2^{s'+1}-1}\big)$. Recall that $T_0^e$
  is chosen to be an integer multiple of $\instime_e$ and $T_0^{e'}$
  is chosen to be an integer multiple of $\instime_e'$. Let
  $\Delta_\ell:=\ell_e-\ell_{e'}\geq 0$ such that $\instime_{e} =
  \instime_{e'}\cdot 2^{\Delta_\ell}$. Defining $S_0^e :=
  T_0^e/\instime_{e'}$ and $S_0^{e'}:=T_0^{e'}/\instime_{e'}$, we can
  then write $T_s^e$ and $T_{s'}^{e'}$ as
  \begin{equation}\label{eq:rewriteinstimes}
  T_s^e = \instime_{e'}\cdot\left(\underbrace{S_0^e +
    2^{\Delta_\ell}\left(1-\frac{1}{2^{s+1}-1}\right)}_{=:x}\right)
  \quad\text{and}\quad
  T_{s'}^{e'} = \instime_{e'}\cdot\left(\underbrace{S_0^{e'} +
    1-\frac{1}{2^{s'+1}-1}}_{=:y}\right),
  \end{equation}
  where $S_0^e$, $S_0^{e'}$, and $\Delta_\ell$ are all non-negative
  integers. Consider $x$ and $y$ as defined in
  \eqref{eq:rewriteinstimes}. We then have $|T_s^e-T_{s'}^{e'}| =
  \instime_{e'}\cdot|x-y|$ and it therefore suffices to bound
  $|x-y|$. For any integer $a\geq 1$, it holds that
  $\frac{1}{2^{a}-1}=\sum_{i=1}^\infty \frac{1}{2^{a\cdot i}}$. As a
  consequence, we can write $x$ and $y$ as
  \begin{equation}\label{eq:xandy}
    x = \lceil x\rceil -
    \sum_{i=1}^\infty\frac{1}{2^{i(s+1)-k_\ell}}
    = \lceil x\rceil - \frac{2^{k_\ell}}{2^{s+1}-1} 
    \quad\text{and}\quad
    y = \lceil y \rceil - 
    \frac{1}{2^{(s'+1)}-1},
  \end{equation}
  where $k_\ell = \Delta_\ell \!\!\mod (s+1)$ and thus
  $k_\ell\in\set{0,\dots, s}$. From \eqref{eq:xandy}, we have
  \begin{equation}\label{eq:boundonx}
    \lceil x\rceil - \frac{1}{2^{s-k_\ell+1}} > x =  \lceil x \rceil -
    \frac{1}{2^{s-k_\ell+1}}\cdot\frac{2^{s+1}}{2^{s+1}-1}
    \stackrel{(s\geq 1)}{\geq}
    \lceil x\rceil-\frac{1}{2^{s-k_\ell+1}}\cdot\frac{4}{3}
    \stackrel{(s- k_\ell\geq 0)}{\geq}
    \lceil x\rceil-\frac{2}{3}
  \end{equation}
  and
  \begin{equation}\label{eq:boundony}
    \lceil y\rceil - \frac{1}{2^{s'+1}} > y =  \lceil y \rceil -
    \frac{1}{2^{s'+1}}\cdot\frac{2^{s'+1}}{2^{s'+1}-1}
    \stackrel{(s'\geq 1)}{\geq}
    \lceil y\rceil-\frac{1}{2^{s'+1}}\cdot\frac{4}{3}.
    \stackrel{(s'\geq 1)}{\geq}
    \lceil y\rceil-\frac{1}{3}.
  \end{equation}
  Let us first consider the case where $\lceil x\rceil\neq\lceil
  y\rceil$. From the last inequalities of \eqref{eq:boundonx} and
  \eqref{eq:boundony}, we then get that
  \begin{equation}
    \label{eq:differentintegralpart}
    |x-y| > \frac{1}{3}.
  \end{equation}
  Let us therefore come to the case where $\lceil x\rceil = \lceil
  y\rceil$. If $s-k_\ell < s'$, i.e., $s-k_\ell \le s'-1$, \eqref{eq:boundonx} and
  \eqref{eq:boundony} imply that
  \begin{equation}
    \label{eq:ylarger}
    y - x > 
    \frac{1}{2^{s-k_\ell+1}} - \frac{1}{2^{s'+1}}\cdot\frac{4}{3}
    \stackrel{(s'-1 \ge s-k_\ell)}{\geq}
    \frac{1}{3\cdot 2^{s-k_\ell+1}}
    \geq
    \frac{1}{6\cdot 2^{\min\set{s,s'}}}.
  \end{equation}
  Similarly, if $\lceil x\rceil=\lceil y\rceil$ and $s-k_\ell>s'$, i.e., $s-k_\ell \ge s'+1$, we
  obtain
  \begin{equation}
    \label{eq:xlarger}
    x - y >
    \frac{1}{2^{s'+1}} - \frac{1}{2^{s-k_\ell+1}}\cdot\frac{4}{3}
    \stackrel{(s-k_\ell \ge s'+1)}{\geq}
    \frac{1}{3\cdot 2^{s'+1}}
    =
    \frac{1}{6\cdot 2^{\min\set{s,s'}}}.
  \end{equation}
  Finally, for $\lceil x\rceil=\lceil y\rceil$ and $s-k_\ell=s'$, we
  either have $k_\ell=0$ and $s=s'$ in which case
  $\eqref{eq:boundonx}$ and $\eqref{eq:boundony}$ imply that $x=y$ and
  therefore also $T_s^e=T_{s'}^{e'}$. Otherwise, assume that
  $k_\ell>0$ and $s = s' + k_\ell$. We then get
  \begin{equation}
    \label{eq:samefirstterm}
    x-y = \frac{1}{2^{s'+1}}\cdot\left(
      \frac{2^{s'+1}}{2^{s'+1}-1}-\frac{2^{s+1}}{2^{s+1}-1}
    \right)
    >
    \frac{1}{2^{s'+1}}\cdot
    \frac{2^{s+1}-2^{s'+1}}{2^{s+s'+2}}
    \stackrel{(s>s')}{\geq}
    \frac{1}{8\cdot 2^{2s'}} = 
    \frac{1}{8\cdot 4^{\min\set{s,s'}}}.
  \end{equation}
  Combining \eqref{eq:differentintegralpart}, \eqref{eq:ylarger},
  \eqref{eq:xlarger}, and \eqref{eq:samefirstterm}, we either have
  $s=s'$ and $x=y$ or we get that
  \[
  |x-y| \geq \frac{1}{8\cdot 4^{\min\set{s,s'}}}
  = \frac{1}{2^7\cdot 4^{\min\set{s,s'}-2}}.
  \]
  Consequently, we either have $s=s'$ and $T_s^e=T_{s'}^{e'}$, or we obtain
  \[
  |T_s^e - T_{s'}^{e'}| \geq
  \frac{\instime_{e'}}{2^7\cdot 4^{\min\set{s,s'}-2}}
  =  \frac{\min\{\instime_e,\instime_{e'}\}}{2^7\cdot 4^{\min\set{s,s'}-2}},
  \]
  and thus the claim of the lemma follows.
\end{proof}

``The'' gradient property here is actually a time-dependent
notion, since the global skew varies over time. The algorithm will take some
time to adapt to a smaller global skew, and this process is complicated by
potential simultaneous edge insertions. To capture the former, we define for each
time $t$ a certain time $P(t)$ that lies sufficiently far in the past for the
algorithm to have time to accomodate the corresponding global skew (at time
$P(t)$). For each $t\in \R^+_0$ such that this value is defined, we set
\begin{equation}\label{eq:Pt}
  P(t):=\max\{t'\in [0,t]\,|\,\insconst\cdot \G(t')=\mu (t-t')\}.
\end{equation}

Note that $P(t)$ exists by continuity of $f(t'):=\insconst\G(t')-\mu(t-t')$ for all
$t\geq \insconst\G(0)/\mu$, since this implies $f(0)< 0$ and trivially we
have that $f(t)\geq 0$. In the following, let $t_{\min}$ denote the minimal time
so that $P(t_{\min})$ exists, i.e., $P:[t_{\min},\infty)\to \R^+_0$ such that
$P(t)\leq t$. Our goal will be to prove a non-trivial gradient property for
times $t\geq t_{\min}$; for smaller times, the algorithm had insufficient time
to converge to small skews (where the meaning of ``small'' depends on
$\G(P(t))$, as clarified below).

Before proceeding with the definition of the gradient sequences we use, let us establish
some basic properties of $P(t)$. In the following, we use the shorthand
\begin{equation*}
  \label{eq:F}
  \Gu:=2\G(P(t)).
\end{equation*}
\begin{lemma}\label{lemma:P(t)}
For each $t\geq t_{\min}$ and all $t'\in [P(t),t]$, it holds that
\begin{compactitem}
\item [(i)] $\G(t')\leq \Gu$ and
\item [(ii)] $\G(t')\geq (t-t')\mu/\insconst$.
\end{compactitem}
\end{lemma}
\begin{proof}
Fix any $t\geq t_{\min}$. For $t'\in [P(t),t]$, by Theorem~\ref{thm:global} it holds that
\begin{equation*}
\G(t')\leq \G(P(t))+2\rho(t'-P(t))\leq \G(P(t))+2\rho (t-P(t))=\G(P(t))+\frac{2\rho
\insconst\G(P(t))}{\mu}\sr{eq:insconst}{\leq}\Gu,
\end{equation*}
yielding Statement (i). Regarding Statement (ii), assume for the sake of contradiction that
\begin{equation*}
f(t'):=\insconst\G(t')-\mu(t-t')< 0.
\end{equation*}
Clearly, $f$ is continuous and $f(t)\geq 0$. Hence, there must exist some
$t''\in (t',t]$ so that $f(t'')=0$. This contradicts the maximality of
$P(t)\leq t'<t''\leq t$ among times smaller or equal to $t$ with the property
that $f(t)=0$. We conclude that Statement (ii) must be true as well.
\end{proof}
The first property enables us to use $\Gu$ as a global skew upper bound
for defining gradient sequences pertinent for the entire interval $[P(t),t]$.
As shown in the following lemma, the second property guarantees that, during
$[(2P(t)+t)/3,t]$, no edge insertions happen on any level $s>0$ that are based on
a global skew estimate substantially smaller than $\G(P(t))$.
\begin{lemma}\label{lemma:no_bad_insert}
For each $t\geq t_{\min}$, all $t'\in [(2P(t)+t)/3,t]$, and any nodes $u,v$, if
$u$ adds $v$ to $\N[u]^s$ for any $s>0$ at time $t'$, then it holds that the
corresponding call to \textbf{insertedge} (see
Listing~\ref{algo:insertiontimes}) computes $\instime_{\{u,v\}}(\Ge) \geq
(1-\rho)\insconst\G(P(t))/(10\mu)$, where $\Ge$ is the global skew estimate passed to
\textbf{insertedge}.
\end{lemma}
\begin{proof}
Fix $t$, $t'$, $u$, and $v$, and suppose $u$ adds $v$ to $\N[u]^s$ at time $t'$.
Denote by $t_0$ the (most recent) time when $u$ added $v$ to $\N[u]^0$ and
suppose that $\Ge$ was used in the call to \textbf{insertedge} at time
$t_0$.

Assume that $w\in\{u,v\}$ is the leader of edge $\{u,v\}$. We distinguish two
cases, the first being that $t_0< P(t)$. For this case, note that the
corresponding logical time for which the insertion is complete on all levels
satisfies
\begin{eqnarray*}
T_{\infty}^{\{u,v\}}&\leq &
L_w(t_0)+(1+\rho)(1+\mu)(\delay_{\{u,v\}}+2\tau_{\{u,v\}})
+2\instime_{\set{u,v}}\\
&\stackrel{\eqref{eq:instime_dynamic}}{=} &
L_w(t_0)+(1+\rho)(1+\mu)(\delay_{\{u,v\}}+2\tau_{\{u,v\}})
+2\insconst\cdot 2^{3+\lceil\log(
\Ge/\mu+\delay_{\{u,v\}}+\tau_{\{u,v\}})\rceil}\\
&\stackrel{\eqref{eq:instime_dynamic}}{<} & L_w(t_0)+3\instime_{\set{u,v}}.
\end{eqnarray*}
We conclude that the time $t_{\infty}$ so that
$L_w(t_{\infty})=T_{\infty}^{\{u,v\}}$ is bounded from above by
\begin{equation*}
t_{\infty}\leq t_0+\frac{3\instime_{\{u,v\}}}{1-\rho}.
\end{equation*}
Moreover, if $t'>t_{\infty}$ (which may happen if $w=v$), Statement~(i) of
Lemma~\ref{lemma:P(t)} yields that
\begin{eqnarray*}
T_{\infty}^{\{u,v\}}&>& L_u(t')\\
&\geq & L_w(t')-\G(t')\\
&\geq & L_w(t_{\infty})+(1-\rho)(t'-t_{\infty})-\Gu\\
&=& T_{\infty}^{\{u,v\}}+(1-\rho)(t'-t_{\infty})-\Gu.
\end{eqnarray*}
Therefore,
\begin{equation*}
t'-t_0= t'-t_{\infty}+t_{\infty}-t_0\leq
\frac{3\instime_{\{u,v\}}+\Gu}{1-\rho}.
\end{equation*}
As $t'\geq (2P(t)+t)/3$ is equivalent to $t-P(t)\leq 3(t'-P(t))$, this leads to
\begin{equation*}
\G(P(t))=\frac{\mu (t-P(t))}{\insconst}\leq \frac{3\mu (t'-P(t))}{\insconst}
< \frac{3\mu (t'-t_0)}{\insconst}\leq \frac{9\mu \instime_{\{u,v\}}+3\mu
\Gu}{(1-\rho)\insconst}\stackrel{(\ref{eq:mu},\ref{eq:insconst})}{<}\frac{9\mu
\instime_{\{u,v\}}}{(1-\rho)\insconst}+\frac{\G(P(t))}{10}.
\end{equation*}
Rearranging this inequality, we obtain that
\begin{equation*}
\frac{(1-\rho)\insconst\G(P(t))}{10\mu}\leq \instime_{\{u,v\}},
\end{equation*}
i.e., the claim of the lemma holds.

The second case is that $t_0\geq P(t)$. Since trivially $t_0<t'\leq t$,
Statement~(ii) of Lemma~\ref{lemma:P(t)} implies that
\begin{equation*}
\instime_{\{u,v\}}> \frac{\insconst \cdot 2^3\cdot\Ge}{\mu} \geq 
\frac{\insconst \cdot 2^3\cdot\G(t_0)}{\mu}\stackrel{(\ref{eq:mu},\ref{eq:insconst})}{>}
3\G(t_0)+\frac{4\insconst \G(t_0)}{\mu}
\stackrel{\eqref{eq:Pt}}{\geq} 3\G(t_0)+4(t-t_0)
\end{equation*}
and hence
\begin{equation*}
T_1^{\{u,v\}}>  L_w(t_0)+\frac{2\instime_{\{u,v\}}}{3}
> L_u(t_0)+\frac{8(t-t_0)}{3}
\sr{eq:mu}{>} L_u(t_0)+(1+\rho)(1+\mu)(t-t_0).
\end{equation*}
We conclude that
\begin{equation*}
t'\geq t_0+\frac{T_1^{\{u,v\}}-L_u(t_0)}{(1+\rho)(1+\mu)} > t,
\end{equation*}
contradicting the prerequisite that $t'\in [(2P(t)+t)/3,t]$. Therefore, the
first case must always apply and the proof is complete.
\end{proof}

In summary, we have established that for any time $t\geq t_{\min}$ that (i)
$2\cdot\G(P(t))$ is a valid upper bound on the global skew throughout $[P(t),t]$ and
(ii) that no edge $e$ is inserted on any level $s>0$ during $[(2P(t)+t)/3,t]$
with a value of $\instime_e<(1-\rho)\insconst\G(P(t))/(10\mu)$. The latter property
ensures that the sets $\Tset_s$, $s>0$, are sufficiently ``sparse'' (i.e.,
insertion times are sufficiently well-separated) for the algorithm to stabilize
to small skews during $[(P(t)+2t)/3,t]$, based on the guaranteed upper bound of
$2\G(P(t))$ on the global skew. Note that we restrict the time interval
for which we show convergence to $[(P(t)+2t)/3,t]$, so there is a buffer
against interference from edge insertions with small values of $\instime_e$ that
may have occured before time $(2P(t)+t)/3$. However, we know little about
$\Tset_s\cap (t,\infty)$, since the global skew might decrease quickly after
time $t$; this is a technical issue that we will deal with by, essentially,
ignoring insertions after time $t$.

Let us now formalize the above intuition by defining suitable time periods and
(time- and node-dependent) gradient sequences so that the preconditions of
Theorem~\ref{thm:catch_up} will be satisfied for any level $s>1$, time, and
node for which the respective gradient sequence $C$ satisfies that
$C_s<C_{s-1}$.

\begin{definition}[Instability Periods]\label{def:periods}
  For a level $s>1$, a time $t\geq t_{\min}$, and a node $u$, we define the set
  of its \emph{$s$-unstable times (with respect to time $t$)} as
  \[
  U_s(u,t):=\set{t'\in [(P(t)+2t)/3,t]\,\Big|\,\exists \, T_s\in \Tset_s:
  |\lc{u}(t')-T_s|\leq A_s(\Gu)},
  \]
  where
  \begin{equation}
  A_s(\Gu) := \left(\frac{(7+2\rho)(1+\mu)}{2\mu(1-\rho)}+2s\right)
  \frac{2\Gu}{\sigma^{s-2}}.\label{eq:A}
  \end{equation}
For a level $s'>2$, a time $t\geq t_{\min}$, and a node $u$, the set of
$s'$-recovery times is
\begin{equation*}
R_{s'}(u,t):=\bigcup_{s<s'}\{t'\in [(P(t)+2t)/3,t]\,\big|\,
\exists \,T_s\in \Tset_s: B_{s,s'-1}(\Gu)<|L_u(t')-T_s|-A_s(\Gu)	\leq
B_{s,s'}(\Gu)\},
\end{equation*}
where
\begin{equation}
B_{s_1,s_2}(\Gu) := \sum_{s=s_1}^{s_2-1} \beta_s(\Gu)\quad
\text{and} \quad \beta_s(\Gu) := \left(\frac{(7+2\rho)(1+\mu)}{2\mu(1-\rho)}+s\right)
\frac{2\Gu}{\sigma^{s-2}}.\label{eq:beta}
\end{equation}
\end{definition}

Here, the time intervals $U_s(u,t)$ provide a ``buffer'' around the (logical!)
times from $T_s$, during which we will not make any non-trivial guarantees on
level $s$ at node $u$, i.e., the respective gradient sequence $C$ will satisfy
that $C_s=C_{s-1}$. The additional ``buffers'' provided by the sets $R_{s'}$
ensure that the gradient sequences at nodes with similar logical times do not
differ in more than a single level. This is
crucial for applying Theorem~\ref{thm:catch_up}, since it requires the
level-$s$ stabilization condition to hold for all nodes with logical times from
a certain range around the logical clock value of the node we examine.

Before defining suitable gradient sequences based on the above sets,
we must show that the sets for different levels are pairwise
disjoint. However, as mentioned earlier, we have no control over edge
insertions at times larger than $t$. We overcome this by first
considering a constrained set of executions for which there are no
insertions after time $t$ and only afterwards inferring skew bounds
for arbitrary executions.
\begin{definition}[Insertion-bounded Executions]
For $t\in \R^+_0$, an execution is called \emph{$t$-insertion-bounded} iff no
edges are inserted on any level at times greater than $t$.
\end{definition}
For such executions, we can show that the sets specified in
Definition~\ref{def:periods} are disjoint.
\begin{lemma}\label{lemma:instab}
Consider a $t$-insertion-bounded execution for some $t\geq t_{\min}$. Then
it holds for all nodes $u$ and $1<s'<s$ that $U_s(u,t)\cap
U_{s'}(u,t)=\emptyset$ and $U_s(u,t)\cap R_{s'}(u,t)=\emptyset$.
\end{lemma}
\begin{proof}
  For a given $s>1$ and $T_{s}\in \Tset_{s}$, abbreviate
  \begin{equation*}
    U_s(T_s):=\set{t'\in [(P(t)+2t)/3,t]\,\Big|\, |\lc{u}(t')-T_s|\leq A_s(\Gu)}
  \end{equation*}
  and, for $s'> s$,
  \begin{equation*}
    R_{s'}(T_s):= \{t'\in [(P(t)+2t)/3,t]\,\big|
    B_{s,s'-1}(\Gu)<|L_u(t')-T_s|-A_s(\Gu)\leq
    B_{s,s'}(\Gu)\}.
  \end{equation*}
  This entails that
  \begin{equation*}
    U_s(u,t)=\bigcup_{T_s\in \Tset_s} U_s(T_s)\quad \mbox{and}\quad
    R_{s'}(u,t)=\bigcup_{s<s'}\bigcup_{T_s\in \Tset_s} R_{s'}(T_s).
  \end{equation*}
  Moreover, we have
  \begin{equation}
    U_s(T_s)\cup \bigcup_{s'> s} R_{s'}(T_s)=
    \Big(T_s-A_s(\Gu)-B_{s,\infty}(\Gu),T_s+A_s(\Gu)+B_{s,\infty}(\Gu)\Big),
    \label{eq:set_union}
  \end{equation}
  where 
  \begin{equation*}
    B_{s,\infty}(\Gu):=\lim_{s'\to \infty}B_{s,s'}(\Gu)
    = \sum_{s'= s}^{\infty}\beta_{s'}(\Gu).
  \end{equation*}
  Evaluating this limit is straightforward, yielding
  \begin{eqnarray}
    A_s(\Gu)+B_{s,\infty}(\Gu)
    &=&\frac{2\Gu}{\sigma^{s-2}}
        \left(\frac{(7+2\rho)(1+\mu)}{2\mu(1-\rho)}\left(1+\frac{\sigma}{\sigma-1}\right)+
        s\left(2+\frac{\sigma}{\sigma-1}\right)+\frac{\sigma}{(\sigma-1)^2}\right)\nonumber\\
    &\stackrel{\sigma\geq 101}{\leq}& \frac{2\Gu}{4^{s-2}}
                                      \left(\left(\frac{1407}{200}+
                                      \frac{201\drift}{100}\right)\cdot\frac{1+\mu}{\mu(1-\rho)} +\frac{301s}{100\cdot
                                      12^{s-2}}+\frac{101}{100^2}\right)\nonumber\\
    &\stackrel{s\geq 2}{\leq} & \frac{2\Gu}{4^{s-2}}
                                \left(\left(\frac{1407}{200}+
                                \frac{201\drift}{100}\right)\cdot\frac{1+\mu}{\mu(1-\rho)} +\frac{301}{100}+\frac{101}{100^2}\right)\nonumber\\
    &\stackrel{(\ref{eq:rho_dynamic},\ref{eq:mu})}{<} &
                                                \frac{16\Gu}{4^{s-2}\mu(1-\rho)}.\label{eq:AB}
\end{eqnarray}
In particular, this expression is maximized for $s=2$, giving
\begin{eqnarray*}
  A_s(\Gu)+B_{s,\infty}(\Gu)+\Gu 
  & \stackrel{(\ref{eq:mu},\ref{eq:AB})}{<} & \frac{17\Gu}{\mu(1-\drift)}\\
  & = & \frac{34\G(P(t))}{\mu(1-\drift)}\\
  & \stackrel{\eqref{eq:Pt}}{=} &
                                  \frac{34\cdot(t-P(t))}{(1-\drift)\insconst}\\
  & \sr{eq:insconst}{\leq} &
                                                         \frac{(1-\rho)(t-P(t))}{3}.
\end{eqnarray*}

Now consider any $T_s\in \Tset_s$ such that $L_v((2P(t)+t)/3)> T_s$ for some $v\in V$. From Statement~(i) of
Lemma~\ref{lemma:P(t)} and the above inequality, we obtain that
\begin{eqnarray*}
L_u\left(\frac{P(t)+2t}{3}\right)-T_s &\geq &
L_u\left(\frac{2P(t)+t}{3}\right)
+\frac{(1-\rho)(t-P(t))}{3}-T_s\\
&\geq & L_v\left(\frac{2P(t)+t}{3}\right)-\G\left(\frac{2P(t)+t}{3}\right)
+\frac{(1-\rho)(t-P(t))}{3}-T_s\\
&\geq & \frac{(1-\rho)(t-P(t))}{3}-\Gu\\
&>& A_s(\Gu)+B_{s,\infty}(\Gu),
\end{eqnarray*}
and hence $U_s(T_s)=\emptyset$ and $R_{s'}(T_s)=\emptyset$ for all $s'>s$. Thus,
any insertion time $T_s$ lies sufficiently far in the past and is of no concern.
As the execution is $t$-insertion-bounded, we conclude that it suffices to
consider $s,s'\geq 2$, $T_s\in \Tset_s$, and $T_{s'}\in \Tset_{s'}$,
so that there exists a node $v$ inserting an edge $\{v,w\}$ on level $s$ at the
time $t_v\in [(2P(t)+t)/3,t]$ satisfying that $L_v(t_v)=T_s$ and a node $v'$
inserting an edge $\{v',w'\}$ on level $s'$ at the time $t_{v'}\in
[(2P(t)+t)/3,t]$ satisfying that $L_{v'}(t_{v'})=T_{s'}$.

Because we have $s\neq s'$, Lemma~\ref{lemma:instime_separation} states that
\begin{equation*}
|T_s-T_{s'}|\geq
\frac{\min\{\instime_{\{v,w\}}(\Ge),\instime_{\{v',w'\}}(\Ge')\}}
{2^7\cdot 4^{\min\set{s,s'}-2}},
\end{equation*}
where $\Ge$ and $\Ge'$ are the estimates used in the computation of the logical
insertion times $T_s$ and $T_{s'}$, respectively (by the leaders of the
inserted edges $\{u,v\}$ and $\{u',v'\}$). Applying
Lemma~\ref{lemma:no_bad_insert}, we see that
\begin{eqnarray*}
\frac{\min\{\instime_{\{v,w\}}(\Ge),\instime_{\{v',w'\}}(\Ge')\}}
{2^7\cdot 4^{\min\set{s,s'}-2}}
&\geq &\frac{(1-\rho)\insconst\G(P(t))}{10\mu \cdot 2^7\cdot 4^{\min\set{s,s'}-2}}\\
&\sr{eq:insconst}{\geq}& \frac{32\G(P(t))}{4^{\min\{s,s'\}-2}\mu(1-\rho)}\\
&=& \frac{16\Gu}{4^{\min\{s,s'\}-2}\mu(1-\rho)}\\
&\sr{eq:AB}{>}& A_s(\Gu)+B_{s,\infty}(\Gu)+A_{s'}(\Gu)+B_{s',\infty}(\Gu).
\end{eqnarray*}
Recalling \eqref{eq:set_union}, we conclude that
\begin{equation*}
\left(U_s(T_s)\cup \bigcup_{s''> s} R_{s''}(T_s)\right)\cap
\left(U_{s'}(T_{s'})\cup \bigcup_{s''> s'} R_{s''}(T_{s'})\right)=\emptyset.
\end{equation*}
Thus, all cases are covered and the proof is complete.
\end{proof}

With this lemma and the definitions preceding it at hand, we can now specify
gradient sequences suitable for deriving our skew bounds.

\begin{definition}[Global and Local Gradient Sequences]\label{def:levelSeq}
  Given times $t\geq t_{\min}$ and $t'\in [P(t),t]$, we define the \emph{global
  gradient sequence} $C^{(t,t')}$ as follows. Set $\Delta_1(t):=t-P(t)$,
  and
  \begin{equation}
  \Delta_s(t):=\sum_{s'=s}^{\infty}\frac{7\Gu}{2(1-\rho^2)\mu\sigma^{\max\{s'-3,0\}}}
  \label{eq:def_delta}
  \end{equation}
  for $s>1$. Denote by $s'\in \nat$ the uniqe level such that $t'\in
  [t-\Delta_{s'}(t),t-\Delta_{s'+1}(t))$. Then
  \begin{equation*}
  C^{(t,t')}_s:=\left\{\begin{matrix}
  \frac{2\Gu}{\sigma^{s-1}} & \mbox{if } s\leq s' \\[1ex]
  \frac{2\Gu}{\sigma^{s'-1}} & \mbox{else.} \end{matrix}\right.
  \end{equation*}
  For the above parameters and a node $u$, the \emph{local gradient sequence
  at $u$} is given by
  \begin{equation*}
  C^{(t,t',u)}_s:=\left\{\begin{matrix}
  C^{(t,t')}_{s-1}& \mbox{if $t'\in \bigcup_{s''\leq s}U_{s''}(u,t)\cup
  R_{s''}(u,t)$}\\[1ex]
  C^{(t,t')}_s & \mbox{else.}\\
  \end{matrix}\right.
  \end{equation*}
\end{definition}

Since
\begin{eqnarray}
\Delta_2(t)&=&
\sum_{s'=2}^{\infty}\frac{7\Gu}{2(1-\rho^2)\mu\sigma^{\max\{s'-3,0\}}}\nonumber\\
&=& \left(1+\frac{\sigma}{\sigma-1}\right)\frac{7\Gu}{2(1-\rho^2)\mu}\nonumber\\
&= & \left(1+\frac{\sigma}{\sigma-1}\right)\frac{7(t-P(t))}
{2(1-\rho^2)\insconst}\nonumber\\
&\sr{eq:insconst}{\leq} & \frac{t-P(t)}{3}\label{eq:Delta_2}\\ &<& t-P(t),\nonumber
\end{eqnarray}
the global sequences are well-defined (i.e., decreasing in $s$), implying the
same for the local sequences. We are now ready to prove our main result.

\begin{lemma}\label{lemma:proto_bound}
For any $t$-insertion-bounded execution with $t_{\min}\leq t$, it holds for all
times $t'\in [P(t),t]$ and nodes $u\in V$ that $u$ is $C^{(t,t',u)}$-legal at
time $t$.
\end{lemma}
\begin{proof}
  Suppose that this statement is false: Let $\bar{t}$ be the smallest
  time such that there exists a node $u$ and a time $t^+\geq \bar{t}$
  such that $\bar{t}\in[P(t^+),t^+]$ and $u$ is not
  $C^{(t^+,\bar{t},u)}$-legal at time $\bar{t}$ for some level
  $\bar{s}$.  W.l.o.g., assume that $\bar{s}$ is the smallest level
  for which $u$ is not $C^{(t^+,\bar{t},u)}$-legal at time
  $\bar{t}$. To simplify the notation, in the following we use
  $\Gu^+:=2\G(P(t^+))$. Note that by statement (i) of Lemma
  \ref{lemma:P(t)}, $\G(t)\leq \Gu^+$ for all $t\in[P(t^+),t^+]$.
  This implies that for all
  $t\in \big[P(t^+)+\stabint{2}^{(t^+,\bar{t},u)},t^+\big]$, condition
  (S0) of the stabilization condition (cf.\ Def.\
  \ref{def:stabcondition}) holds and therefore by Lemma
  \ref{lemma:1legality}, the system is $(C^{(t^+,\bar{t},u)},1)$-legal
  for all nodes $v\in V$ and all times $t'\in[P(t^+),t^+]$; hence, we
  in particular have $\bar{s}>1$.

We will now define a gradient sequence $\bar{C}$ such that
$\bar{C}_{\bar{s}}=C^{(t^+,\bar{t},u)}_{\bar{s}}$ and, for $s<\bar{s}$,
$\bar{C}_{s}\geq C^{(t^+,\bar{t},u)}_{s}$. Note that, by minimality of
$\bar{s}$, node $u$ is $(\bar{C},\bar{s})$-legal if and only if it is
$(C^{(t^+,\bar{t},u)},\bar{s})$-legal. It is thus sufficient to derive a
contradiction to the assumption that node $u$ is not $(\bar{C},\bar{s})$-legal
at time $\bar{t}$. The sequence is defined as follows:
\[ \bar{C}_{s} :=
\begin{cases}
\frac{2\Gu^+}{\sigma^{s-1}} & \text{ if }
C^{(t^+,\bar{t},u)}_{\bar{s}}=\frac{2\Gu^+}{\sigma^{\bar{s}-1}}\\
\frac{2\Gu^+}{\sigma^{\max\{s-2,0\}}} & \text{ if }
C^{(t^+,\bar{t},u)}_{\bar{s}}=\frac{2\Gu^+}{\sigma^{\bar{s}-2}}.
\end{cases}
\]
Note that due to the minimality of $\bar{s}$, Statement~(iii) of
Lemma~\ref{lemma:simple_stuff} shows that $C^{(t^+,\bar{t},u)}_{\bar{s}}<
C^{(t^+,\bar{t},u)}_{\bar{s}-1}$ and therefore either
$C^{(t^+,\bar{t},u)}_{\bar{s}}=\frac{2\Gu^+}{\sigma^{\bar{s}-1}}$ or
$C^{(t^+,\bar{t},u)}_{\bar{s}}=\frac{2\Gu^+}{\sigma^{\bar{s}-2}}$. Hence,
$\bar{C}$ is well-defined for all cases.

We will use Theorem~\ref{thm:catch_up} to prove $(\bar{C},\bar{s})$-legality
of $u$ at time $\bar{t}$, so our goal is to show that node $u$ satisfies the
preconditions to apply the lemma. We set the time
\begin{equation}\label{eq:def_undert1}
\underline{t} \coloneq \bar{t} -
\frac{5}{2}\cdot\frac{\bar{C}_{\bar{s}-1}}{\mu(1-\drift^2)}.
\end{equation}
Further, for a time $t\in[\underline{t},\bar{t}]$, we define
\begin{equation}\label{eq:defVut}
V_u(t):=\set{v\in V: \exists \text{ path } p=(u,\dots,v)\in
P_u^{\bar{s}}(t) \text{ with } \kappa_{p}\leq \bar{C}_{\bar{s}-1}}.
\end{equation}
In order to apply Theorem \ref{thm:catch_up}, we show that for all
times $t\in\left[\underline{t},\bar{t}\right]$, each node $v\in
V_u(t)$ satisfies the $(\bar{C},\bar{s})$-stabilization condition at
time $t$.  Setting $t^-=\underline{t}$ and $t^+=\bar{t}$, we can
then apply the lemma. Since
\[
\bar{t}-\underline{t}=\frac{5}{2}\cdot\frac{\bar{C}_{\bar{s}-1}}{\mu(1-\drift^2)}
> \frac{\bar{C}_{\bar{s}-1}}{2(1-\drift)\mu} + 2\barstabint{\bar{s}},
\]
the lemma implies that at time $\bar{t}$,
\[
\Psi_u^{\bar{s}}(\bar{t}) < \frac{\bar{C}_{\bar{s}-1}}{2\sigma}
\leq \frac{\bar{C}_{\bar{s}}}{2}.
\]
By Definition \ref{def:legality}, this is a contradiction to the
assumption that $u$ is not $(\bar{C},\bar{s})$-legal at time
$\bar{t}$.

We now show that the stabilization condition applies. Since
$C^{(t^+,\bar{t},u)}_{\bar{s}}<C^{(t^+,\bar{t},u)}_{\bar{s}-1}$ (or
Statement~(iii) of Lemma~\ref{lemma:simple_stuff} and the minimality of
$\bar{s}$ yield a contradiction), we have that $\bar{t}\notin
\bigcup_{s=0}^{\bar{s}-1}[t^+-\Delta_s(t),t^+-\Delta_{s+1}(t)]$. Therefore,
$\bar{t} \geq t^+-\Delta_{\bar{s}}$. In particular,
\begin{eqnarray*}
\underline{t}&\geq & \bar{t}-\frac{5}{2}\cdot \frac{\bar{C}_1}{\mu(1-\rho^2)}\\
&> & t^+-\Delta_2(t^+)-\frac{5\Gu^+}{2\mu(1-\rho^2)}\\
&\sr{eq:Delta_2}{\geq}&
P(t^+)+\frac{2(t^+-P(t^+))}{3}-\frac{5\Gu^+}{2\mu(1-\rho^2)}\\
&\sr{eq:Pt}{=}&
P(t^+)+\frac{2\insconst\Gu^+}{3\mu}-\frac{5\Gu^+}{2\mu(1-\rho^2)}\\
&\sr{eq:insconst}{>} &
P(t^+)+\frac{14\Gu^+}{2\mu(1-\rho^2)}-\frac{5\Gu^+}{2\mu(1-\rho^2)}\\
&\sr{eq:stabint}{>} & P(t^+)+\barstabint{2}.
\end{eqnarray*}
In other words, $[\underline{t},\bar{t}]\subseteq [P(t^+)+\barstabint{2},t^+]$.
As shown earlier, this entails that (S0) is satisfied at all nodes and times
$t\in [\underline{t},\bar{t}]$.

Concerning (S1), from the previous observation that $\bar{t}\geq
t^+-\Delta_{\bar{s}}$, we have for $1<s<\bar{s}$ that
\begin{eqnarray*}
\underline{t}&=& \bar{t}-\frac{5}{2}\cdot
\frac{\bar{C}_{\bar{s}-1}}{\mu(1-\rho^2)}\\
&\geq & t^+-\Delta_s(t^+)+\left(\Delta_s-\Delta_{s+1}-\frac{5}{2}\cdot
\frac{\bar{C}_{\bar{s}-1}}{\mu(1-\rho^2)}\right)\\
&\sr{eq:def_delta}{\geq}&
t^+-\Delta_s(t^+)+\left(\frac{7\Gu^+}{2\mu(1-\rho^2)\sigma^{\max\{s-3,0\}}}-
\frac{5\Gu^+}{2\mu(1-\rho^2)\sigma^{\max\{\bar{s}-3,0\}}}\right)\\
&>& t^+-\Delta_s(t^+)+\frac{\Gu^+}{(1+\rho)\mu\sigma^{\max\{s-3,0\}}}\\
&\geq & t^+-\Delta_s(t^+)+\barstabint{s}.
\end{eqnarray*}
Thus, for $t\in [\underline{t}-\barstabint{s},\bar{t}]$ and $1<s< \bar{s}$,
$C_s^{(t^+,t)}=2\Gu^+/\sigma^{s-1}$. We distinguish two cases according to
which gradient sequence we use for $\bar{C}$.

\paragraph{Case \boldmath$\bar{C}_s=2\hat{\G}/\sigma^{\max\{s-2,0\}}$ for all
$s$:}
Because for any node $v$ and $s>1$, $C_s^{(t^+,t,v)}\leq \sigma
C_s^{(t^+,t)}=\bar{C}_s$, the minimality of $\bar{t}$ and $\bar{s}$ imply that
for all $t\in [\underline{t}-\barstabint{s},\bar{t}]$, $1<s< \bar{s}$,
and $v\in V$, the system is $(\bar{C},s)$-legal at node $v$ and time $t$. In
particular, (S1) is satisfied at time $t$ at all nodes $v\in V_u(t)$ w.r.t.\
$\bar{C}$ and $\bar{s}$.

\paragraph{Case \boldmath$\bar{C}_s=2\hat{\G}/\sigma^{s-1}$ for all
$s$:}
Consider $1<s<\bar{s}$, a time $t\in [\underline{t}-\barstabint{s},\bar{t}]$,
and a node $v\in V_u(t)$. We need to show that $C_s^{(t^+,t,v)}=
C_s^{(t^+,t)}=\bar{C}_s$. Assume for the sake of contradiction that there is a
minimal level $1<s<\bar{s}$ violating this claim for some $t$ and $v\in
V_u(t)$. We apply Lemma~\ref{lemma:clocks} to nodes $u$ and $v$ with respect
to $(\bar{C},s-1)$-legality, where $\kappa_p\leq \bar{C}_{\bar{s}-1}\leq
\bar{C}_{s-1}$. This yields
\begin{equation*}
|L_v(t)-L_u(t)|<\left(s+\frac{1}{2}\right)\kappa_p+\frac{\bar{C}_{s-1}}{2}
\leq s\bar{C}_{s-1},
\end{equation*}
and thus
\begin{eqnarray}
|L_v(t)-L_u(\bar{t})|&< & s\bar{C}_{s-1}
+(1+\rho)(1+\mu)(\bar{t}-(\underline{t}-\bar{\Theta}_s))\nonumber\\
&=& s\bar{C}_{s-1} +\frac{5(1+\mu)\bar{C}_{\bar{s}-1}}{2\mu(1-\rho)}
+\frac{(1+\rho)(1+\mu)\bar{C}_{s-1}}{\mu(1-\rho)}\nonumber\\
&\sr{eq:beta}{\leq} & \beta_s(\Gu^+).\label{eq:close}
\end{eqnarray}

Moreover, the fact that $C_s^{(t^+,t,v)}\neq C_s^{(t^+,t)}$ entails that
\begin{equation*}
t\in \bigcup_{s'\leq s}U_{s'}(v,t^+)\cup R_{s'}(v,t^+)
\end{equation*}
Therefore, there exist $s'\leq s$ and $T_{s'}\in\Tset_{s'}$ so that
\begin{equation*}
|L_v(t)-T_{s'}|\leq A_{s'}(\Gu^+)+B_{s',s}(\Gu^+).
\end{equation*}
Similarly, as we have that
\begin{equation}\label{eq:t_bar_range}
\frac{2P(t^+)+t^+}{3}\sr{eq:Delta_2}{\leq}t^+-\Delta_2 \leq \bar{t}
\leq t^+
\end{equation}
and $C_{\bar{s}}^{(t^+,\bar{t},u)}=
2\Gu^+/\sigma^{\bar{s}-1}= C_{\bar{s}}^{(t^+,\bar{t})}$, it holds that
\begin{equation*}
\bar{t}\notin \bigcup_{s''\leq \bar{s}}U_{s''}(v,t^+)\cup R_{s''}(v,t^+)
\end{equation*}
and hence
\begin{equation*}
|L_u(\bar{t})-T_{s'}|> A_{s'}(\Gu^+)+B_{s',\bar{s}}(\Gu^+).
\end{equation*}
Combining these two inequalities yields that
\begin{equation*}
|L_u(\bar{t})-L_v(t)|>B_{s',\bar{s}}(\Gu^+)-B_{s',s}(\Gu^+)=
B_{s,\bar{s}}(\Gu^+)\geq \beta_s(\Gu^+),
\end{equation*}
contradicting \eqref{eq:close}.

It remains to show (S2) for each $v\in V_u(t)$ and all $t\in
[\underline{t},\bar{t}]$. Since we already established (S1), we can apply
Lemma~\ref{lemma:clocks} to nodes $u$ and $v$ with respect to
$(\bar{C}_{\bar{s}-1},\bar{s}-1)$-legality, where $\kappa_p\leq
\bar{C}_{\bar{s}-1}$. This yields
\begin{equation*}
|L_v(t)-L_u(t)|<\left(\bar{s}-\frac{1}{2}\right)\kappa_p+\frac{\bar{C}_{\bar{s}-1}}{2}
\leq \bar{s}\bar{C}_{\bar{s}-1}.
\end{equation*}
We claim that $\bar{t}\notin U_{\bar{s}}(u,t^+)$. For contradition,
assume that $\bar{t}\in U_{\bar{s}}(u,t^+)$. We then have
$C_{\bar{s}}^{(t^+,\bar{t},u)} = C_{\bar{s}-1}^{(t^+,\bar{t})}$. In
order to have
$C_{\bar{s}}^{(t^+,\bar{t},u)}\neq C_{\bar{s}-1}^{(t^+,\bar{t},u)}$,
we thus need that
$C_{\bar{s}-1}^{(t^+,\bar{t},u)} = C_{\bar{s}-2}^{(t^+,\bar{t})}$,
which implies that
$\bar{t}\in \bigcup_{s'\leq \bar{s}-1} U_{s'}(u,\bar{t})\cup
R_{s'}(u,\bar{t})$,
which is a contradiction to the statement of Lemma
\ref{lemma:instab}. We can therefore conclude that
$\bar{t}\notin U_{\bar{s}}(u,t^+)$.  Together with
\eqref{eq:t_bar_range}, this entails for any
$T_{\bar{s}}\in \mathbb{T}_{\bar{s}}$ that
$|L_u(\bar{t})-T_{\bar{s}}|\geq A_{\bar{s}}(\Gu^+)$ and thus
\begin{eqnarray*}
|L_v(t)-T_{\bar{s}}|&\geq &
|L_u(\bar{t})-T_{\bar{s}}|-|L_u(\bar{t})-L_u(t)|-|L_u(t)-L_v(t)|\\
&\geq & A_{\bar{s}}(\Gu^+)-(1+\rho)(1+\mu)(\bar{t}-\underline{t})
-\bar{s}\bar{C}_{\bar{s}-1}\\
&= & A_{\bar{s}}(\Gu^+)
-\frac{5(1+\mu)\bar{C}_{\bar{s}-1}}{2\mu(1-\rho)}
-\bar{s}\bar{C}_{\bar{s}-1}\\
&\sr{eq:A}{=} &
\left(\frac{(1+\rho)(1+\mu)}{(1-\rho)\mu}+\bar{s}\right)\bar{C}_{\bar{s}-1}\\
&= & (1+\rho)(1+\mu)\bar{\Theta}_{\bar{s}}+\bar{s}\bar{C}_{\bar{s}-1}.
\end{eqnarray*}
We conclude that the preconditions for the application of
Theorem~\ref{thm:catch_up} described earlier are met, yielding the stated
contradiction to the original assumption that the claim of the lemma is wrong.
\end{proof}

\begin{theorem}\label{theorem:gradient}
At all times $t\geq t_{\min}$ and nodes $u\in V$, the system is $C^{(t)}$-legal,
where $C_s^{(t)}=4\G(P(t))/\sigma^{\max\{s-2,0\}}$.
\end{theorem}
\begin{proof}
Consider any execution of the algorithm and fix a time $t\geq t_{\min}$. We
create a $t$-insertion-bounded execution that is identical on $[0,t)$ as
follows. We modify the given execution in that at time $t$ all edges
fail, i.e., $E(t')=\emptyset$ for all $t'\geq t$. Moreover, all
nodes become aware of the non-existence of their incident edges at time $t$.
Hence, all nodes clear their neighbor sets at time $t$ and, for all $s\in \nat$,
$\Tset_s\cap (t,\infty)=\emptyset$. Therefore, the resulting execution is
$t$-insertion-bounded and identical to the original one on $[0,t)$.

We can apply Lemma~\ref{lemma:proto_bound} to the new execution, showing that
at time $t$, each node $u$ is $C^{(t,t,u)}$-legal. As $C^{(t,t,u)}_s\geq
C^{(t,t)}_{s-1}=C_s^{(t)}$, each node is $C^{(t)}$-legal at time $t$ in the
modified execution. Since the modified execution is identical to the original
execution during $[0,t)$ and logical clocks are continuous, the claim of the
theorem follows.
\end{proof}

\begin{corollary}[Gradient Property]\label{coro:gradient}
For $t\geq t_{\min}$, set $E^{\infty}(t):=\cap_{s=1}^{\infty}E^s(t)$ and let
$p$ be a path connecting $u$ and $v$ in $(V,E^{\infty})$ of
minimal weight $\kappa_p$. For
\begin{equation*}
s(p):=\max\{2+\lceil\log_{\sigma}(4\G(P(t))/\kappa_p\rceil,1\},
\end{equation*}
it holds that
\begin{equation*}
|L_u(t)-L_v(t)|\leq (s(p)+1)\kappa_p.
\end{equation*}
\end{corollary}
\begin{proof}
Since $p$ is a path in $(V,E^{\infty})$ and $E^{s(p)}\subseteq E^{\infty}$,
$p\in P_u^{s(p)}(t)$. By Theorem~\ref{theorem:gradient}, the system is
$C^{(t)}$-legal at $u$ and $v$ at time $t$. Applying Lemma~\ref{lemma:clocks}
for level $s(p)$, we obtain that
\begin{equation*}
|L_u(t)-L_v(t)|\leq
\left(s(p)+\frac{1}{2}\right)\kappa_p+\frac{C^{(t)}_{s(p)}}{2}
=\left(s(p)+\frac{1}{2}\right)\kappa_p+\frac{2\G(P(t))}{\sigma^{\max\{s(p)-2,0\}}}
\leq (s(p)+1)\kappa_p,
\end{equation*}
where the last inequality holds because $s(p)=1$ implies that
$\kappa_p>4\G(P(t))$.
\end{proof}

%%% Local Variables:
%%% mode: latex
%%% TeX-master: "main"
%%% End:

\section{Lower Bound on the Insertion Time}
\label{sec:lower}

In this section, we strengthen the lower bound in~\cite{kuhn09} to
match the stabilization time of $\Aopt$. The original lower bound
stated, roughly speaking, that the stabilization time of any
$\localskew$-dynamic gradient CSA with a stable gradient skew of
$\stableskew$ cannot be better than $\Omega(D / \stableskew(1))$ in
graphs of diameter $D$. For CSAs with $\BO(\log_{1 / \drift} D)$-local
skew, this bound implies that the stabilization time must be $\Omega(D
/ \log_{1 / \drift} D)$. Algorithm $\Aopt$ has a stabilization time of
$\BO(D)$, which does not match the bound in~\cite{kuhn09}; however, by
refining the analysis in the lower bound, we can show that the
algorithm is in fact asymptotically optimal in its stabilization time.
In the stronger bound we reason about the \emph{full} gradient
property, which bounds the skew on paths of all distances, rather than
just the local skew property, which bounds the skew on single edges.

Let us call a dynamic gradient CSA \emph{non-trivial} if it has a
stable gradient skew satisfying $\stableskew(1) \in o(D)$. This
essentially means that the algorithm guarantees a \emph{local} skew
(e.g., along single edges) that is better than the global skew.

\begin{theorem}\label{thm:lower}
  Let $\mathcal{F} = \set{ f_D : \posreals \rightarrow \posreals \st D
    \in \reals}$ be a family of functions,
  and let $c_1,c_2 \in (0, 1/16)$ be constants such that for all $f_D
  \in \mathcal{F}$ we have $f_D(c_1 D) \leq c_2 D$. Let $\mathcal{A}$
  be a non-trivial stabilizing CSA guaranteeing a dynamic gradient
  skew of $f_D$ in graphs of weighted diameter $D$. Then the
  stabilization time of $\mathcal{A}$ is at least $\Omega(D)$.
\end{theorem}
\begin{proof}[Proof Sketch]
  We show that for sufficiently large diameters $D$, we can add a new
  edge and cause the skew on it to be larger than $\mathcal{S}$ after
  $\Omega(D)$ time. For simplicity we consider only line networks,
  where $D \in \Theta(n)$, but the proof can easily be modified to
  hold in general networks.

  Consider a static line graph over $n+1$ nodes $v_0, \ldots, v_{n}$,
  where the estimate graph is the same as the communication graph and
  the weights of all edges are $T$. The diameter of the graph is $D =
  nT$. Let $c_1, c_2$ be the constants from the statement of the
  theorem, and let $u \coloneq v_{\lceil c_1 n \rceil}, v \coloneq
  v_{\lfloor n - c_1 n \rfloor}$. Finally, let $t_s \geq \stabtime$ be
  some time after the stabilization time of the algorithm. By
  definition of $\stabtime$, at any time after $t_s$, the skew on any
  path of weight $d$ cannot exceed $2f_{D}(d)$.
	
  The distance between $v_0$ and $u$ and between $v$ and $v_n$ is at
  least $c_1 n$; thus, for all $t \geq t_s$ we have
  \begin{eqnarray*}
    \lc{v_0}(t) - \lc{u}(t) &\leq& f_{D}(c_1 n) \leq c_2 n,\\
    \lc{v}(t) - \lc{v_n}(t) &\leq& f_{D}(c_1 n) \leq c_2 n.
    \label{eq:diff}
  \end{eqnarray*}
  Also, $\dist(u,v) \geq n - c_1 n - 1 - (c_1 n + 1) = n - 2c_1
  n - 2$.
  
  In~\cite{kuhn09}, we show that we can create an execution $E$ in which
  \begin{enumerate}[(a)]
  \item There exists a time $t_2 \geq t_s$ such that
    $\lc{u}(t_2) - \lc{v}(t_2) \geq \frac{1}{4} \dist(u,v) \geq n - 2c_1
    n - 2$, and
  \item The message delays on all edges between $v_0$ and $u$ and
    between $v$ and $v_n$ are always at least $T/(1 + \drift)$.
  \end{enumerate}
  Next we create a new execution $E'$, which is identical to $E$ until
  time $t_1 \coloneq t_2 - c_1 n \cdot T / (1 + \drift)$. At time $t_1$
  in $E'$, an edge between $v_0$ and $v_n$ appears. Our goal is to
  maintain a large skew on $\set{v_0, v_n}$ until time $t_2$ to show
  that the algorithm has not stabilized by then.
  
  Due to the large message delays, nodes $u,v$ do not find out about the
  new edge until time $t_2$. Consequently, their skew in execution $E'$
  is the same as in $E$. %, at least $n - 2c_1 n - 2$.
  The paths $v_0, v_1, \ldots, u$ and $v, \ldots, v_{n-1}, v_n$, which
  have weight at least $c_1 n$ by definition, are stable in both $E$ and
  $E'$. Thus, the skew on each path cannot exceed $2f_{D}(c_1 n)$, that
  is, it cannot exceed $2 c_2 n$. It follows that the skew between $v_0$
  and $v_n$ at time $t_2$ in $E'$ is at least 
  \begin{eqnarray*}
    \lc{v_0}(t_2) - \lc{v_n}(t_2) 
    &=& 
    \lc{v_0}(t_2) - \lc{u}(t_2) + \lc{v}(t_2) - \lc{v_n}(t_2) + \lc{u}(t_2) - \lc{v}(t_2) \\
    & \geq & 
    n - 2c_1 n -2 - 4c_2 n > n/2 - 2. 
  \end{eqnarray*} 
  For sufficiently large $n$, this value exceeds $\stableskew(1)$, since
  we assumed that $\stableskew(1) \in o(D)$. Thus, we have showed that
  after $c_1 n \cdot T / (1 + \drift) \in \Omega(D)$ time since edge
  $\set{v_0, v_n}$ appeared, the algorithm has not yet stabilized.
\end{proof}

\bibliographystyle{abbrv}
\bibliography{journal,podc10}

\end{document}